\documentclass[ecta,nameyear, publication]{econsocart2}
\RequirePackage[colorlinks,citecolor=blue,linkcolor=blue,urlcolor=blue,pagebackref]{hyperref}
\usepackage{booktabs, graphicx, subcaption}

\startlocaldefs

\theoremstyle{plain}
\newtheorem{asm}{Assumption}[section]
\newtheorem{corollary}{Corollary}
\newtheorem{lemma}{Lemma}
\newtheorem{thm}{\protect\theoremname}
\newtheorem{proposition}{Proposition}

\theoremstyle{remark}
\newtheorem{remark}{Remark}[section]
\newtheorem{example}{Example}[section]

\usepackage{mycommands}
\graphicspath{{figures/}}

\newcommand{\beginsupplement}{%
        \setcounter{section}{0}  
        \renewcommand{\thesection}{S\arabic{section}}
        \setcounter{subsection}{0}  
        \renewcommand{\thesubsection}{S\arabic{section}.\arabic{subsection}}
        \setcounter{table}{0}  
        \renewcommand{\thetable}{S\arabic{table}}
        \setcounter{figure}{0}  
        \renewcommand{\thefigure}{S\arabic{figure}}
        \setcounter{equation}{0}  
        \renewcommand{\theequation}{S\arabic{equation}}
     }

\endlocaldefs

\begin{document}

\begin{frontmatter}

\title{An Adversarial Approach to Identification}
\runtitle{An Adversarial Approach to Identification}

\begin{aug}
%
%
\author[id=au1,addressref={add1}]{\fnms{Irene}~\snm{Botosaru}\ead[label=e1]{botosari@mcmaster.ca}}
\author[id=au2,addressref={add2}]{\fnms{Isaac}~\snm{Loh}\ead[label=e2]{lohi@uncw.edu}}
\author[id=au3,addressref={add1}]{\fnms{Chris}~\snm{Muris}\ead[label=e3]{muerisc@mcmaster.ca}}
\address[id=add1]{%
\orgdiv{Department of Economics},
\orgname{McMaster University}}

\address[id=add2]{%
\orgdiv{Department of Economics},
\orgname{University of North Carolina Wilmington}}

\end{aug}

\support{We thank Bo Honor\'e, Jiaying Gu, Hide Ichimura, and Jim Powell for discussions and suggestions. Botosaru gratefully acknowledges financial support from the Canada Research Chairs Program. This paper was presented at the University of Arizona in March 2024 and the Southern Economic Association in November 2024.}
\begin{abstract}

We introduce a new framework for characterizing identified sets of structural and counterfactual parameters in econometric models. By reformulating the identification problem as a set membership question, we leverage the separating hyperplane theorem in the space of observed probability measures to characterize the identified set through the zeros of a discrepancy function with an adversarial game interpretation. The set can be a singleton, resulting in point identification. A feature of many econometric models, with or without distributional assumptions on the error terms, is that the probability measure of observed variables can be expressed as a linear transformation of the probability measure of latent variables. This structure provides a unifying framework and facilitates computation and inference via linear programming. We demonstrate the versatility of our approach by applying it to nonlinear panel models with fixed effects, with parametric and nonparametric error distributions, and across various exogeneity restrictions, including strict and sequential.
\end{abstract}


\end{frontmatter}

\section{Introduction}
\label{sec:intro}

Identification of structural and counterfactual parameters is a central challenge in econometric models with unobserved heterogeneity. In many cases, the distribution of unobserved heterogeneity is not point identified, leading to partial identification of structural and/or counterfactual parameters. This issue is pervasive in nonlinear panel models with fixed effects. 
Fixed effects obstruct the point identification of \textit{both} structural parameters and partial effects in all but a narrow class of models \citep{ArellanoBonhomme2012}.
This highlights the need for methods that achieve sharp identification for both types of parameters, while remaining computationally feasible and enabling valid inference. 

Addressing these issues in nonlinear panel models is challenging. Many existing approaches are tailored to specific model features, relying on parametric assumptions about error distributions or support and exogeneity restrictions on the covariates. Additionally, existing methods focus on either structural or counterfactual parameters. This has led to a fragmented literature with various methods addressing only isolated aspects of the broader identification problem.

We introduce a novel and unifying framework for characterizing identified sets of structural and counterfactual parameters in econometric models with unobserved heterogeneity. Departing from existing approaches, we reformulate the identification problem as a set membership question in the space of \textit{observed} probability measures (i.e., probability measures of observed random variables). This reformulation allows us to leverage the separating hyperplane theorem in the infinite-dimensional space of observed probability measures, providing a characterization of the identified set via the zeros of a new discrepancy function. The discrepancy function checks whether there exists at least one hyperplane that separates the observed probability measure from the set of model probabilities, defined as the collection of all probability measures of observed random variables consistent with a given parameter value. By aggregating over all hyperplanes and all measures in the set of model probabilities, the discrepancy function reveals whether the observed measure belongs to the set of model probabilities. The discrepancy function has a maximin representation or adversarial game interpretation, inspiring the name of our approach.\footnote{Our approach is distinct from the simulation-based adversarial \textit{estimation} method of \citet{Kaji2023}. Ours is an identification approach.}

We establish new sufficient and necessary conditions for sharp identification: when the set of model probabilities is convex, our approach obtains a sharp characterization of the identified set; when convexity fails, our approach characterizes an outer set. When the identified set is a singleton, we obtain point identification. Remarkably, many econometric models naturally feature a convex set of model probabilities.

A key innovation of our framework is its ability to accommodate both parametric and nonparametric error distributions, along with a wide range of exogeneity restrictions and conditioning variables --- whether continuous or discrete, strictly exogenous or predetermined. To demonstrate the power of our approach, we characterize the identified set for structural and counterfactual parameters in the semiparametric binary choice panel model under sequential exogeneity (or ``predeterminedness'') and without parametric assumptions on the distribution of the error terms, addressing a long-standing gap in the literature on nonlinear panel models. We further highlight the versatility of our method by applying it to the binary choice panel model with error terms that follow a fixed but arbitrary distribution. For the nested case of logistic errors, we recover established results for the structural parameter and recently derived results for counterfactual parameters. These applications illustrate the versatility of our framework and its potential to advance partial identification in nonlinear econometric models.

Let $Z$ denote the observed random variables, and \(\mu^*_Z\) the observed probability measure characterizing the distribution of \(Z\). Let $\theta\in\Theta$ denote the parameter of interest, which can include both structural parameters and arbitrary functionals of the latent probability measure, and  $\overline{\mathcal{M}}_\theta$ denote the closure of the set of model probabilities.\footnote{\label{takingtheclosure} Explicitly taking the closure ensures the inclusion of limit measures, such as those arising from degenerate distributions when latent random variables reach extreme values (e.g., fixed effects approaching \(\pm \infty\)).
} We reformulate the identification problem as a set membership question: For a candidate parameter value \(\theta\), does \(\mu^*_Z\) belong to \(\overline{\mathcal{M}}_\theta\)? Accordingly, the identified set for \(\theta\) is given by:
\begin{align}
\label{E:essentialIDset}
    \Theta_{\mathrm{I}} = \{\theta \in \Theta: \mu^*_Z \in \overline{\mathcal{M}}_{\theta} \}.
\end{align}
This set contains all parameter values compatible with \(\mu^*_Z\). If \(\mu^*_Z \in \overline{\mathcal{M}}_\theta\), the observed probability measure and the model probabilities are indistinguishable at \(\theta\); otherwise, \(\theta\) is incompatible with \(\mu^*_Z\). 

To determine membership of $\mu^*_Z$ in $\overline{\mathcal{M}}_\theta$, we leverage the separating hyperplane theorem. Specifically, when $\mu_Z^*\notin\overline{\mathcal{M}}_\theta$ and $\mathcal{M}_\theta$ is convex, there exists a hyperplane $\phi$ that separates $\mu_Z^*$ from $\mathcal{M}_\theta$. Aggregating across all hyperplanes $\phi$ leads to the following discrepancy function:
\begin{align}
    T(\theta) 
    \equiv 
    \sup_{\phi \in \Phi_b(\mathcal Z)} 
        \inf_{\mu \in \mathcal{M}_\theta} 
        \left(
            \EE{\mu_Z^*}{\phi} - \EE{\mu}{\phi}
        \right),
    \label{def:T_theta}
\end{align}
where $\mathcal{Z}$ denotes the support of $Z$ and $\Phi_b(\mathcal{Z})$ a set of bounded functions supported on $\mathcal{Z}$, defined in Section \ref{sec:mainIDresult}.

The discrepancy function $T(\theta)$ evaluates to zero if and only if no separating hyperplane $\phi$ exists. Consequently, the zeros of $T(\theta)$ can be used to characterize $\Theta_\mathrm{I}$. This characterization is sharp provided that $\mathcal{M}_\theta$ is convex --- a property shared by many econometric models. When  $\mathcal{M}_\theta$ is not convex, the zeros of $T(\theta)$ describe an outer set. When $\Theta_\mathrm{I}$ is a singleton, our framework obtains point identification of $\theta$. 

In many econometric models, the discrepancy function admits a low-dimensional representation through an extreme point characterization. This facilitates the computation of zeros of $T(\theta)$ using linear programming. Two features of many models are central to this result: (i) the observed probability measure can be expressed as a linear transformation of the latent probability measure (i.e., the measure of latent random variables),\footnote{In semiparametric models with unrestricted error distributions, this transformation corresponds to a pushforward measure, whereas in models with parametric error distributions, it takes the form of a linear operator.} and (ii) the latent probability measures lie within a convex set.

For example, the semiparametric binary choice panel model under strict exogeneity exhibits both features, allowing for efficient computation of the identified set for the structural and counterfactual parameters using linear programming. Conversely, the semiparametric binary choice panel model under sequential exogeneity satisfies the linearity condition but not the convexity of the set of latent probabilities; instead, this model features a convex set of model probabilities $\mathcal M_\theta$, which yields sharp identification. However, computing the identified set in this case requires an extension of our linear programming approach. Beyond these computational advantages, the two features enable the sample analog of the discrepancy function to serve as a test statistic, facilitating valid inference. 
We establish the asymptotic distribution of this test statistic and construct critical values for hypothesis tests that are uniformly valid across the underlying probability distributions and parameter values.

A few key features distinguish our proposed method. The first distinguishing feature underscores the simplicity of our method in establishing sharp identification. A sufficient condition for sharpness is the convexity of $\mathcal M_\theta$. This can be established directly, or, given the linearity of the transformation, via the convexity of the set of latent probability measures. The latter holds in two key cases: (i) when no assumptions are imposed on the latent probability measure, such as in panel models with no distributional assumptions on the error terms, and (ii) when the latent probability measure is required to satisfy a finite number of linear restrictions. Such linear restrictions typically arise in two contexts: (i) when analyzing counterfactual parameters, many of which can be expressed as linear functionals of the latent probability measure,\footnote{See, e.g., \citet{ChristensenConnault2023, torgovitsky2019}.} and (ii) when imposing exogeneity conditions, such as zero-mean or zero-median constraints on the error terms. We show this in our examples.

The second distinguishing feature is the versatility and broad applicability of our approach. Our method applies to many econometric models with unobserved heterogeneity, regardless of whether the models involve error terms that follow parametric distributions, or outcomes and covariates that are discrete or continuous. Our method accommodates various exogeneity restrictions on the covariates, and, if desired, restrictions on the latent probability measures. 

The third distinguishing feature of our method is its computational ease. We implement our procedure and examine the size of the identified set in the canonical semiparametric binary choice model without parametric restrictions on the error terms, both in the cross-sectional case and in the panel case with strictly exogenous  regressors, time effects, and fixed effects.
With predetermined regressors, although the set of latent measures is not convex, the set of model probabilities is, so the identified set can be computed via an extension of our linear programming method. These illustrations to the semiparametric binary choice panel model with strictly exogenous or predetermined covariates contribute new results to the nonlinear panel literature.

The fourth distinguishing feature is that the sample analog of the discrepancy function can be used as a test statistic for valid inference.

\subsection{Related literature}
\label{relatedlit}

This paper contributes to the literature on sharp identification in general classes of models and to the literature on nonlinear panel models.

A wide range of methods has been developed to characterize identified sets across different models.\footnote{For reviews, see \citet{BontempsMagnac2017,CanayShaikh2017,Molinari2020,ChesherRosen2020,KlineTamer2023}.} Approaches include the use of random set theory \citet{beresteanuSharpIdentificationRegions2011, ChesherRosen2017}, optimal transport \citet{GalichonHenry2009,EkelandGalichonHenry2010,GalichonHenry2011}, and information-theoretic methods \citet{schennachEntropicLatentVariable2014}. Other contributions, such as \citet{torgovitsky2019}, focus on extending subdistributions, while more recent work explores minimum relevant partition and latent space enumeration \citet{Tebaldi2023,GuRussellStringham2022}. Some of these methods focus exclusively on complete models, while others are explicitly designed to also address incomplete models.\footnote{In particular, models where the relationship between the observed random variables and the latent random variables is a correspondence, e.g., \citet{Tamer2003}.} Many of these methods, like ours, leverage convex analysis for sharp identification or low-dimensional representations. 

Our framework differs by reformulating the identification problem as a set membership question in the space of observed probability measures, treating \(\mathcal{M}_\theta\) as primitive. This allows us to (i) apply the separating hyperplane theorem in the space of observed probability measures, (ii) link sharpness to the convexity of \(\mathcal{M}_\theta\) (with outer sets characterized when convexity fails), and (iii) exploit two features common to many econometric models: the linearity of the transformation between observed and latent probability measures and the convexity of the set of latent probability measures. These properties facilitate the computation of the identified set via linear programming. This structure enables a unified approach for structural and counterfactual parameters, accommodating both parametric and nonparametric error distributions. This versatility is noteworthy. For example, methods designed for nonparametric error distributions rarely handle parametric restrictions (e.g., \citet{guDualApproachWassersteinRobust2023, ChesherRosenZhang}), while parametric approaches often rely on specific distributional assumptions (e.g., \citet{Bonhomme2012,daveziesFixedEffectsBinary2023,GuRussellStringham2022}). 

While we focus on models where the relationship between observed and latent random variables is a mapping rather than a correspondence, our results can be extended to models involving correspondences. This is possible since our approach operates in the space of observed probability measures, and a sufficient assumption for \textit{all} our results is that the observed probability measure is a linear map of the latent probability measure. Correspondences between random variables can induce such linear maps. However, as we show in the paper, such linear maps often arise in models commonly used in the literature on nonlinear panel models. Since our methodology is tailored to address a long-standing gap in the literature on nonlinear panel models, we leave a detailed investigation of correspondence-defined models for future research.

Identification challenges have long been a hallmark of nonlinear panel models due to the presence of fixed effects, see, e.g., \citet{arellanoPanelDataModels2001},  \citet{honoreBoundsParametersPanel2006}. As \citet{ArellanoBonhomme2012} note, point identification of structural parameters is rare, and even then, functionals of the fixed effects distribution, such as partial effects, remain only partially identified.\footnote{For exceptions of point-identified counterfactual parameters in specific models, see \citet{honoreMarginalEffectsSemiparametric2008},  \citet{aguirregabiriaIdentificationAverageMarginal2024}, and \citet{danoTransitionProbabilitiesIdentifying2023}.} 
Addressing dynamic exogeneity in nonlinear models remains an open challenge, especially when both structural \textit{and} counterfactual parameters are of interest, see \citet{HonoredePaula2021}, \citet{ArkhangelskyImbens}.

We showcase adversarial identification by applying it to the semiparametric binary choice panel model (see \cite{Manski1987}). We obtain new results for the identified set for both structural parameters and the average structural function (ASF), without imposing parametric restrictions on the distribution of the error terms and under various dynamic exogeneity assumptions, such as strict and sequential exogeneity. 

Most work on the semiparametric binary choice panel model focuses on the identification of structural parameters; for a non-exhaustive list, see \citet{Manski1987},
\citet{aristodemouSemiparametricIdentificationPanel2021},
\citet{khanIdentificationDynamicBinary2023},
\citet{botosaruIdentificationTimevaryingTransformation2021},
\citet{mbakopIdentificationDiscreteChoice2023}, 
\citet{gaoIdentificationNonlinearDynamic2024}, 
\citet{ChesherRosenZhang}. 
The identification of partial effects has received less attention: \citet{ChernValHahnNewey2013} derive results under time-homogeneity and restrictions on the support of the covariates; \citet{BotosaruMuris2024} relax the latter assumptions while assuming that the structural parameters are a priori either point- or partially-identified. 

When the error distribution is fixed but arbitrary, the structural parameters can be point-identified in a narrow class of models, see \citet{Bonhomme2012}. Recent work focuses on partial identification of counterfactual parameters starting from point-identified or $\sqrt{n}$-estimable structural parameters, see  \citet{dobronyiIdentificationDynamicPanel2021,daveziesFixedEffectsBinary2023,pakelBoundsAverageEffects2024}. Seminal contributions by  \citet{honoreBoundsParametersPanel2006} and \citet{ChernValHahnNewey2013} treat partial identification of structural parameters and average treatment effects, relying on discrete outcomes and covariates, and support conditions on the fixed effects. 

Even with parametric restrictions on the error terms, identification of both structural and counterfactual parameters under sequential exogeneity is challenging, cf. \citet{arellanoPanelDataModels2001}, \citet{ArellanoBonhomme2012}, and
\citet{HonoredePaula2021}. Results for the identification of structural parameters under sequential exogeneity and error terms that follow parametric distributions can be found in \citet{arellanoBinaryChoicePanel2003}, \citet{piginiConditionalInferenceBinary2022},    \citet{chamberlainIdentificationDynamicBinary2023}, and for partial effects in \citet{bonhommeIdentificationBinaryChoice2023}. 

In contrast, our approach obtains the identified set of structural and counterfactual parameters for models with known or unknown error distributions, does not require functional form restrictions such as an index structure or additivity in the unobservables, time-homogeneity, or support restrictions on the covariates or on the fixed effects. This flexibility extends to models with various forms of dynamic exogeneity --- strict, sequential, or lagged, offering new insights for nonlinear panel models. To the best of our knowledge, ours are the only identification results in nonlinear panel models with sequential exogeneity and without parametric restrictions on the error terms. 

\subsubsection*{Organization.} 
We introduce our main result in Section \ref{sec:mainIDresult}. This result applies to a very general class of models. We specialize the result to semiparametric models with unobserved heterogeneity in Section \ref{sec:semiparametric_pushforward} and with observed covariates in Section \ref{sec:convexityinmodels}. We also illustrate our approach on the semiparametric binary choice model with a discrete regressor in Section \ref{sec:example_manski}. Section \ref{sec:adversarial} shows that the discrepancy function and identified set can be computed via linear programming.
Section \ref{sec:nonlinearpanels} extends all results to models with parametric error terms: Section \ref{sec:panel_ID_theta} focuses on the identification of common parameters, Section \ref{sec:panel_ID_counterfactuals} on the identification of both structural and counterfactual parameters, and Section \ref{linprog_parametric} discusses computation via linear programming. Section \ref{sec:application_to_panels} applies the results to panel models. Appendix \ref{sec:additional_remark} contains additional remarks, including inference in Section \ref{sec:inference}, while Appendix \ref{sec:proofs} contains proofs not contained in the main text.

\subsubsection*{Notation.}

For a Polish space $\mathcal{S}$ endowed with its Borel $\sigma$-algebra, $\mathfrak{B}(\mathcal{S})$ denotes the set of all Borel measures on the set $\mathcal{S}$, $\mathcal{P}(\mathcal{S}) \subseteq \mathfrak{B}(\mathcal{S})$ denotes the set of all Borel probability measures supported on $\mathcal{S}$, and $\delta_{s}$ denotes the Dirac measure at $s\in\mathcal{S}$. For an index $\theta$, $
\Gamma_\theta(\mathcal{S})\subseteq \mathcal{P}(\mathcal{S})$ denotes a generic set of Borel probability measures supported on $\mathcal{S}$. The set $C_c(\mathcal{S})$ denotes the space of compactly supported continuous functions defined on $\mathcal{S}$. The product of two or more Polish spaces is endowed with the product topology. 

For an arbitrary convex set $\mathcal{C}$ in a linear space, $\mathrm{ext}( \mathcal{C})$ denotes the set of all extreme points of $\mathcal{C}$, $\mathrm{co}(\mathcal{C})$ denotes the convex hull of $\mathcal{C}$, and $\overline{\text{co}}(\mathcal{C})$ denotes the smallest closed convex set containing the set $\mathcal{C}$.

We use $\subset$ to denote a strict subset, i.e. $A \subset B$ means that $x\in A \Rightarrow x \in B$ and $\exists b \in B: b \not \in A$. We use $A \subseteq B$ to denote that $A \subset B$ or $A = B$.

For arbitrary Borel measures $\mu, \mu'$, the total variation norm between $\mu$ and $\mu'$ is denoted by $\norm{\mu -\mu'}_{\mathrm{TV}} \equiv \sup_{B \text{ Borel}} |\mu(B) - \mu'(B)|$. 
For a given arbitrary measure $\mu$ and an arbitrary vector-valued function $f$, $f \in L^1(\mu)$ if each component function of $f$ is integrable with respect to $\mu$. We let $\EE{\mu}{\cdot}$ denote integration against a general probability measure $\mu$.

We denote by $\Phi(\cdot)$ the cumulative distribution function of the standard normal distribution, and by $\Lambda(\cdot)$ the cumulative distribution function of the logistic distribution.

\section{Main result}
\label{sec:mainIDresult}

Let $Z$ denote an observable Borel measurable random variable supported on a space $\mathcal{Z}$, and let $\mu^*_{Z}$ denote the true probability measure of $Z$. Denote by $\Theta$ the set of possible values of the parameter of interest $\theta$, and by \(\Gamma_\theta\) a set of auxiliary parameters $\gamma$ that may vary with $\theta$. The parameter of interest may include functions of $\gamma\in\Gamma_\theta$. For example, in the binary choice model of Section \ref{sec:example_manski}, $\gamma$ is the unknown distribution of an unobservable error term and $\theta$ includes counterfactual choice probabilities, which are functionals of $\gamma$. 

For each $\theta \in \Theta$ and $\gamma \in \Gamma_\theta$, the econometric model for $Z$ specifies a probability measure $\mu_{Z,(\theta,\gamma)}$, which we call the \textit{model probability}. For fixed $\theta$, define the set of model probabilities as:
\begin{align} 
\label{eq:Mtheta}
    \mathcal{M}_\theta &\equiv \{\mu_{Z,(\theta, \gamma)}: \gamma \in \Gamma_\theta\}.
\end{align}
This set is the collection of all probability measures of $Z$ that are consistent with the econometric model under parameter value $\theta$. The set is a fundamental object for our analysis, and its geometric properties are essential for our main result in Theorem \ref{P:main} below. 

The identified set for $\theta$ is defined as the set of parameter values compatible with the probability measure $\mu_Z^*$, formally defined in \eqref{E:essentialIDset}, where $\overline{\mathcal{M}}_{\theta}$ is the closure of $\mathcal{M}_\theta$  with respect to the topology induced by the total variation (TV) norm:
\begin{align}
\label{eq:closureMtheta}
    \overline{\mathcal{M}}_{\theta} &\equiv \{m \in \mathcal P(\mathcal Z): \forall \varepsilon >0, \exists \gamma \in \Gamma_\theta \text{ such that } d_{\mathrm{TV}}(m, \mu_{Z,(\theta, \gamma)}) <\varepsilon \}.
\end{align}

Defining the identified set through the closure \(\overline{\mathcal{M}}_{\theta}\) offers two benefits. First, it explicitly includes limit measures, which may arise when, e.g., $\gamma$ converges to a degenerate measure.\footnote{Alternatively, such measures could be directly incorporated into $\mathcal M_\theta$.}
Second, it ensures that the identified set includes all parameter values $\theta$ for which the model measures are indistinguishable from the observed measures in the total variation norm, thereby avoiding issues related to impossible inference. Further details are provided in Section \ref{IDnormclosure}.

Computing the identified set based on \eqref{E:essentialIDset} involves a search over $\overline{\mathcal M}_\theta$, which can be challenging when $\Gamma_\theta$ is an infinite-dimensional space. This motivates the alternative characterization of the identified set in Theorem \ref{P:main} below, which uses the following assumptions.

\begin{asm} 
\label{A:zero}
    $\mathcal{Z}$ is a Polish space.
\end{asm}

\begin{asm} 
\label{A:one}
    For all $\theta \in \Theta$, there exists some $\sigma$-finite positive measure $\lambda_\theta \in \mathfrak{B}(\mathcal{Z})$ with respect to which every $\mu \in \M_\theta$ is continuous.
\end{asm}
Assumption \ref{A:zero} accommodates random variables $Z$ supported on various separable spaces, and excludes spaces that are non-metrizable, which rarely arise in econometric applications. Assumption \ref{A:one} requires that measures in $\mathcal{M}_\theta$ be well-behaved with respect to a $\sigma$-finite Borel measure $\lambda_{\theta}$ that is allowed to vary with $\theta$. The assumption is mild, allowing for a broad class of random variables, including continuous, discrete, and mixed types, and excluding singular measures. We discuss the necessity of this assumption in Remark \ref{A2.2_necessary} in Section \ref{sec:additional_remark}, while in Section \ref{sec:nonlinearpanels} we relax it.

Consider the discrepancy function defined in \eqref{def:T_theta} with
\begin{align}
\label{phi_b}
    \Phi_b(\mathcal{Z}) \equiv \{\phi: \mathcal{Z} \to [0,1]: \phi \text{ is Borel measurable}\},
\end{align}
and define the set of parameter values that set \eqref{def:T_theta} to zero as:
\begin{align}
\label{E:MI}
    \Thetam \equiv \{\theta \in \Theta: T(\theta) = 0\}.
\end{align}

Our main result below clarifies the connection between $\Theta_\mathrm{I}$ in \eqref{E:essentialIDset} and $\Thetam$ in \eqref{E:MI}, and serves as a building block for the subsequent analysis.

\begin{thm}[Main result]
\label{P:main}
     Let Assumptions \ref{A:zero} and \ref{A:one} hold. For any $\mu^*_Z \in \mathcal{P}(\mathcal{Z})$, 
     \begin{align} 
     \label{E:ineqthm}
         \Theta_{\mathrm{I}} \subseteq \Thetam. 
     \end{align}
     
     Additionally, let $\overline{\mathcal{M}}_\theta$ be convex for all $\theta$. Then $\Theta_{\mathrm{I}} = \Theta_{\mathrm{MI}}$.
\end{thm}
\begin{proof}
    Theorem \ref{P:main} is a direct implication of Proposition \ref{P:main_appendix}, see Section \ref{sec:proofs}.
\end{proof}

Convexity of \(\mathcal{M}_\theta\) plays a central role in our approach.\footnote{While this convexity may influence the geometry of the identified set, convexity of the identified set itself is neither implied by nor required for our results.} First, it serves as a sufficient condition for sharp identification by ensuring that \(\overline{\mathcal{M}}_\theta\) is convex. When convexity of \(\mathcal{M}_\theta\) fails, the characterization via \eqref{def:T_theta} provides an outer set. Second, if extreme points of \(\mathcal{M}_\theta\) exist, convexity of \(\mathcal{M}_\theta\) reduces the search in \eqref{def:T_theta} over \(\mu \in \mathcal{M}_\theta\) to a search over these extreme points. While characterizing these extreme points can be a complex, model-specific task, we show that many econometric models have a specific feature that renders this step unnecessary. This feature refers to the fact that the model probability can be expressed as a linear map on the space \(\Gamma_\theta\). That is,
\begin{equation}
\label{IB:linearoperator}
\mu_{Z,(\theta,\gamma)} = \mathcal{T}_\theta \gamma, \; \gamma\in\Gamma_\theta,
\end{equation}  
where \(\mathcal{T}_\theta: \Gamma_\theta \to \mathcal{P}(\mathcal{Z})\) is a linear map. We illustrate this structure with concrete examples in subsequent sections: semiparametric models where the error distribution is either unrestricted or linearly restricted (see Section \ref{sec:semiparametric_pushforward}) or parametrically specified (see Section \ref{sec:nonlinearpanels}).
The representation in \eqref{IB:linearoperator} implies that convexity of \(\Gamma_\theta\) is sufficient for Theorem \ref{P:main} to apply,\footnote{This assumption is not necessary. We show in Section \ref{sec:predetermined_panels} that under sequential exogeneity, \(\mathcal{M}_\theta\) is convex despite \(\Gamma_\theta\) failing to be convex.}enabling dimensionality-reduction via a search over the extreme points of the convex set \(\Gamma_\theta\), which are easier to characterize. 

Both linearity of the map and convexity of \(\Gamma_\theta\) enable the computation of the zeros of \(T(\theta)\) via a linear program (see Sections \ref{sec:adversarial} and \ref{linprog_parametric}), making computation of \(\Theta_{\mathrm{I}}\) through \(\Theta_{\mathrm{MI}}\) a feasible task, even when direct computation using \eqref{E:essentialIDset} may be impractical.

\begin{remark}
    The discrepancy function in \eqref{def:T_theta} follows directly from our definition of the identified set in \eqref{E:essentialIDset} and is new to the literature. We provide a comparison to other criterion functions used in the literature in Remark \ref{discrepacycomparisons} in Section \ref{IDnormclosure}.
\end{remark}

\subsubsection*{Intuition}
\label{intuition}

We provide a heuristic interpretation for the discrepancy function $T(\theta)$ and Theorem \ref{P:main}. The maintained assumption for the discussion here is that $\mathcal{M}_\theta$ is convex.

Given a \(\theta \in \Theta\) and a probability measure \(\mu^*_Z \in \mathcal{P}(\mathcal{Z})\), there may exist several \(\gamma \in \Gamma_\theta\) such that the corresponding model probabilities \(\mu_{Z,(\theta,\gamma)} \in \mathcal{M}_\theta\) are indistinguishable from the true \(\mu^*_Z\) in the sense of \eqref{eq:closureMtheta}. If there exists at least one such \(\gamma\),\footnote{Or if a sequence of \(\gamma\)'s can be constructed such that \(\mu_{Z, (\theta, \gamma)}\) converges to \(\mu^*_Z\).} then \(\mu^*_Z \in \overline{\mathcal{M}}_\theta\), and consequently, \(\theta \in \Theta_{\mathrm{I}}\). Theorem \ref{P:main} uses insights from convex analysis to solve this existence problem. In particular, the proof of Theorem \ref{P:main} shows that $\mu^*_Z \in \overline{\mathcal{M}}_\theta$ if and only if
\begin{align}
\EE{\mu^*_Z}{\phi} \le \sup_{\mu \in \mathcal{M}_\theta} \EE{\mu}{\phi} \text{ for all } \phi \in \Phi_b(\mathcal{Z}). 
\label{eq:ThetaMI_via_inequalities}
\end{align}
The discrepancy function \(T(\theta)\) in \eqref{def:T_theta} is then obtained after rearranging and taking the supremum over \(\phi\). The decision rule is based on the sign of \(T(\theta)\): If \(T(\theta) > 0\), the parameter \(\theta\) is excluded from the identified set, while if \(T(\theta) \leq 0\), \(\theta\) is included in the identified set. When $\mathcal M_\theta$ is convex, this characterization is sharp.

The discrepancy function $T(\theta)$ has an adversarial formulation where two opposing players, a critic and a defender, interact strategically. The critic selects a feature $\phi\in\Phi_b(\mathcal Z)$ and the defender selects a measure $\gamma\in\Gamma_\theta$. The critic seeks to maximize the discrepancy between the feature observed in the data, i.e. $\mathbb{E}_{\mu^*_Z}[\phi]$, and the corresponding feature predicted by the model, i.e. $\mathbb{E}_{\mu}[\phi]$, for a given parameter value $\theta$ and taking into account the defender's action. In response, the defender adjusts $\gamma$ to minimize this discrepancy. The resulting discrepancy function captures the maximum discrepancy that the critic can enforce, even after the defender optimally adjusts the probability measure of the unobserved heterogeneity. The sign of the discrepancy function determines a decision rule: A positive value means that the critic has identified a feature where the model, under a specified parameter value, fails to replicate the observed data; the parameter value is then excluded from the identified set. A non-positive value means that the defender can always find a measure that aligns the model's prediction with the observed data; the parameter value is then included in the identified set. Note that if \(T(\theta) \leq 0\) for all \(\phi\), the critic selects \(\phi = 0\) to ensure \(T(\theta) = 0\).

This interpretation, together with the central role of the discrepancy function in our identification, computation, and inference results, forms the basis of our adversarial approach.

\section{Models without parametric restrictions}
\label{sec:semiparametric_pushforward}

Models in this section are described as follows. Let the input variables be denoted by \(W \in \mathcal{W}\) (which may include latent variables) and continue to denote the observed random variables by \(Z \in \mathcal{Z}\).
\footnote{For example, \(Z\) may consist of observed outcomes \(Y\) and conditioning covariates \(X\), such as regressors and instrumental variables, whereas \(W\) may contain \(X\) along with stochastic error terms or other latent random variables. See Section \ref{sec:convexityinmodels} for further discussion.}  
In this setting, \(\Gamma_\theta \subseteq \mathcal{P}(\mathcal{W})\) is the set of probability measures supported on \(\mathcal{W}\) that are allowed by the model under the parameter value \(\theta\); we denote this set by \(\Gamma_\theta(\mathcal{W})\) to emphasize its dependence on \(\mathcal{W}\). 

For any $\theta\in\Theta$, there exists a measurable map $\psi_\theta:\mathcal{W} \mapsto \mathcal{Z}$ known up to $\theta$, 
such that 
\begin{equation}
\label{ZWas}
    Z=\psi_\theta(W) \textrm{ almost surely}.
\end{equation}
This specification includes semiparametric models with, e.g., outcome equations such as $Z=h(\beta,W)$, $\theta=(h,\beta)$, with $h$ an unknown function and $\beta$ an unknown finite-dimensional parameter, and models with outcome equations such as $Y=m(W)$, where $\theta=m$ is an unknown function. 

For fixed $\theta$, $\psi_\theta$ induces a pushforward measure $(\psi_\theta)_*: \Gamma_\theta(\mathcal W) \ra \mathcal{P}(\mathcal{Z})$ defined as:
\begin{align}
\label{eq:define_pushforward}
    ((\psi_\theta)_* \gamma)(S) \equiv \gamma(\psi_\theta^{-1}(S)) \quad \text{for all measurable sets } S \subseteq \mathcal{Z},
\end{align}
where \(\psi_\theta^{-1}(S) = \{ w \in \mathcal{W} : \psi_\theta(w) \in S \}\) is the preimage of \(S\) under \(\psi_\theta\). 
Then,
\begin{equation}
\label{mtheta_pushforward}
\mathcal{M}_\theta = \{(\psi_\theta)_* \gamma : \gamma \in \Gamma_\theta(\mathcal{W})\} = (\psi_\theta)_* \Gamma_\theta(\mathcal{W}).
\end{equation}

\begin{asm}
\label{A:pushforward}
    For any $\theta$, there exists a measurable map $\psi_\theta:\mathcal W \mapsto \mathcal Z$ such that for any $\gamma\in\Gamma_\theta$, 
    the model probability is given by:
    $$\mu_{Z,(\theta,\gamma)}(S) = (\psi_\theta)_*\gamma (S), \text { for any Borel } S \subseteq \mathcal{Z}.$$
\end{asm}

Assumption \ref{A:pushforward} is naturally satisfied in the class of models with outcome equations  described by \eqref{ZWas}. Here, the linear map in \eqref{IB:linearoperator} is the pushforward measure $(\psi_\theta)_*$, so that measures on \(\mathcal{Z}\) are obtained by ``pushing forward'' \(\gamma \in \Gamma_\theta(\mathcal{W})\) via \((\psi_\theta)_*\). \footnote{While $\psi_\theta$ here is not a correspondence, our result in Corollary \ref{thm:pushforward_sharp} may still apply to such models. This is because a sufficient condition for that result is the existence of $\mathcal T_\theta$ in \eqref{IB:linearoperator}, which can be induced by a correspondence between $\mathcal W$ and $\mathcal Z$. In Section \ref{sec:nonlinearpanels}, we extend our results to cover linear \textit{operators} from \(\Gamma_\theta(\mathcal{W})\) to \(\mathcal{P}(\mathcal{Z})\). Notably, a correspondence between \(\mathcal{W}\) and \(\mathcal{Z}\) can induce a map or an operator between spaces of probability measures on \(\mathcal{W}\) and \(\mathcal{Z}\). Deriving sufficient conditions for when such correspondences lead to well-defined \textit{linear} maps or operators is left for future work.} 

\begin{corollary}
\label{thm:pushforward_sharp}
    Let Assumptions  \ref{A:zero}, \ref{A:one}, and \ref{A:pushforward} hold, and assume that $\Gamma_\theta(\mathcal W)$ is convex. Then,
    $\mathcal{M}_\theta$ is convex 
    and $\Theta_{\mathrm{I}} = \Theta_{\mathrm{MI}}$. 
\end{corollary}
\begin{proof}
    The set $\mathcal{M}_\theta$ is convex since it is the image of a convex set under a linear map. This follows from Assumption \ref{A:pushforward}, which defines elements of $\mathcal{M}_\theta$ via the application of the linear pushforward measure $(\psi_\theta)_*$ to elements of the convex set $\Gamma_\theta(\mathcal{W})$.
    That $\Theta_{\mathrm{I}} = \Theta_{\mathrm{MI}}$ follows from Theorem \ref{P:main}, and from the convexity of $\mathcal M_\theta$.
\end{proof}

The set $\Gamma_\theta(\mathcal W)$ plays a central role in our analysis. In some semiparametric models, it is unrestricted, i.e.  $\Gamma_\theta(\mathcal{W}) = \mathcal{P}(\mathcal{W})$ is the set of all probability measures on $\mathcal W$, while in other models, $\gamma$ is known to satisfy certain linear restrictions, i.e. $\Gamma_\theta(\mathcal{W}) = \mathcal{P}(\mathcal{W})^g$ which we define below. Models with parametric restrictions on $W$ are discussed in Section \ref{sec:nonlinearpanels}.

For a known vector of functions $g:\Theta\times\mathcal{W}\to\mathbb{R}^{d_g}$, define 
\begin{align}
\label{Gamma_restricted}
        \mathcal{P}(\mathcal{W})^{g} \equiv \{\gamma\in  \mathcal{P}(\mathcal{W}): g \in L^1(\gamma) \text{ and } \EE{\gamma}{g(\theta, w)} = 0\}.
\end{align}
This is the set of all probability measures on $\mathcal W$ that satisfy a set of $d_g<\infty$ linear restrictions that may depend on $\theta$. 
As will become clear from examples throughout this paper, such restrictions are important in many econometric models.
First, many restrictions commonly made on the distribution of latent variables, such as mean- or median-independence, can be expressed as in \eqref{Gamma_restricted}. 
Second, many counterfactuals of interest can be cast in the form $\EE{\gamma}{g(\theta, w)}=0$, implicitly imposing linear restrictions on \(\gamma\).
Section \ref{sec:example_manski} illustrates this for a semiparametric binary choice model.

Corollary \ref{thm:pushforward_sharp} shows that convexity of $\Gamma_{\theta}(\mathcal W)$ is sufficient for sharp identification.
Under Assumption \ref{A:probabilitymeasure} below, each of $\mathcal{P}(\mathcal{W})$ and $\mathcal{P}(\mathcal{W})^g$ is convex, and Corollary \ref{thm:pushforward_sharp} applies. 
\begin{asm}
\label{A:probabilitymeasure}
    $\mathcal{W}$ is a Polish space. 
\end{asm}
We are now ready to establish an extremal point characterization of the result in Theorem \ref{P:main}. Characterizing the identified set using \eqref{def:T_theta} involves a search over the space $\mathcal{M}_\theta$. 
Given Assumption \ref{A:pushforward}, the search can instead be conducted over $\Gamma_\theta(\mathcal{W})$, i.e. for any $\phi\in\Phi_b(\mathcal{Z})$,
    \begin{align}
    \label{E:ps12}
    \sup_{\mu \in \mathcal{M}_\theta} \EE{\mu}{\phi} = \sup_{\gamma \in \Gamma_\theta(\mathcal{W})} \EE{\gamma}{\phi \circ \psi_\theta}.
    \end{align}
Proposition \ref{thm:convexity} below shows that this search can further be confined to a smaller space. 
\begin{proposition}[Extremal point characterization of the identified set]
\label{thm:convexity}
    Let Assumptions  \ref{A:zero}, \ref{A:one}, \ref{A:pushforward} and \ref{A:probabilitymeasure} hold. Additionally,
    \begin{itemize}
        \item[(i)] if $\Gamma_\theta(\mathcal{W}) = \mathcal{P}(\mathcal{W})$, then $\theta \in \Theta_\mathrm{I}$ if and only if
        \begin{align}
            \label{IDset_W}
            \EE{\mu^*_Z}{\phi} \le \sup_{w \in \mathcal{W}} (\phi \circ \psi_\theta)(w) \text{ for all } \phi\in \Phi_b(\mathcal{Z}).
        \end{align}
        \item[(ii)] if $\Gamma_\theta(\mathcal{W}) = \mathcal{P}(\mathcal{W})^g$, then $\theta \in \Theta_\mathrm{I}$ if and only if
        \begin{align}
            \EE{\mu^*_Z}{\phi} \le &\sup_{\{c_j, w_j\}_{j=1}^{d_g+1}} \sum_{j=1}^{d_g+1} c_j \left(\phi \circ \psi_\theta(w_j)\right)  \text{ for all } \phi\in \Phi_b(\mathcal{Z}) \nonumber \\ &\text{subject to }  \sum_{j=1}^{d_g+1} c_j g(\theta, w_j) = 0, \, \sum_{j=1}^{d_g+1} c_j = 1, \, c_j \geq 0 .\label{IDset_restricted}
        \end{align} 
    \end{itemize}
\end{proposition}

This result has important implications. 
For instance, when \(\gamma\) is unrestricted, so that \(\Gamma_\theta(\mathcal{W}) = \mathcal{P}(\mathcal{W})\), Proposition \ref{thm:convexity}(i) implies that the identified set can be characterized by searching over the set of extreme points of \(\Gamma_\theta(\mathcal{W})\), all of which correspond to Dirac measures on \(\mathcal{W}\). Thus, the search is restricted to \(\mathcal{W}\) rather than the much larger set of probability measures on \(\mathcal{W}\). This refinement is particularly useful in nonlinear panel models, as it means that determining whether \(\theta \in \Theta_I\) requires only a search over the \emph{values} of the fixed effects, rather than over all possible conditional distributions of the fixed effects.

\begin{remark}
\label{rem:LP_extreme}
The discrepancy function $T(\theta)$ in \eqref{def:T_theta} simplifies. For example, under the conditions for case (i) in Proposition \ref{thm:convexity},
\begin{equation}
T(\theta) = \sup_{\phi \in \Phi_b(\mathcal Z)} 
        \inf_{w \in \mathcal W} 
        \left(
            \EE{\mu_Z^*}{\phi} -  (\phi \circ \psi_\theta)(w) 
        \right).
        \label{eq:bilinear_extreme}
\end{equation}
Section \ref{sec:adversarial} uses this insight to show that $T(\theta)$ can be recast as a linear program.
\end{remark}

\begin{remark}
Under additional regularity conditions, the dimensionality of the search can be further reduced by considering all continuous functions $\phi:\mathcal{Z}\to [0,1]$, see Section \ref{sec:inference_dualrole}.
\end{remark}

\subsection{Semiparametric regression models} 
\label{sec:convexityinmodels}

In this section, we focus on semiparametric models with both observed and unobserved heterogeneity. Our primary goal is to clarify how conditioning variables (including regressors and instrumental variables) are treated within our framework. We also provide a blueprint for verifying the assumptions of Corollary \ref{thm:pushforward_sharp} for models with discrete and continuous conditioning variables. 
Section \ref{sec:example_manski} provides an example.

Many semiparametric models have an outcome equation of the form:
\begin{equation}
\label{eq:model_y_h}
    Y = h(X,U;\theta),
\end{equation}
where $Y \in \mathcal{Y}$, $X \in \mathcal{X}$, and $U \in \mathcal{U}$ are random variables, $\theta \in \Theta$ is an unknown parameter, and $h:\mathcal{X} \times \mathcal {U} \times \Theta \to \mathcal{Y}$ is a (structural) function known up to $\theta$. $Y$ denotes outcome variables, $X$ denotes observed variables, and $U$ unobserved variables.  

Using the notation from Section \ref{sec:semiparametric_pushforward}, let \(Z = (Y, X)\), \(W = (X, U)\), \(\mathcal{Z} = \mathcal{Y} \times \mathcal{X}\), and \(\mathcal{W} = \mathcal{X} \times \mathcal{U}\). Denote the distribution of $W$ supported on $\mathcal W$ by \(\gamma \in \Gamma_\theta(\mathcal{W})\subseteq\mathcal{P}(\mathcal{W})\), the marginal distribution of \(X\) under \(\gamma\) by \(\gamma_\mathcal{X}\in\Gamma_{\theta, \mathcal{X}}(\mathcal{X})\subseteq \mathcal{P}(\mathcal{X})\), and the conditional distribution of \(U \mid X = x\) for each \(x \in \mathcal{X}\) by \(\gamma_{U \mid x}\in\Gamma_{\theta, x}(\mathcal{U})\subseteq \mathcal{P}(\mathcal{U})\). By the disintegration theorem, measures \(\gamma \in \Gamma_\theta(\mathcal{W})\) have differential \(\mathrm{d}\gamma = \mathrm{d}\gamma_{U|x} \, \mathrm{d} \gamma_X\)
for all \(x \in \mathcal{X}\), allowing thus for correlation between $X$ and $U$.

Finally, for any \(\theta \in \Theta\), the mapping 
\begin{equation}
\label{lem:map}
    \psi_\theta: W \mapsto (h(X,U; \theta), X)
\end{equation}
induces, for each \(\gamma \in \Gamma_\theta(\mathcal{W})\), the pushforward measure \((\psi_\theta)_*\) on \(\mathcal{P}(\mathcal{Z})\), as defined in \eqref{eq:define_pushforward}. Consequently, the set \(\Gamma_\theta(\mathcal{W})\) induces the set \(\mathcal{M}_\theta = (\psi_\theta)_* \Gamma_\theta(\mathcal{W})\) as in \eqref{mtheta_pushforward}.

The assumption below allows us to specialize the assumptions of Corollary \ref{thm:pushforward_sharp} to models with outcome equation as in \eqref{eq:model_y_h}.

\begin{asm} \label{Asm:semip_sharpness}
    (i) $\mathcal{Y}$ and $\mathcal{X}$ are Polish spaces; (ii) The set $\Gamma_{\theta, \mathcal{X}}$ is convex, and the set  $\Gamma_{\theta, x}$ is convex for every $x\in\mathcal X$; (iii) There exists a $\sigma$-finite $\lambda_{\theta, \mathcal{X}} \in \mathfrak{B}(\mathcal{X})$ with respect to which every $\gamma_X \in \Gamma_{\theta, \mathcal{X}}(\mathcal{X})$ is continuous; (iv) There is a collection of $\sigma$-finite $\mathcal{X}$-measurable measures $\{\lambda_{\theta, x}: x \in \mathcal{X}\} \subseteq \mathcal{P}(\mathcal{Y})$ such that, for all $x \in \mathcal{X}$ and $\gamma_{U|x} \in \Gamma_{\theta, x}(\mathcal{U})$, the pushforward of $\gamma_{U|x}$ under $h(x,\cdot; \theta): \mathcal{U} \mapsto \mathcal{Y}$ is continuous with respect to $\lambda_{\theta, x}$.
\end{asm}

Assumption \ref{Asm:semip_sharpness}(i) is a regularity condition which fulfills the requirements of Assumption \ref{A:zero}. Assumption \ref{Asm:semip_sharpness}(ii) specifies sufficient conditions for the convexity of $\Gamma_\theta(\mathcal{W})$, and consequently the convexity of $\mathcal{M}_\theta$. Typically, $X$ are treated as conditioning variables, in the sense that the model specifies assumptions on features of the conditional distribution $\gamma_{U|x}, \; x\in\mathcal X$. By including $X$ in $W$, typical conditions on $\gamma_{U|x}$ can be formulated as linear restrictions on \(\gamma\) as in \eqref{Gamma_restricted}. If, additionally, $\mathcal U$ is a Polish space so that Assumption \ref{A:probabilitymeasure} holds, convexity of $\Gamma_\theta(\mathcal W)$ is preserved and straightforward to verify by determining whether assumptions on $U \mid X$ can be expressed as in \eqref{Gamma_restricted}. Alternatively, convexity of $\Gamma_\theta(\mathcal W)$ can be established by verifying Assumption \ref{Asm:semip_sharpness}(ii). Note that here $\gamma_\mathcal X$ is treated as a latent distribution. Since $X$ enters $Z$, it is possible to treat $\gamma_X$ as known and equal to the marginal distribution of $X$ under the observed distribution of $Z$. In this case, $\Gamma_{\theta,\mathcal X}$ is a singleton and trivially convex. Convexity of $\Gamma_{\theta, x}$ is enforced via assumptions on $\gamma_{U \mid x}$ that guarantee convexity of the set for every $x\in\mathcal X$.  

Assumptions \ref{Asm:semip_sharpness}(iii) and (iv) are mild continuity conditions that guarantee that all measures in $\mathcal{M}_\theta$ are continuous with respect to the $\sigma$-finite measure $\lambda_\theta \in \mathfrak{B}(\mathcal{Y} \times \mathcal{X})$ with differential $\d \lambda_\theta = \d \lambda_{\theta, x} \, \d \lambda_{\theta, \mathcal{X}}$ for all $x \in \mathcal{X}$, thereby satisfying Assumption \ref{A:one}. 
These assumptions require, respectively, that the marginal and conditional distributions of measures in $\mathcal{M}_\theta$ have density with respect to $\sigma$-finite measures, which may depend upon $\theta$. These dominating measures may be continuous, discrete, or a mixture of both. For example, when $\mathcal{Y}$ is discrete, $\lambda_{\theta, x}$ is the counting measure on $\mathcal{Y}$ for all $x$, in which case Assumption \ref{Asm:semip_sharpness} is fulfilled. In contrast to much of the existing relevant literature, Assumptions \ref{Asm:semip_sharpness}(iii) and (iv) allow for $(Y,X)$ to be continuous, discrete, or a mixture. Importantly, Corollary \ref{corr1_prime} pertains to semiparametric models with outcome equations of the form \eqref{eq:model_y_h}, where, for example, $X$ is continuous with respect to Lebesgue measure $\lambda_{\theta,\mathcal{X}}$ and either $Y$ is discrete with $\lambda_{\theta,x}$ the counting measure or $Y$ is continuous with $\lambda_{\theta,x}$ the Lebesgue measure (or some mixture of these cases).

\begin{corollary}
\label{corr1_prime}
    Consider the outcome equation in \eqref{eq:model_y_h}. Let Assumption \ref{Asm:semip_sharpness} hold for all $\theta \in \Theta$. Then $\Thetaw = \Thetam$.
\end{corollary}

Corollary \ref{corr1_prime} provides a blueprint for checking the assumptions of Corollary \ref{thm:pushforward_sharp}, and, consequently, for establishing sharp identification in semiparametric models with an outcome equation as in \eqref{eq:model_y_h}.

\begin{remark}[Extremal point representation]
Consider now the extremal point representation of Proposition \ref{thm:convexity}. Restrictions on the probability measure of $X$ in $Z$, such as Assumption \ref{Asm:semip_sharpness}(iii), impose restrictions on the measure of $W$; these restrictions \emph{may} constrain $\mathcal P(\mathcal W)$ and $\mathcal P(\mathcal W)^g$ in measure theoretic ways that are not allowed by Proposition \ref{thm:convexity}. 

When $X$ is discrete and $\lambda_{\theta,\mathcal X}$ in Assumption \ref{Asm:semip_sharpness}(iii) is the counting measure, this phenomenon does not occur. In this case, the assumptions of Proposition \ref{thm:convexity} hold, and $\Thetam(=\Thetaw)$ is characterized by the extremal point representation in either \eqref{IDset_W} or \eqref{IDset_restricted}. 

When $X$ is continuous and $\lambda_{\theta,\mathcal X}$ in Assumption \ref{Asm:semip_sharpness}(iii) is the Lebesgue measure, the extremal point representation in \eqref{IDset_W} or \eqref{IDset_restricted} describes an \emph{outer set} for $\Thetam$. It is still the case that $\Thetaw=\Thetam$, but \eqref{IDset_W} or \eqref{IDset_restricted} may give an outer set for $\Thetam$.\footnote{In Section \ref{sec:nonlinearpanels}, we recover an extremal point representation result for $\Thetam$ when $X$ is continuous. 
In that setting, some components of $U$ satisfy parametric restrictions. 
The additional structure on the conditional distribution of $Y$  allows us to further relax the already mild requirements on the distribution of $X$ in Corollary \ref{corr1_prime}.}
\end{remark}

\subsection{Example: semiparametric binary response}
\label{sec:example_manski}

We illustrate our approach using the model in \citet{manskiMaximumScoreEstimation1975} with a discrete regressor. We characterize the identified set for both regression coefficients and counterfactual choice probabilities.
\footnote{The purpose of this section is to illustrate our approach rather than to obtain new results. For existing results, see \citet{komarovaBinaryChoiceModels2013}, \citet{blevinsNonStandardRates2015}, and \citet{torgovitsky2019}.
} 

Consider the binary choice model with outcome equation:
\begin{equation}
Y = 1\{\beta_0 + \beta_1 X - U \geq 0\} \equiv h(X,U; \beta),
\label{eq:maxscore_outcome}
\end{equation}
where $Y\in\{0,1\} = \mathcal Y$, 
$X\in\{x_1,\cdots,x_K\} = \mathcal X$ has $K$ points of support,  
$U\in\mathcal U = \mathbb{R}$ is a scalar error term that satisfies the following conditional median-zero assumption:
\begin{align}
\label{medzero}
    \mathbb P(U \leq 0 | X = x_k) = \frac{1}{2}, \quad k = 1,\cdots,K,
\end{align} 
and $\beta = (\beta_0,\beta_1)$ are unknown structural parameters. The counterfactual choice probabilities are given by: 
\begin{align}
    \tau_k (x^*) 
    \equiv 
    \E{\left.1\{\beta_0 + x^*\beta_1 - U \geq 0 \right| X = x_k\}},
    \label{eq:restriction_due_to_ASF}
\end{align}
which correspond to the average counterfactual outcome obtained by exogenously setting the values of the observed regressors to counterfactual values $x^*$ for the subpopulation given by $X=x_k$.

To characterize the identified set of 
\begin{align}
    \theta = \left(\beta,\tau_1(x^*),\cdots,\tau_K(x^*)\right),
    \label{eq:theta_maxscore}
\end{align}
we verify the assumptions of Corollary \ref{corr1_prime}. Using the notation of the previous sections, $Z=(Y,X),\; W=(X,U)$, $h$ is given by \eqref{eq:maxscore_outcome}, $\gamma_x$ denotes the marginal distribution of $X$, $\gamma_{U|x}$ denotes the conditional distribution of $U \mid X=x, \; x\in\mathcal X$ satisfying \eqref{medzero}, and $\gamma \in \Gamma_\theta(\mathcal{W})$ denotes the distribution of $W$ with differential \(\mathrm{d}\gamma = \mathrm{d}\gamma_{U|x} \, \mathrm{d} \gamma_X\)
for all \(x \in \mathcal{X}\).   

Assumption \ref{Asm:semip_sharpness}(i) is trivially satisfied because $\mathcal Y$ and $\mathcal X$ are finite. 
The same is true for \ref{Asm:semip_sharpness}(iii) and \ref{Asm:semip_sharpness}(iv) with $\lambda_{\theta,\mathcal X}$ and $\lambda_{\theta,x}$ counting measures.
\ref{Asm:semip_sharpness}(ii) is also satisfied: (a) $\Gamma_{\theta, \mathcal X}$ is a point, so convexity is trivially satisfied; (b) for each $x$, the set $\Gamma_{\theta,x}$ is the set of all probability measures $\gamma_{U|x}$ that satisfy the linear restrictions in \eqref{medzero} and \eqref{eq:restriction_due_to_ASF}, so $\Gamma_{\theta,x}$ is convex for all $x$ (see proof of Proposition \ref{thm:convexity}).
To see that \eqref{medzero} and \eqref{eq:restriction_due_to_ASF} impose linear restrictions on $\gamma_{U|x}$ for each $x\in\mathcal X$ consider that these restrictions can be written as
$$E_{\gamma_{U|x}}\left[ \widetilde g(U,\theta) \mid X = x\right] = 0$$
where 
\begin{align}
\widetilde g(U,\theta) = \begin{bmatrix}
    1\left\{U \leq 0\right\} - \frac 1 2 \\
    1\{\beta_0 + x^*\beta_1 - U \geq 0 \} - 
    \tau_k (x^*)
   \end{bmatrix}.
\end{align}
    
Because Assumption \ref{Asm:semip_sharpness} is satisfied, Corollary \ref{corr1_prime} guarantees sharp identification of $\theta$ via $\Thetam$. Moreover, the results of Proposition \ref{thm:convexity} also apply.

It follows that the identified set can be computed using the linear programming approach developed in Section \ref{sec:adversarial}.
Section \ref{app:computation_maximum_score} describes in detail how to apply it to the semiparametric binary choice model, and presents numerical results. 
The computation time for the identified set of $\theta$ is trivial, see Section \ref{app:computation_speed_maxscore_cs}.

\subsection{Linear programming}
\label{sec:adversarial}

The computation of $\Thetam$ requires computation of the discrepancy function \eqref{def:T_theta}. This may seem challenging, as it involves a search over measures in $\mathcal{M}_\theta$ and over functions $\phi \in \Phi_b(\mathcal{Z})$. However, computing $T(\theta)$ only requires solving a linear program (LP). Below, we provide a sketch, leaving the details to Section \ref{app:computation_general}. The approach extends to the models in Section \ref{sec:nonlinearpanels}.

Our theoretical results allow for $\mathcal Z$ and $\mathcal W$ to be Polish spaces, cf. Assumptions \ref{A:zero} and \ref{A:probabilitymeasure}. 
For the purpose of computation, we may restrict attention to the case of finite supports,
\(
    \mathcal Z = \{z_1,\cdots,z_L\}, \;
    \mathcal W = \{w_1,\cdots,w_M\}.
\)
Note that the linear program presented here does not rely on discreteness of $\mathcal Z$ and $\mathcal W$. Rather, it relies on the linearity of the map between probability measures as in \eqref{IB:linearoperator}. When $\mathcal Z$ and $\mathcal W$ are discrete, the map takes the form of a matrix, which is the case here.

The probability measure $\mu_Z^*$ can be represented by a probability mass function (pmf) $p_Z^* = \left(p_{Z,l}^*\right)$,
an $L \times 1$ probability vector,
so that the first term in \eqref{def:T_theta} is
\[
\EE{\mu_Z^*}{\phi} = \sum_{l = 1}^L \phi(z_l) p_{Z,l}^* = \phi^\prime p_Z^*.
\]

Similarly, every model probability $\mu_{Z,(\theta,\gamma)}$ corresponds to a pmf
$p_{Z}^{(\theta,\gamma)} 
 = 
 \left(p_{Z,l}^{(\theta,\gamma)}\right) 
$,
and every probability measure $\gamma$ for $W$ to an $M\times 1$ pmf $p_W = \left(p_{W,m}\right)$. 
By the pushforward representation in Section \ref{sec:semiparametric_pushforward}, 
cf. \eqref{eq:define_pushforward} and Assumption \ref{A:pushforward}, 
there exists a matrix $\widetilde C_\theta \in \mathbb R^{L \times M}$  such that 
\begin{equation}
    p_{Z}^{(\theta,\gamma)} = \widetilde C_\theta p_W.
    \label{eq:pushforward_pmf}
\end{equation}
For a given $\theta$, this pushfoward matrix maps the pmf of $W$ to a model pmf of $Z$ under parameter value $\theta$.
For examples of $\widetilde C_\theta$, see Sections \ref{sec:computation_examples} and \ref{app:computation_maximum_score}.

The second term of \eqref{def:T_theta} can be written as
\begin{equation}
    \EE{\mu}{\phi} 
    = 
    \phi^\prime \widetilde C_\theta p_W,
\end{equation}
so the infimum over $\mu$ in \eqref{def:T_theta} can be replaced by a minimum over $p_W$
subject to 
\begin{equation}
\label{equalconstraints}
    p_W \geq 0, \quad A_\theta p_W = b_\theta.
\end{equation}
These constraints enforce that $p_W$ is a probability vector and allow for additional constraints corresponding to $\Gamma_\theta$.
Without additional constraints, $A_\theta = \iota_M^\prime$ and $b_\theta = 1$. 
For examples of other constraints, see Sections \ref{sec:computation_examples} and \ref{app:computation_maximum_score}.

Taken together, we have
\begin{align}
    T(\theta)
    &= 
    \max_{\phi \in \mathbb R^L: \; 0 \leq \phi \leq 1} \quad
    \min_{p_W  \in \mathbb R^M: \; p_W \geq 0, \; A_\theta p_W = b_\theta} \quad
        \phi^\prime C_\theta p_W,
    \label{eq:T_is_bilinear}
\end{align}
where
\begin{equation}
    C_\theta = p_Z^* \iota_M^\prime - \widetilde C_\theta.
    \label{def:C_theta}
\end{equation}
In Section \ref{app:linear_programming}, we show that the bilinear program in \eqref{eq:T_is_bilinear} is equivalent to the LP
\begin{equation}
T(\theta) = \left\{
\begin{array}{ll}
\text{max}_{\lambda,\phi}   & \lambda^\prime b_\theta \\
\text{subject to}           & A_\theta^\prime \lambda \le C_\theta^\prime \phi , \\
                            & 0 \leq \phi \leq 1,
\end{array}
\right.
\label{eq:LP_formulation_for_Theta_I_main}
\end{equation}
where $\lambda$ is the Lagrange multiplier associated with the equality constraints in \eqref{equalconstraints}.
We conclude that determining whether $\theta \in \Theta_I$ amounts to solving the LP \eqref{eq:LP_formulation_for_Theta_I_main} and checking that $T(\theta) \leq 0$. This task has negligible computation time even for very large $(L,M)$, see for example the computation times reported in Section \ref{app:computation_speed_maxscore_cs}. 
It is also trivial to write the code for a specific model.
The user specifies 
(i)  supports $\mathcal Z, \, \mathcal W$;
(ii) true parameter values $(\theta^*,p_W^*)$;
(iii) the pushforward matrix $\widetilde C_\theta$;
(iv) restrictions $(A_\theta,b_\theta)$;
then hands off the LP \eqref{eq:LP_formulation_for_Theta_I_main} to a solver.

\begin{remark}[Relationship of LP to Proposition \ref{thm:convexity}]
Proposition \ref{thm:convexity} and Remark \ref{rem:LP_extreme} establish the extremal point problem, i.e. when optimizing a linear functional over a convex set of probability measures, the search can be restricted to the set of extreme points.
Expression \eqref{eq:T_is_bilinear} reformulates the discrepancy function as a bilinear program, where the inner minimization is over the convex hull of extreme points. 
This is a relaxation of the extremal point problem that leaves the value of $T(\theta)$ unchanged, but facilitates computation via the dual LP in \eqref{eq:LP_formulation_for_Theta_I_main}.
There, the problem is expressed in terms of Lagrange multipliers associated with the constraints on the extreme points, eliminating the need to compute them explicitly.
\end{remark}

\begin{remark}
The resulting LP is different from the LP in \citet{honoreBoundsParametersPanel2006} and \citet{bonhommeIdentificationBinaryChoice2023}, and from the quadratic program in \citet{ChernValHahnNewey2013}. In contrast to those approaches, ours avoids a search over the space of all conditional distributions of unobserved heterogeneity, with
the inner minimization restricted instead to the \emph{support} of the unobserved heterogeneity, cf. Proposition \ref{thm:convexity}.
This is expressed by the dual in the LP discussed above. In the resulting LP, $\phi$ ranges over $\Phi_b$, which is typically of low dimension.
Searching over a subset of $\phi$ yields an outer set. 
In contrast, an incomplete search over fixed effects distributions may lead to an inner set.
\end{remark}

\section{Models with parametric restrictions}
\label{sec:nonlinearpanels}

In this section, the latent component $U$ in \eqref{eq:model_y_h} consists of two components: one that follows a known distribution and another whose distribution is unrestricted. The unrestricted component is labeled $\alpha$ to evoke the literature on parametric nonlinear panel models, see, e.g., \citet{Bonhomme2012}. 

We consider the following semiparametric regression model:
\begin{align}
    \label{E:ps25} 
    Y = h(X, \alpha, V; \beta), \;  \text{where }(X,\alpha) \sim \gamma_{X,\alpha} \; \text{ and } V | X,\alpha \sim F_{V|X,\alpha;\beta}, 
\end{align}
where $\beta \in \mathrm{B}$ is the parameter of interest, the observed random variables are $Z=(Y,X)$, the latent random variables $U=(\alpha,V)$  consist of a component $\alpha\in\mathcal{A}$ and a component $V\in\mathcal{V}$, where $\gamma_{X,\alpha} \in \mathcal{P}(\mathcal{X} \times \mathcal{A})$ denotes the joint distribution of $(X,\alpha)$ and $F_{V|X,\alpha;\beta}$ denotes the distribution of $V$ that is allowed to depend on  $(X,\alpha,\beta)$. Here, $\gamma_{X,\alpha}$ is unknown and unrestricted, while $F_{V|X,\alpha;\beta}$ is known up to $\beta$.

In Section \ref{sec:panel_ID_theta}, we characterize the identified set for $\beta$. 
In Section \ref{sec:panel_ID_counterfactuals}, we characterize the identified set of counterfactual parameters.
In Section \ref{linprog_parametric}, we extend the LP approach in Section  \ref{sec:adversarial}.
In Section \ref{sec:panel_examples}, we apply the tools developed here to a parametric binary choice model with fixed effects.

\subsection{Identification of common parameters}
\label{sec:panel_ID_theta}

For the models in this section, the pushforward representation for the model probability in Assumption \ref{A:pushforward} does not hold. Instead, we show that the model probability can be represented as a linear integral operator similar to the one in \citet{Bonhomme2012}.

In the following assumption, $f_{Y|X,\alpha}$ is the conditional density of $Y$ given $(X,\alpha)$ determined by the outcome equation and the parametric distribution of $V$, see \eqref{E:ps25}.
\begin{asm}
\label{A:momentdiscrete}
   For $\beta \in \mathrm{B}$, the densities $(y,x,a) \mapsto f_{Y|X,\alpha}(y| x, a;\beta)$ are defined relative to a $\sigma$-finite measure $\lambda_\mathcal{Y} \in \mathfrak{B}(\mathcal{Y})$, and Borel measurable for all $(x,a)\in\mathcal X \times \mathcal A$.
\end{asm}

The parametric nature of the distribution of $V$ allows us to drop Assumption \ref{A:one}, with Assumption \ref{A:momentdiscrete} imposing a measure continuity condition only on the conditional distributions of $Y$ given $X = x, \alpha = a$ for $\beta \in \mathrm{B}$.

Denote by $\gamma_X$ the $\mathcal{X}$-marginal of $\gamma_{X,\alpha}$. 
Under \eqref{E:ps25} and Assumption \ref{A:momentdiscrete}, for fixed $\beta$, the conditional density $f_{Y|X,\alpha}$ and the joint measure $\gamma_{X, \alpha}$ induce the following linear integral operator:
\begin{equation}
\label{linearoperator}
    \mathcal{L}_\beta[\gamma_{X,\alpha}](S) \equiv \int_{\mathcal{X} \times \mathcal{A}} \int_{\mathcal{Y}} \mathbf{1}_S(y, x) f_{Y|X,\alpha}(y \mid x, a; \beta) \, \mathrm{d}\lambda_{\mathcal{Y}}(y) \, \mathrm{d}\gamma_{X,\alpha}(x,a),
\end{equation}
where $S\subseteq \mathcal Y \times \mathcal X$ is a measurable set and $\mathbf{1}_S(y, x)$ is the indicator function for $S$. That is, for any joint probability measure $\gamma_{X, \alpha}\in \mathcal{P}(\mathcal{X} \times\mathcal{A})$, $\mathcal{L}_\beta[\gamma_{X,\alpha}]$ is the probability measure of $Z$ under $\gamma_{X, \alpha}$. In particular,
\begin{equation}
    \mu_{(Z, \beta, \gamma_{X, \alpha})}(S) = \mathcal{L}_\beta[\gamma_{X,\alpha}](S), \; \text{for all } S\subseteq \mathcal Y \times \mathcal X.
\end{equation}
That is, when a latent component follows a parametric distribution, the linear map $\mathcal T_\theta$ in \eqref{IB:linearoperator} is the integral operator $\mathcal{L}_\beta$. The set of all model probabilities compatible with $\beta$ can then be defined as
\[
\mathcal M_\beta = \left\{ \mathcal{L}_\beta[\gamma_{X,\alpha}]: \gamma_{X,\alpha} \in \mathcal P(\mathcal X \times \mathcal A) \right\}.
\]

\begin{proposition}
\label{cor:panel_ID_beta}
   Let $\mathcal Y, \mathcal X, \mathcal A$ be Polish spaces and let Assumption \ref{A:momentdiscrete} hold for the model described by \eqref{E:ps25}. Let $ \mu_{Z}^*$ be the true probability measure of $Z=(Y,X)$. Then $\beta$ is in the identified set if and only if 
   \begin{align}
   \label{eq:max_for_all_phi_panel}
       \EE{\mu^*_{Z}}{\phi(Z)} & \le \sup_{\substack{ x \in \mathcal{X} \\ a \in \mathcal{A} }} 
       \int_{\mathcal{Y}} \phi(y,x) f_{Y|X,\alpha} (y|x, a; \beta) \, \d \lambda_{\mathcal{Y}} \text{ for all } \phi\in\Phi_b(\mathcal Y \times \mathcal X).
   \end{align}
\end{proposition}

\subsection{Identification of counterfactual parameters}
\label{sec:panel_ID_counterfactuals}

We now turn to the identification of counterfactual parameters \(\tau\). Here, \(\theta = (\beta, \tau)\) denotes the joint parameter of interest. The counterfactual parameters $\tau$ are defined as solutions to the integral equation:
\begin{align}
\label{eq:tau_integral}
\int_{\mathcal{A}} g(x, \alpha, \beta, \tau) \, \mathrm{d}\gamma_{\alpha|x} = 0, \quad \text{for } \gamma_X\text{-a.s.},
\end{align}
where \(g(x, \alpha, \beta, \tau): \mathcal{X} \times \mathcal{A} \times \Theta \to \mathbb{R}^{d_g}\) is a known, vector-valued function. Typically, $\tau$ is a function from $\mathcal{X}$ to $\R^{d_g}$ and $g(x,\alpha, \beta, \tau) = g_0(x,\alpha, \beta) - \tau(x)$ for some $g_0: \mathcal{X} \times \mathcal{A} \times \mathrm{B} \ra \R^{d_g}$, the linear restriction in \eqref{eq:tau_integral} becomes 
\begin{align} \label{E:taug2}
    \int_{ \mathcal{A}} g_0(x,a, \beta) = \tau(x), \, \text{a.e.}. 
\end{align}
Various specifications of $g_0$ in \eqref{E:taug2} allow a researcher to restrict the conditional moments of $\gamma_{X,\alpha}$ and determine if specific conditional moments, parameterized by $\tau$, belong to the identified set.

Using the operator \(\mathcal{L}_\beta\) defined in \eqref{linearoperator}, the set of model probabilities is then:
\begin{align}
\mathcal{M}_\theta &= \left\{ \mathcal{L}_\beta[\gamma_{X,\alpha}] : \gamma_{X,\alpha} \in \Gamma_\tau \right\} \text{ where }\\
\Gamma_{\tau} &= \{\gamma_{X,\alpha} \in \mathcal{P}(\mathcal{X} \times \mathcal{A}):\, g(x,\cdot ,\beta, \tau) \in L^1(\gamma_{\alpha|x}) \text{ and } \eqref{eq:tau_integral} \text{ holds}\}.
\end{align}

In order to identify $(\beta, \tau)$, we impose a light regularity constraint on $g$.
In the assumption that follows, $\R^0$ is identified with $\{0\}$, and a function that maps to $\R^0$ is the zero-map. 
\begin{asm}
\label{A:marginalcond} 
$g: \mathcal{X} \times \mathcal{A} \times \Theta \ra \R^{d_g}$ is Borel and satisfies 
\begin{align} \label{E:taueqn}
        0 \in \mathrm{co}\{g(x,a, \beta, \tau): a \in \mathcal{A}\}, \text{ for all }x \in \mathcal{X}.
    \end{align}
\end{asm}

\begin{proposition}
\label{cor:counterfactual_ID}
Let $\mathcal Y, \mathcal X, \mathcal A$ be Polish spaces and let Assumptions \ref{A:momentdiscrete} and \ref{A:marginalcond} hold for the model described by \eqref{E:ps25}. Let $ \mu_{Z}^*$ be the true probability measure of $Z=(Y,X)$. Then, \(\theta = (\beta, \tau)\) is in the identified set if and only if, for all $\phi\in\Phi_b(\mathcal Y \times \mathcal X)$:
\begin{align}
\mathbb{E}_{\mu_Z^*}[\phi(Z)] \le &\sup_{\substack{x \in \mathcal{X} \\ (c_j, a_j)_{j = 1}^{d_g + 1} }} \sum_{j = 1}^{d_g + 1} c_j \int_{\mathcal{Y}} \phi(y,x) f_{Y|X,\alpha} (y|x,a_j; \beta) \, \d \lambda_{\mathcal{Y}} \nonumber \\
        & \text{subject to }\sum_{j = 1}^{d_g + 1} c_j g(x, a_j, \beta, \tau) = 0,\, \sum_{j = 1}^{d_g +1} c_j = 1, \, c_j \ge 0.\label{E:ps133}
\end{align}
\end{proposition}

\subsection{Linear programming}
\label{linprog_parametric}

We extend the LP approach developed in Section \ref{sec:adversarial} to the present setting in Appendix \ref{sec:comparison_computation}. Here, we note that in the parametric error case, the optimization problem is reformulated using a linear integral operator, reflecting the parametric structure imposed on \( V \). Specifically, the model probabilities are expressed as an integral transform of the joint measure of the unrestricted components of the latent variables through a conditional density function. This representation mirrors the structure observed in the pushforward map from Section \ref{sec:semiparametric_pushforward} but adapts to the parametric setting. 

Crucially, the dimensionality of the resulting LP depends only on the support of \( X \) and \( \alpha \), and not on that of \( V \). This invariance ensures that the computational complexity remains manageable even as the dimension of the latent space increases.

\section{Binary choice with fixed effects}
\label{sec:application_to_panels}

We study the two-period binary choice model with outcome equation
\begin{align}
	Y_t &= 1\{X_t^\prime \beta + \alpha + V_t \geq 0\},\; t=1,2,
	\label{eq:maxscore_panel_outcome_new}
\end{align}
which has a fixed effect $\alpha \in \mathcal A \subseteq \mathbb R$,
error terms $V_t \in \mathcal{V}_t \subseteq \mathbb R$,
and time-varying regressors $X_1 \in \mathcal X_1 = \{x_{11},\cdots,x_{1K_1}\}$ 
and
$X_2 \in \mathcal X_2 = \{x_{21},\cdots,x_{2K_2}\}$. 

We consider the particular specification in \eqref{eq:maxscore_panel_outcome_new} because of its canonical status within the nonlinear panel literature.
Our method applies to a much larger class of models, and does not require the index structure or additive separability.

In Sections \ref{sec:panel_partial_effects} and \ref{sec:predetermined_panels}, we do not impose parametric assumptions on the distribution of the error terms. In Section \ref{sec:panel_partial_effects}, we provide new results for \emph{partial effects} under strict exogeneity, and discuss the extension to correlated random coefficients. 

In Section \ref{sec:predetermined_panels}, we provide new results for both structural parameters and partial effects under \emph{sequential exogeneity} (predetermined covariates). To the best of our knowledge, these are the first such results for the binary choice panel model without parametric restrictions on the distribution of the error terms.
Section \ref{sec:numerical_panel} presents the results from a numerical experiment and visualizes the identified sets.

In Section \ref{sec:panel_examples}, we use our results in Section \ref{sec:nonlinearpanels} to analyze \emph{parametric} binary choice models with fixed effects (logit, probit, and other distributions). We recover the results in the numerical experiment of \citet{ChernValHahnNewey2013} and explore some variations of their design.
The proofs of all results can be found in Appendix \ref{app:panel_proofs}.

\subsection{Partial effects under strict exogeneity}
\label{sec:panel_partial_effects}

We are interested in the average structural functions (ASF):
\begin{align}
    \tau_t(x) = E[1\{x^\prime \beta + \alpha + V_t \geq 0\}], \; t=1,2.
    \label{eq:def_ASF_binary_panel}
\end{align}
This is the period $t$ choice probability obtained by exogenously setting regressors $X_t$ to a fixed value $x$.
We characterize the identified set of $\theta = (\beta,\tau_1,\tau_2)$ under the following assumption:
\footnote{Section \ref{sec:predetermined_panels} provides results under a weaker, sequential exogeneity condition.}
\begin{asm}
\label{a:conditional_stationarity}
The random variables $(\alpha,V_1,V_2,X_1,X_2)$ satisfy
\[V_1 | \alpha, X_1, X_2 \stackrel{d}{=}
  V_2 | \alpha, X_1, X_2.\]
\end{asm}
This is the standard strict stationarity or exogeneity assumption for nonlinear panel models, see \citet{Manski1987} and \citet{ChernValHahnNewey2013}. 

To map this model to the framework of Section \ref{sec:semiparametric_pushforward}, define 
$Y = (Y_1,Y_2)$
and $X = (X_1,X_2)$
so that $Z = (Y,X)$. 
The input variables are 
$W = (\alpha, V_1, V_2, X_1, X_2) \in \mathcal W$. 
The set of probability distributions of $W$ consists of all $\gamma\in\Gamma_\theta(\mathcal W)$ that satisfy Assumption \ref{a:conditional_stationarity}. 
Finally, the mapping \(\psi_\theta\) is given by
\begin{align}
\label{binarypsi}
    \psi_\theta: W \mapsto (1\{X_1' \beta + \alpha + V_1 \geq 0\}, 1\{X_2' \beta + \alpha + V_2 \geq 0\}, X_1, X_2).
\end{align}
Using this mapping and Corollary \ref{thm:pushforward_sharp}, the following result establishes sharp identification of partial effects.
\begin{thm}
Consider the model described by outcome equation \eqref{eq:maxscore_panel_outcome_new} and Assumption \ref{a:conditional_stationarity}. The set $\Theta_{\mathrm{MI}}$ is the sharp identified set for $\theta = (\beta,\tau_1,\tau_2)$.
\label{thm:joint_stat}
\end{thm}

Computation of $\Thetam$ is straightforward via the LP approach that we developed in Section \ref{sec:adversarial}. In Section \ref{sec:numerical_panel}, we use this approach to visualize $\Thetam$. 

\begin{remark}
Consider a version of this model with the outcome equation
$$Y_t = 1\{X_{1t}^\prime \beta + X_{2t}^\prime \alpha + V_t \geq 0\},$$
where $\alpha$ is a random vector of individual-specific coefficients that may be correlated with $X_1, X_2$. 
Without restrictions on the distribution of $(\alpha,X)$, the proof of Theorem \ref{thm:joint_stat} carries through without modification and the identified set of $\theta$, which can include moments of the random coefficients, is characterized by $\Thetam$.
\end{remark}

\begin{remark}
\citet{aristodemouSemiparametricIdentificationPanel2021} and \citet{ChesherRosenZhang} consider the identification of regression coefficients in nonlinear panel models under (conditional) independence restrictions on the errors and regressors. For example, under $V \perp X$ instead of Assumption \ref{a:conditional_stationarity}, $\Gamma_\theta$ is convex, and our results deliver partial effects.
\end{remark}

\subsection{Predeterminedness}
\label{sec:predetermined_panels}

Assumption \ref{a:conditional_stationarity} does not allow for correlation between current covariates and past shocks. A sequential exogeneity assumption that allows for such  feedback is:
\begin{asm}
\label{a:sequential_stationarity}
The random variables $(\alpha,V_1,V_2,X_1,X_2)$ satisfy
\begin{equation}
\label{E:seqex}
	V_2 | \alpha, X_1 , X_2 \stackrel{d}{=}
	V_1 | \alpha, X_1.
\end{equation}
\end{asm}
This is Assumption 3 in \citet{ChernValHahnNewey2013}, which is less restrictive than the sequential exogeneity assumption in parametric models.\footnote{It does not specify the marginal distribution of the error terms, nor does it require serial independence.}
Under Assumption \ref{a:sequential_stationarity}, $\Gamma_\theta(\mathcal W)$ is not convex.\footnote{The restriction that $V_2$ and $X_2$ are independent conditional on $(A,X_1)$ is  nonlinear.}
Nonetheless, the following result establishes convexity of the set of model probabilities $\mathcal M_\theta$ and uses Theorem \ref{P:main} to characterize the identified set.
\begin{thm}
   \label{thm:sharpness_sequential_exogeneity}
    Consider the model described by outcome equation \eqref{eq:maxscore_panel_outcome_new} and Assumption \ref{a:sequential_stationarity}. The set $\Thetam$ is the sharp identified set for $\theta=(\beta,\tau_1,\tau_2)$.
 \end{thm}
$\Thetam$ may be computed using Theorem \ref{P:main}, the pushforward representation of Section \ref{sec:convexityinmodels}, and the extremal point representation of Proposition \ref{thm:convexity}. Setting $\psi_\theta$ as in \eqref{binarypsi} and $W$ as in Section \ref{sec:panel_partial_effects}, Theorem \ref{P:main} implies that $\Thetam$ is the set of $\theta$ satisfying
\begin{align} \label{E:seqbnd}
    \EE{\mu^*}{\phi(Z)} \le \sup_{\gamma \in \Gamma_\theta(\mathcal{W}) } \EE{\gamma}{\phi \circ \psi_\theta(W)} \text{ for all } \phi \in \Phi_b(\mathcal{Z}),
\end{align}
where $\Gamma_\theta(\mathcal{W})$ is the set of all distributions satisfying \eqref{E:seqex}. Let $\Gamma^\mathrm{seq}$ be the set of all distributions of $(V_1, V_2, X_2)$ which satisfy $V_2|X_2 \overset{d}{=} V_1$. Then, because both sides of \eqref{E:seqex} are conditional on $(\alpha, X_1)$, the extremal points of $\Gamma_\theta(\mathcal{W})$ all take the form $\gamma \times \delta_{(a, x_1)}$, $\gamma \in \Gamma^\mathrm{seq}$. As in the proof of Proposition \ref{thm:convexity}, we may rewrite the right hand side of \eqref{E:seqbnd} as the supremum over $a, x_1$ of
\begin{align} \label{E:seqdual}
   \sup_{\gamma \in \Gamma^\mathrm{seq}} \EE{\gamma}{\phi \circ \psi_\theta(a, V_1, V_2, x_1, X_2)} = \sup_{\gamma_{X_2}} \sup_{\gamma \in \Gamma^\mathrm{seq}(\gamma_{X_2})}  \EE{\gamma}{\phi \circ \psi_\theta(a, V_1, V_2, x_1, X_2)},
\end{align}
where $\gamma_{X_2}$ is a marginal distribution of $X_2$ and $\Gamma^\mathrm{seq}(\gamma_{X_2})$ is the subset of $\Gamma^\mathrm{seq}$ having $X_2$-marginal $\gamma_{X_2}$. The inner supremum in \eqref{E:seqdual} is an LP and the outer supremum is only over the space of distributions of $X_2$ and points $a, x_1$. Section \ref{sec:numerical_panel} below applies this extremal point characterization to obtain $\Thetam$ for various $\theta_0$.  

\subsection{Numerical experiment}
\label{sec:numerical_panel}

\begin{figure}
    \begin{subfigure}{0.5\textwidth}
        \includegraphics[width=\textwidth]{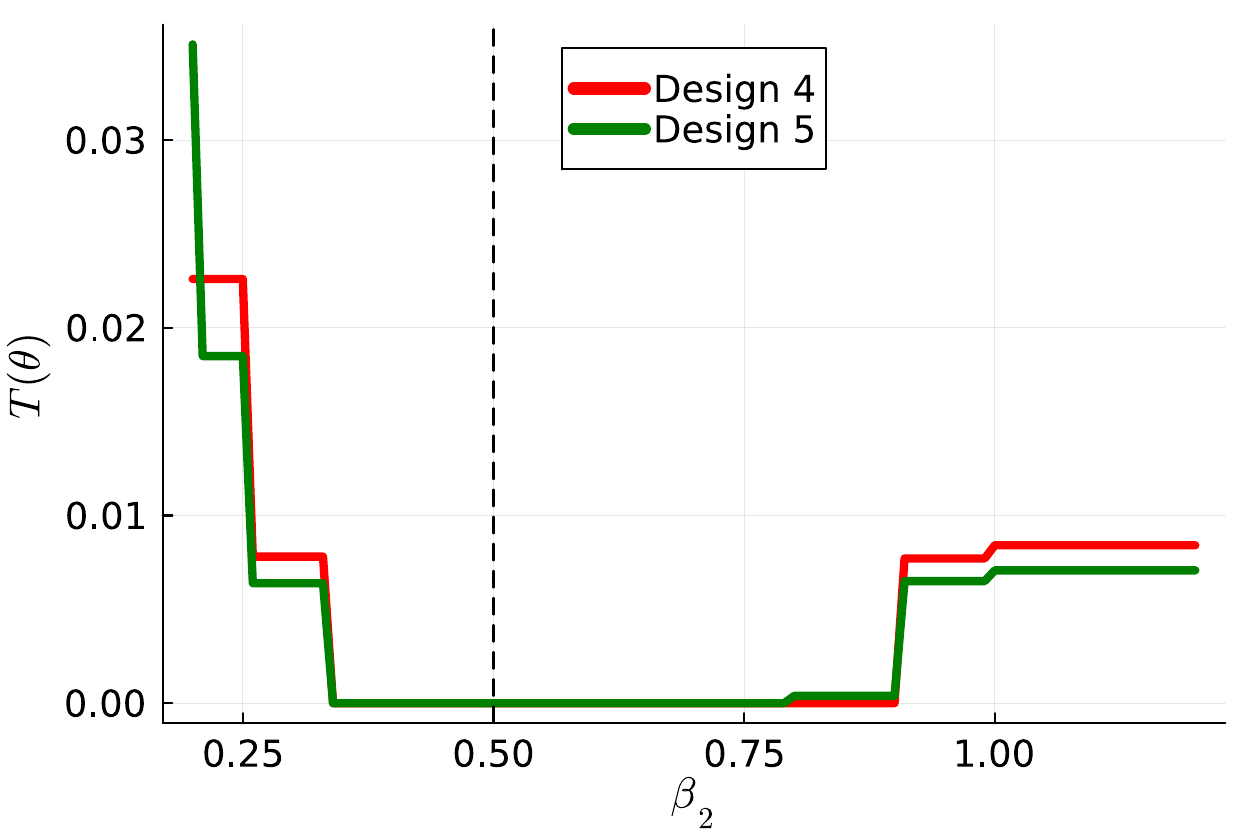} 
        \caption{DGP1: Discrepancy function $T(\theta)$.}
        \label{fig:pd_maxscore_fixed_beta0_T}
    \end{subfigure}\hfill
    \begin{subfigure}{0.5\textwidth}
        \includegraphics[width=\textwidth]{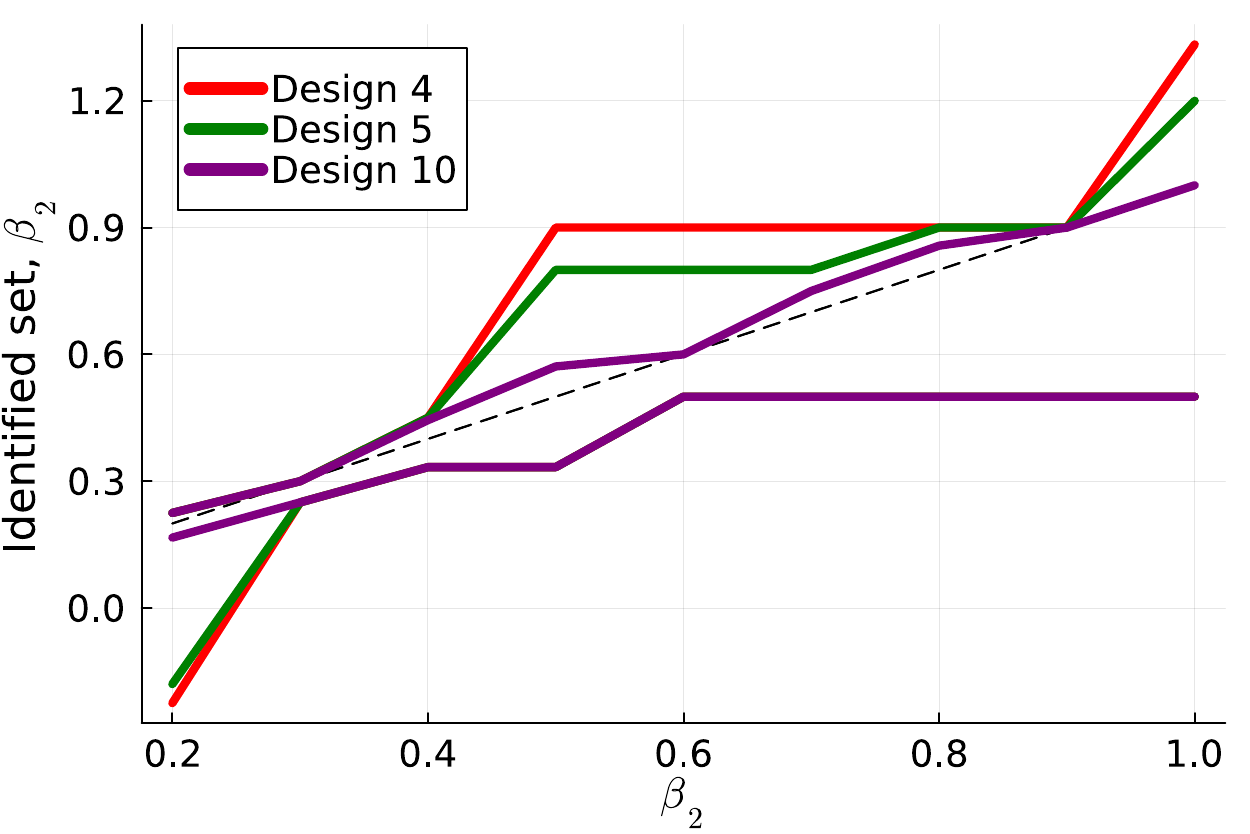} 
        \caption{DGP1: Identified set for $\beta_2$.}
        \label{fig:pd_maxscore_BMI}
    \end{subfigure}

    \vspace{1em}
    
    \centering
    \begin{subfigure}{0.5\textwidth}
        \centering
        \includegraphics[width=\textwidth]{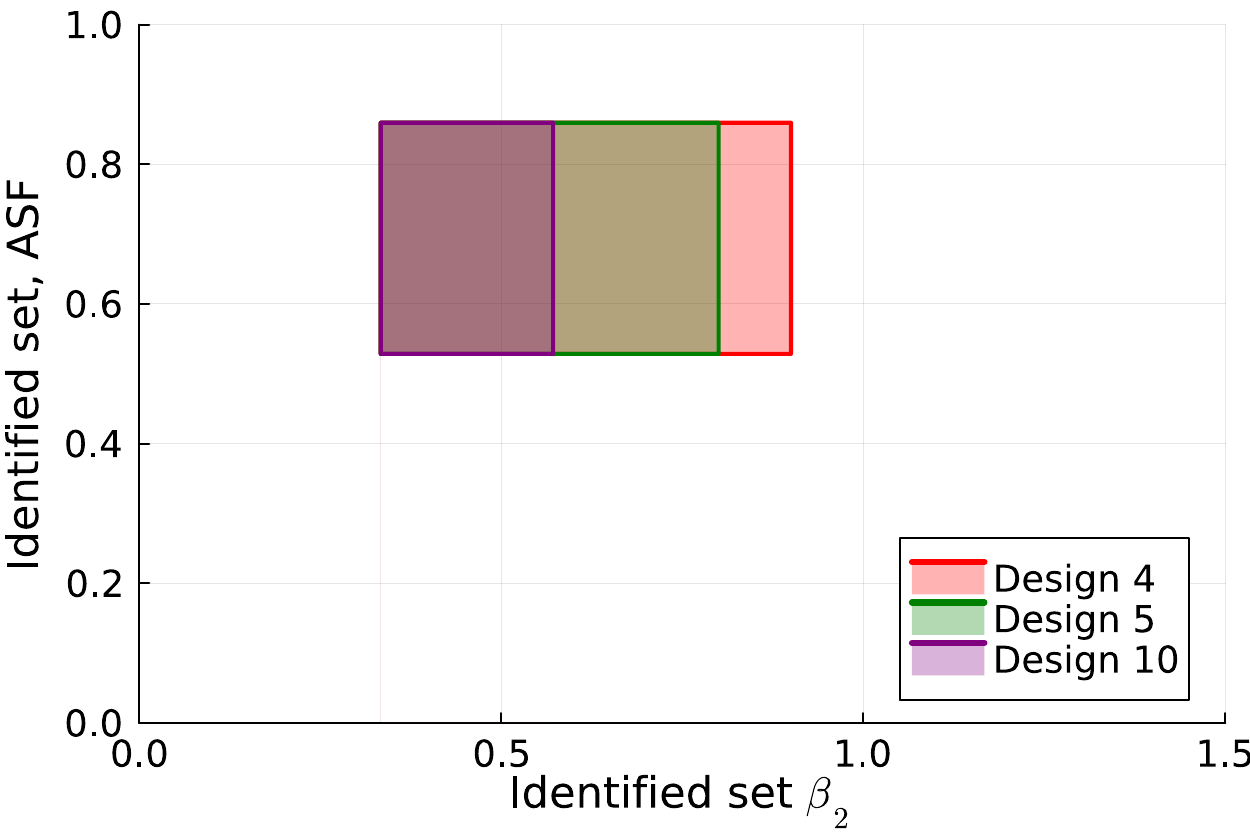} 
        \caption{DGP1: Identified set for $\beta_2$ and ASF.}
        \label{fig:panel_ASF_beta_05}
    \end{subfigure}
    \hfill
    \begin{subfigure}{0.48\textwidth}
        \centering
        \includegraphics[width=\textwidth]{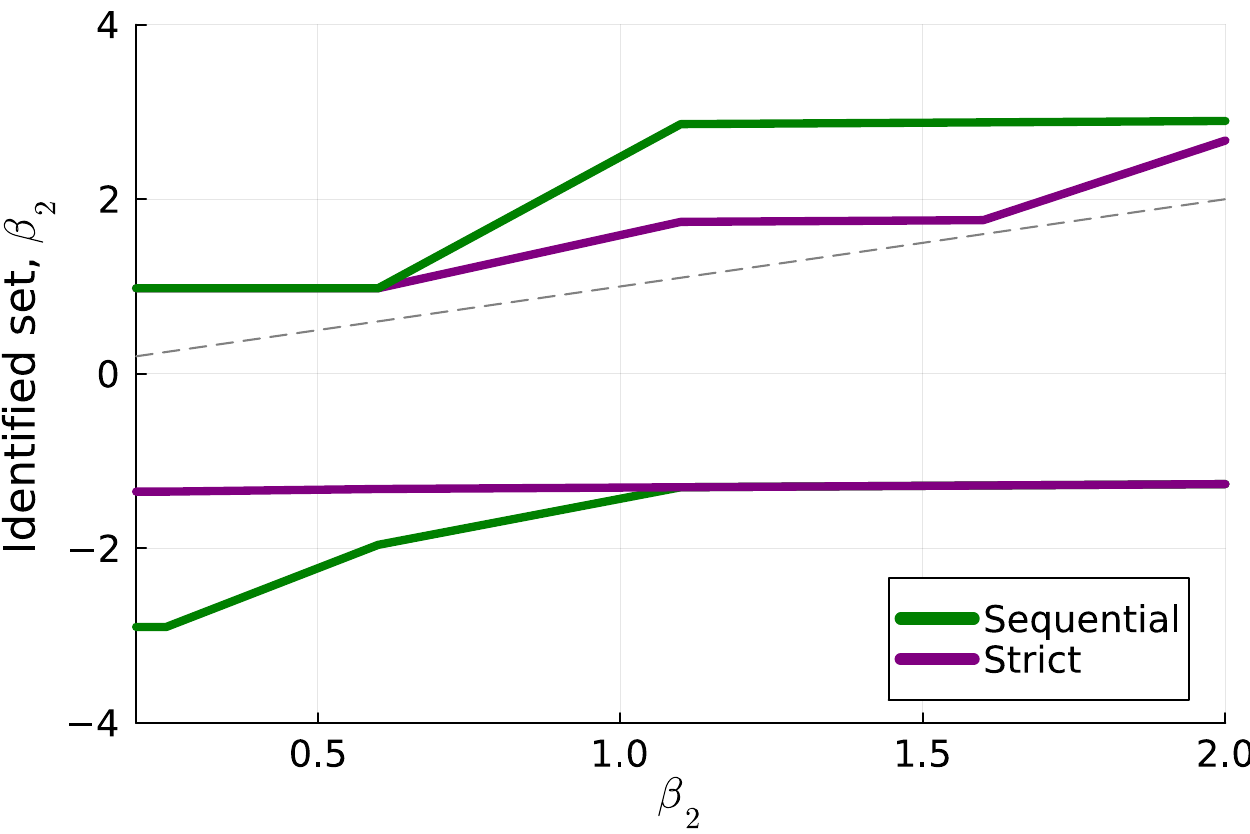} 
        \caption{DGP2: Identified set for $\beta_2$.}
         \label{fig:panel_sequential}
    \end{subfigure}
    \caption{Results of numerical experiment for DGP1 (strict exogeneity) and DGP2 (strict and sequential exogeneity).}
    \label{fig:pd_maxscore_fixed_beta0}
 \end{figure}

We conduct a numerical experiment for two different DGPs. Using DGP1, we explore the size of the identified sets of the regression coefficient and the partial effects under strict exogeneity. Using DGP2, we explore the relative widths of the identified set of the regression coefficient under strict and sequential exogeneity.

\textbf{DGP1. }
This DGP includes a time dummy $X_{1t} = t-1$ and a regressor $X_{2t}$ whose support we vary across designs as follows. In design $p$, $X_{21} = 0$ and $X_{22} \in \{-p,-p+1,\cdots,p\}$, so that higher values of $p$ imply more variation in the second regressor.
All designs have $X$ discrete uniform, 
and $(\alpha,V_1,V_2)$ independent of $X$, with $P((V_1,V_2,\alpha) = (u_1,u_2,a) \propto \exp(-u_1^2/2 - u_2^2 /2 - a^2/2)$, with support for $V_1,V_2$ as $\{-3,-2.9,\cdots,2.9,3\}$ and the support of $\alpha$ as $\{-2,-1,\cdots,2\}$.
The identified sets are computed via linear programming, see \ref{app:computation_maximum_score} for details on the implementation.

Figure \ref{fig:pd_maxscore_fixed_beta0_T} presents results for $\beta_0 = (1,-0.5)$ and designs $p \in \{4,5\}$. We 
plot the discrepancy function $T(\theta)$ against candidate values $\beta_2$.
The identified set is smaller for Design 5 (green line) as expected, because of increased variation in $X$. 
Figure \ref{fig:pd_maxscore_BMI} presents the identified set for $\beta_2$ as we vary the true value of the regression coefficient. 
As expected, additional variation in the regressors tightens the bounds on the regression coefficient, see Design 10 (purple line).\footnote{\citet{Manski1987} shows that point identification obtains under continuous variation in (one of the) regressors over the entire real line.}
Figure \ref{fig:panel_ASF_beta_05} presents results for the (joint) identified set for the regression coefficient and the ASF, at $\beta_2 = 0.5$. The ASF is for time period $t=1$, counterfactual value $x = 1$, and a subpopulation with $X_{21}=0,X_{22}=1$.  
The identified set for the ASF is the height of the box; the identified set for the regression coefficient is its width (coinciding with that in Figure \ref{fig:pd_maxscore_BMI}).
The ASF is not point identified because of the time dummy (see \cite{BotosaruMuris2024}). The identified sets for the ASF are informative across all designs. 
The information about the regression coefficient does not translate to information about counterfactual parameters.

\textbf{DGP2. }
We consider a worst case of the preceding design, with $X_{2t}$ independently and uniformly distributed on $\{0,1\}$. 
The time effect is $\beta_1 = 2$. 
The distribution of $\alpha + V_1$ is uniformly distributed on 11 equidistant points in the interval $[-3,3]$, 
and we set
$V_1 = V_2$ (so strict exogeneity holds).
We compute the identified set for $\beta_2$ under strict exogeneity as before, and under sequential exogeneity by applying \eqref{E:seqbnd} as described in Section \ref{sec:predetermined_panels}.\footnote{A genetic optimizer from the \texttt{deap} Python package handles the outer optimization over $\phi \in \Phi_b$.}

Figure \ref{fig:panel_sequential} depicts identified sets for $\beta_2$ under the assumptions of sequential and strict exogeneity as $\beta_2$ is varied. 
The identified sets are larger than the ones under DGP1. This is as expected, as DGP2 has little variation in $X_{2t}$. The results clearly show that the identified set is larger under sequential exogeneity.

\subsection{Parametric errors}
\label{sec:panel_examples}

Let $H(\cdot)$ denote an arbitrary cumulative distribution function.
In this section, we assume that the conditional distribution of the error terms is known, and given by
\[
P(V_t \leq v | \alpha = a, X = x) = H(v).
\]
Furthermore, we assume that the $V_t$ are independent across $t$ conditional on $(\alpha,X)$. 
This obtains the static binary choice model with fixed effects,
\begin{equation}
        P(Y_t=1|X;\beta,\alpha)
        =
        H(X_{t}^{\prime}\beta+\alpha), \; 
        t = 1,\cdots,T,
        \label{eq:static_BC_CE}
\end{equation}
with $\alpha \in \mathcal A = \mathbb R$, 
and the sequence of regressors
$X = (X_{1},\cdots,X_{T}) \in \mathcal X$.
If $H(\cdot) = \Lambda(\cdot)$ is the standard logistic distribution, this is the static panel logit model
\citet{raschStudiesMathematicalPsychology1960, rasch1961general}.
Setting $H(\cdot)=\Phi(\cdot)$ obtains the probit version, etc.

Conditional serial independence of $V_t$ implies that $(Y_{1},\cdots,Y_{T})$ are independent conditional on $(\alpha,X)$.
For any $y = (y_1,\cdots,y_T) \in \{0,1\}^T$, we therefore have
\begin{equation}
f_{Y|X,\alpha}(y|X;\beta,\alpha) 
= 
\prod_{t=1}^T 
    (H(X_{t}^{\prime}\beta+\alpha))^{y_t} 
    (1-H(X_{t}^{\prime}\beta+\alpha))^{1-y_t},
\label{eq:static_binary_model_probabilities}
\end{equation}
so that the model fits the framework of Section \ref{sec:nonlinearpanels}. 
Assumptions \ref{A:momentdiscrete}(i) and (ii) are trivially satisfied for any $H(\cdot)$.
It follows that equation  \eqref{eq:max_for_all_phi_panel} in Proposition \ref{cor:panel_ID_beta} characterizes the identified set for $\beta$.

We now describe results from a numerical experiment, starting from the special case of $T=2$, $X_1 = 0,\; X_2 = 1$.\footnote{A full description of the computation is in Appendix \ref{app:computation_general}, with details about this specific model in \ref{app:computation_parametric_binary_T2}.}

Figure \ref{fig:binary_choice_01_probit} plots $T(\theta)$ for the probit model with $\beta_0 = 1$, with $\alpha$ on a grid of $101$ equally spaced points from -5 to 5, and with $P(\alpha = a) \propto \exp(-a^2/2)$. 
We consider a grid for $\beta$ of 601 equally spaced points from 0.7 to 1.3.
On a single core i7-11370H at 3.30GHz, it takes 0.0036 seconds to compute $T(\theta)$ per value of $\theta$. For more details on computation times, see Section \ref{app:computation_general}.
For the probit model, $\beta_0$ is not point-identified, and we find $\Theta_I = [0.968,1.065]$.

Figure \ref{fig:binary_choice_01_all_distributions} presents $\Theta_I$ for five different choices of $H(\cdot)$. ``Logistic'' corresponds to the standard logit model, for which we recover the well-known result that $\beta_0$ is point-identified. ``Normal'' refers to the result in Figure \ref{fig:binary_choice_01_probit}. ``Uniform'' sets $H(\cdot)$ equal to the cumulative distribution function of the continuous uniform distribution on $[-2,2]$. ``Cauchy'' sets it equal to the standard Cauchy distribution, and ``Truncated normal'' sets it equal to a standard normal distribution truncated to $[-2,3]$. 

The width of these identified sets is not directly comparable across distributions, due to their different scales.
Nonetheless, the identified set is rather small across all distributions, especially given the limited variation in $X$. 
This is in line with previous findings on the size of identified sets in parametric nonlinear panels, see \citet{honoreBoundsParametersPanel2006,ChernValHahnNewey2013}.

\begin{figure}
    \centering
    \begin{subfigure}{0.5\textwidth}
        \centering
        \includegraphics[width=\textwidth]{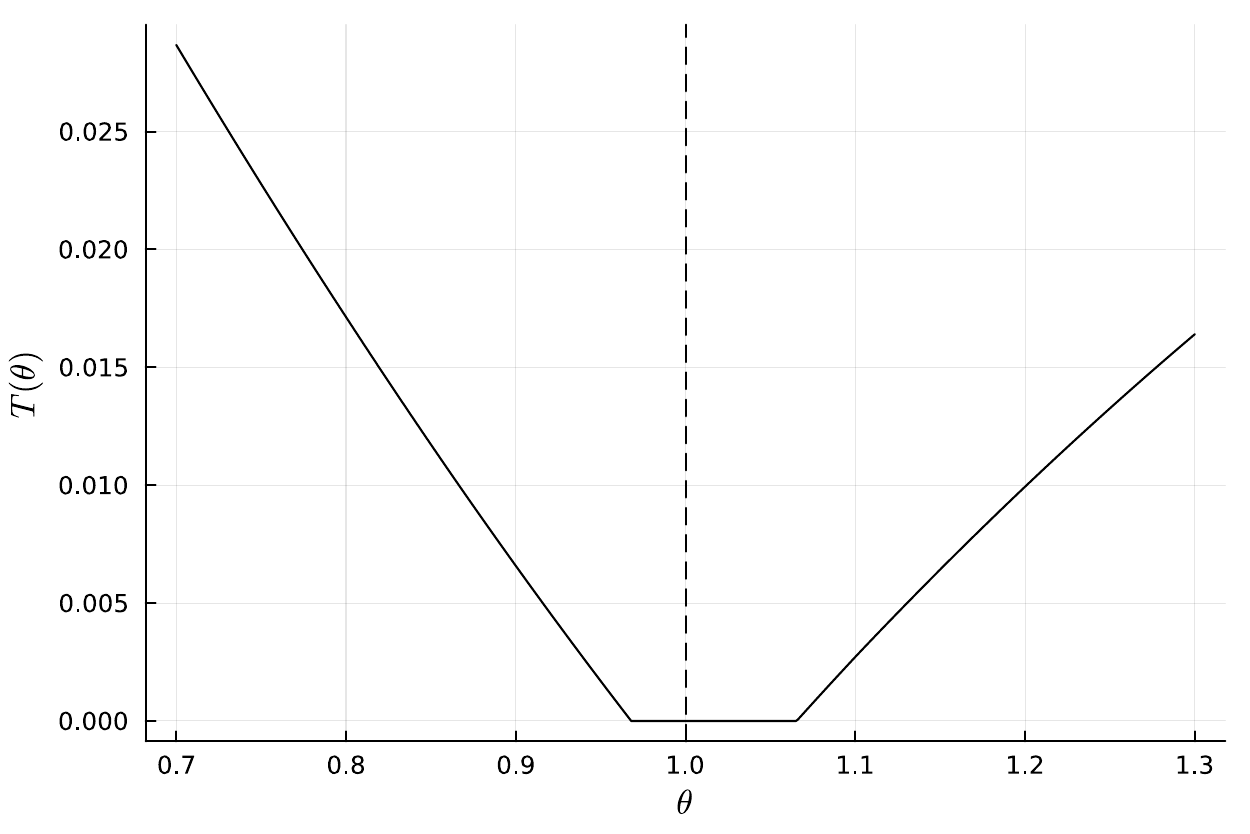} 
        \caption{$T(\theta)$ for the probit model.}
        \label{fig:binary_choice_01_probit}
    \end{subfigure}\hfill
    \begin{subfigure}{0.5\textwidth}
        \centering
        \includegraphics[width=\textwidth]{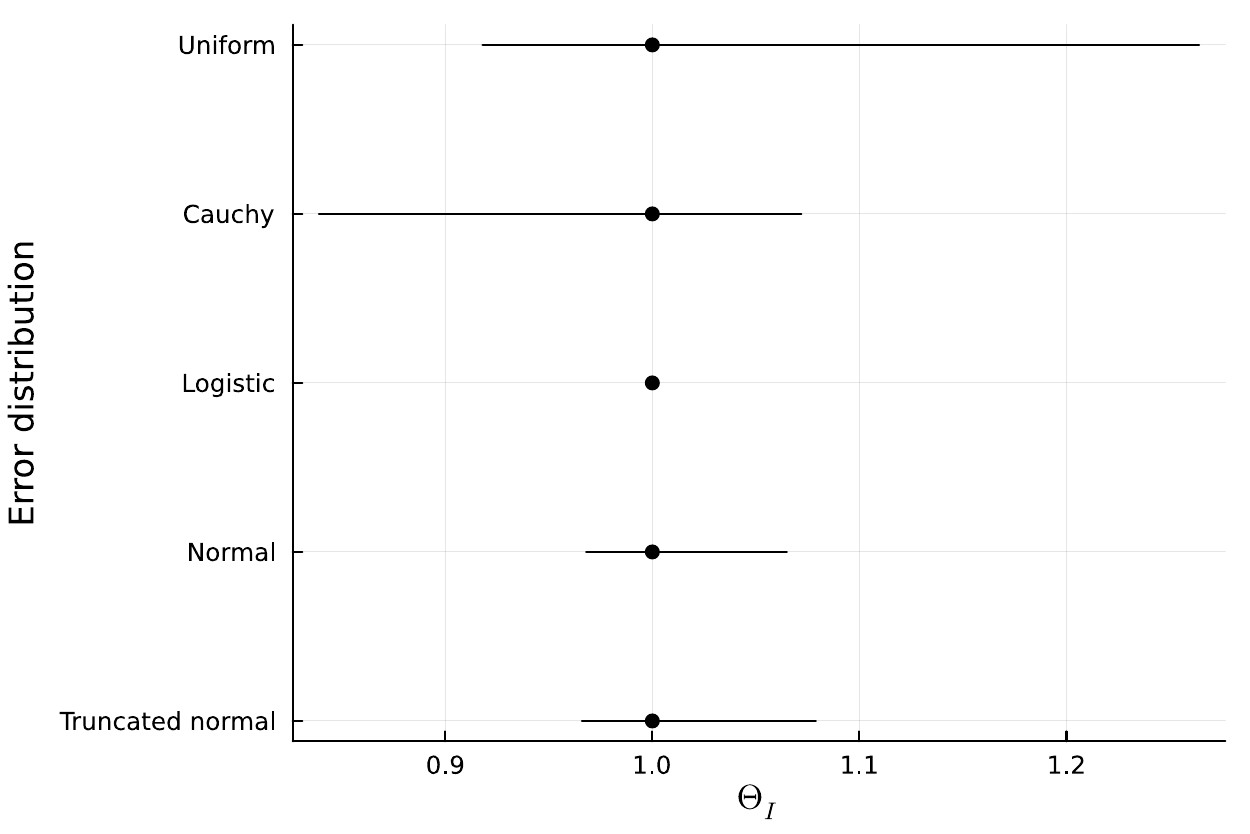} 
        \caption{$\Theta_I$ for various error distributions.}
        \label{fig:binary_choice_01_all_distributions}
    \end{subfigure}
    \caption{Identified sets for the static binary choice model with $X = (0,1)$, $\beta_0 = 1$.}
    \label{fig:binary_choice_01}
\end{figure}

Our second result is for the regression coefficient and average treatment effects
for the probit model with $T \in \{2,3\}$ and binary regressors $\mathcal X = \{0,1\}^T$.
Similarly to the numerical experiments in \citet{ChernValHahnNewey2013}, \S 8,
we set $P(X = x) = 2^{-T}$ for all $x$,
and the support and distribution of the fixed effects, independent of $X$.
Finally, $\beta_0 \in \left\{0, 0.1, \cdots, 1.9, 2\right\}$ and $\beta \in \{\beta_0 - 0.3, \beta_0 - 0.299, \cdots, \beta_0 + 0.3\}$.
Appendix \ref{app:computation_parametric_binary_T2} describes how to modify the computation from the baseline case above to this more general setup.

Figure \ref{fig:CFHN_beta} presents the identified sets as a function of $\beta_0$, for $T = 2$ and $T=3$.
\footnote{These results recover those in the top left panel in Figures 2 and 3 in \citet{ChernValHahnNewey2013}.}
Unless $\beta_0 = 0$, the regression coefficient is partially identified. The identified set is small for $T=2$, much smaller for $T=3$, and indistinguishable from the 45-degree line when $T=4$ (not reported).
The figures are based on 12621 values of $(\beta, \beta_0)$. The computation time per value is 0.013 seconds for $T=2$, and 0.066 seconds for $T=3$ on a single core i7-11370H at 3.30GHz.

\begin{figure}
    \centering
    \begin{subfigure}{0.5\textwidth}
        \centering
        \includegraphics[width=\textwidth]{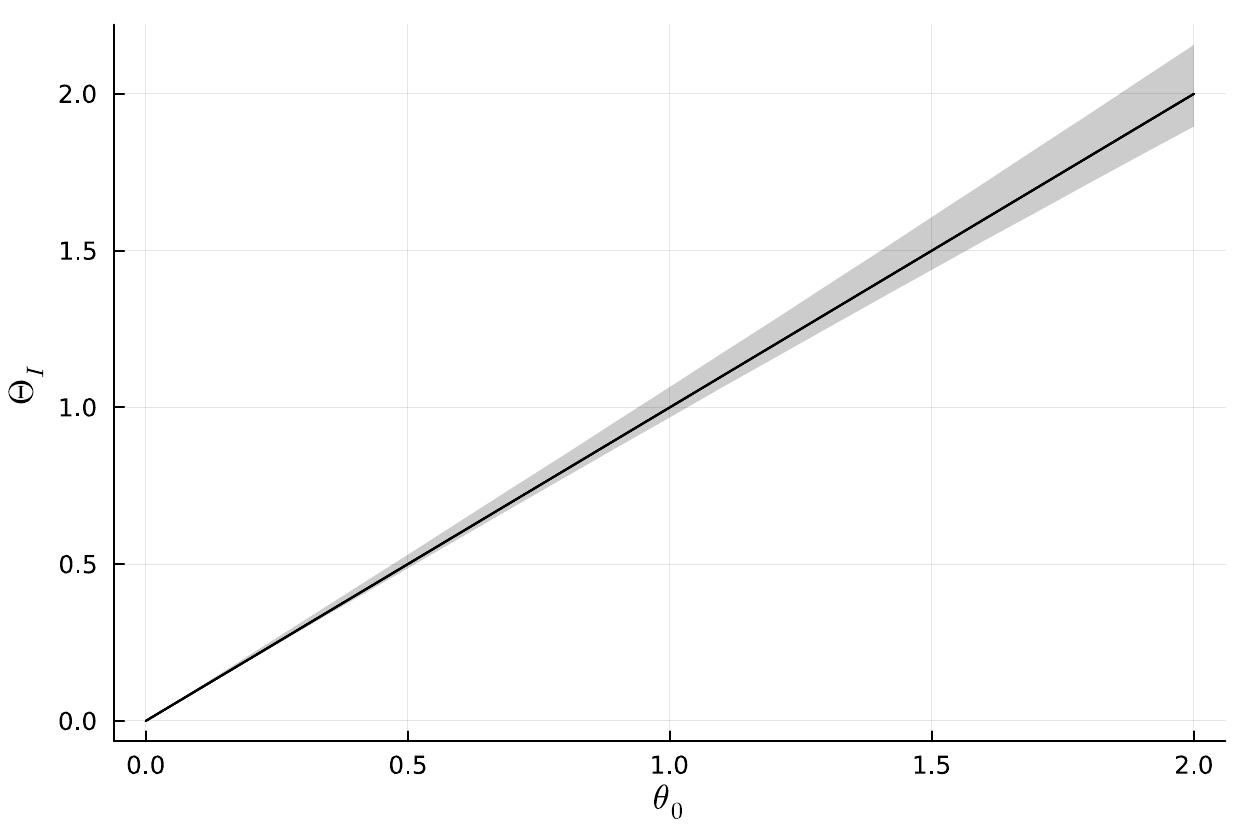} 
        \caption{$T=2$}
        \label{fig:CFHN_beta_T2}
    \end{subfigure}\hfill
    \begin{subfigure}{0.5\textwidth}
        \centering
        \includegraphics[width=\textwidth]{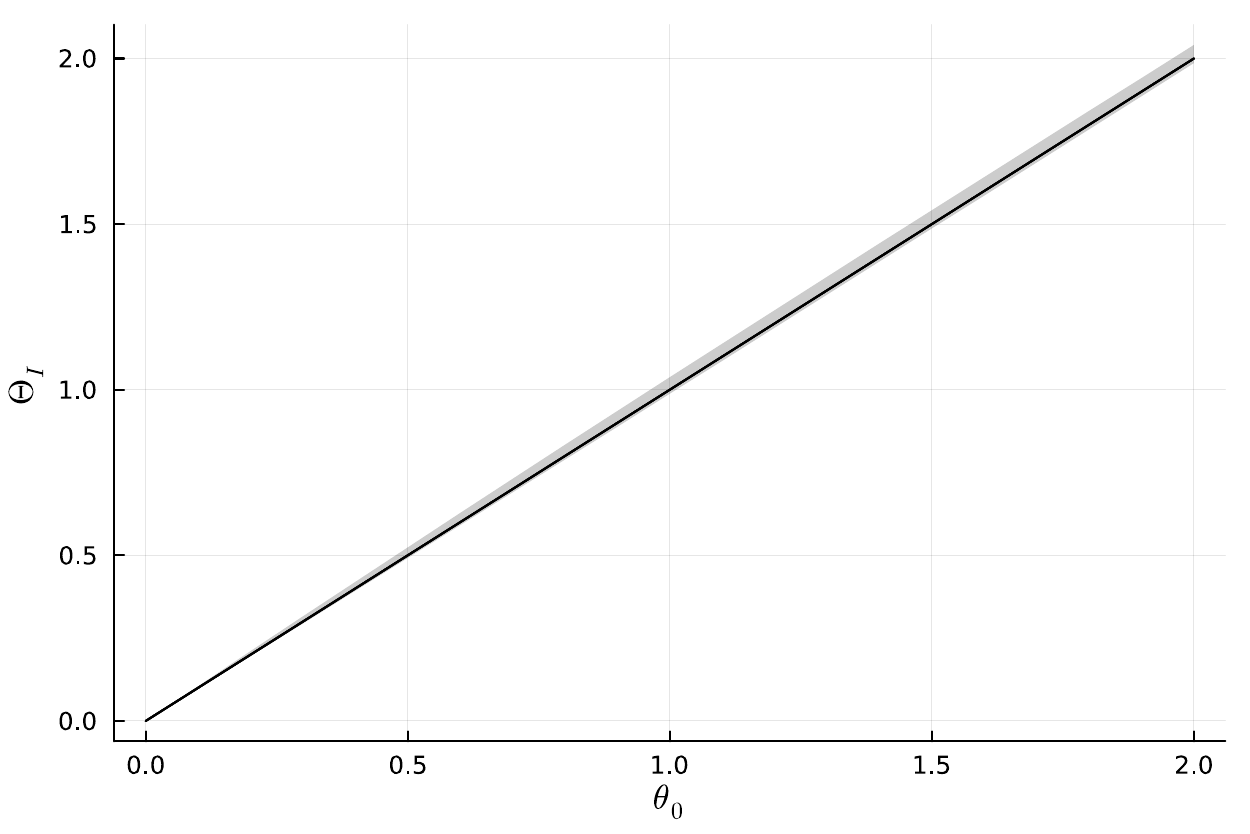} 
        \caption{$T=3$}
        \label{fig:CFHN_beta_T3}
    \end{subfigure}
    \caption{Identified sets for regression coefficient in static binary choice probit.}
    \label{fig:CFHN_beta}
\end{figure}

Next, we compute the identified set for the average treatment effect of moving a randomly selected individual's $x_t$ from $0$ to $1$, i.e.
\[
\text{ATE}(0, 1;\beta) 
    = 
    E[H(\alpha + \beta) - H(\alpha)].
\]
Proposition \ref{cor:counterfactual_ID} applies, and \eqref{E:ps133} gives an expression for the identified set.
Figure \ref{fig:CFHN_PE} presents the identified sets as a function of $\beta_0$, for $T = 2$ and $T=3$.
\footnote{These results recover those in the bottom right panel in Figures 2 and 3 in \citet{ChernValHahnNewey2013}.}
The size of the identified sets is greatly reduced when moving from $T=2$ to $T=3$.

\begin{figure}
    \centering
    \begin{subfigure}{0.5\textwidth}
        \centering
        \includegraphics[width=\textwidth]{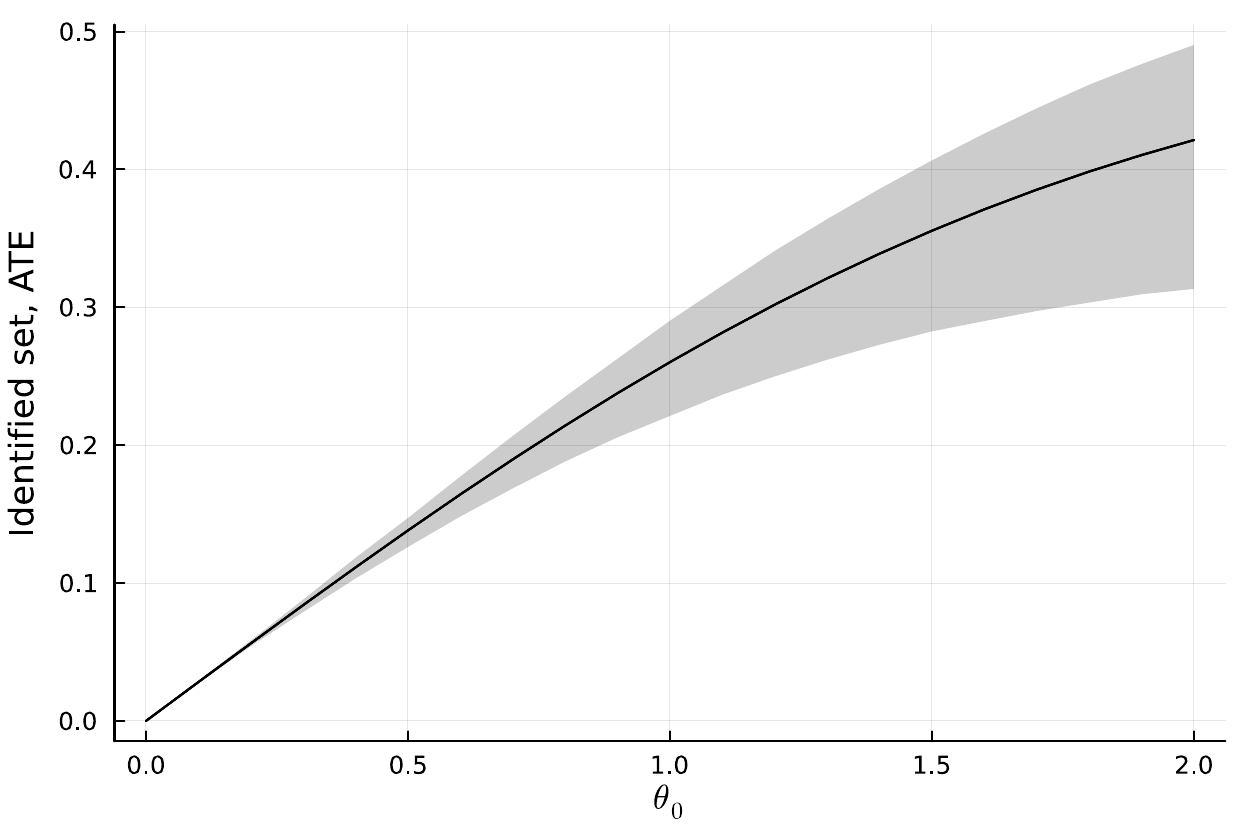} 
        \caption{$T=2$}
        \label{fig:CFHN_PE_T2}
    \end{subfigure}\hfill
    \begin{subfigure}{0.5\textwidth}
        \centering
        \includegraphics[width=\textwidth]{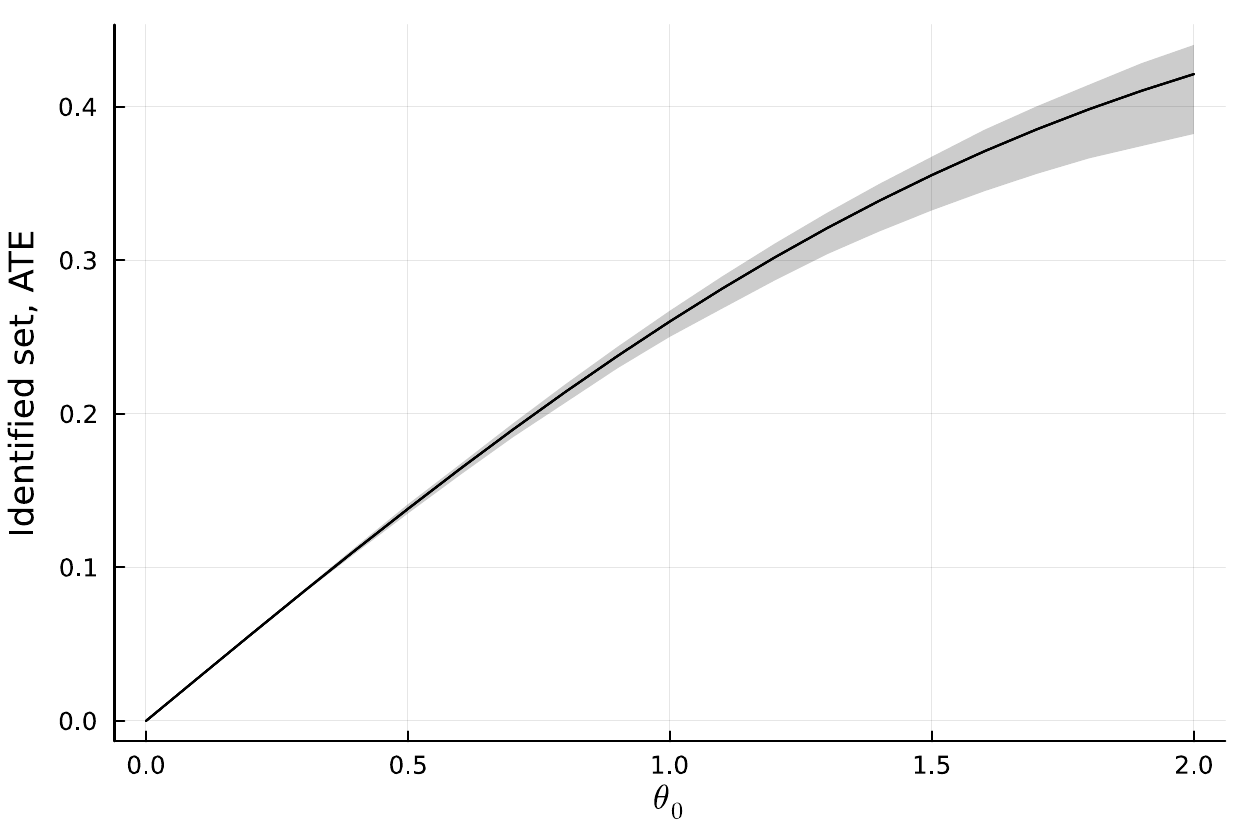} 
        \caption{$T=3$}
        \label{fig:CFHN_PE_T3}
    \end{subfigure}
    \caption{Identified sets for ATE in static binary choice probit.}
    \label{fig:CFHN_PE}
\end{figure}

\begin{appendix}

\section{Additional remarks}\label{sec:additional_remark}

Throughout the Appendices, we let $\Theta^0_\mathrm{I}$ denote the set: 
\begin{align}
\label{convtheta0}
    \Theta^0_{\mathrm{I}} &\equiv \{\theta \in \Theta: \exists \gamma\in\Gamma \text{ such that } \mu^*_{Z} = \mu_{Z,(\theta,\gamma)} \text{ a.e. } Z\},
\end{align}
which is closely related to the definition of the identified set for $\theta$ in \citet{honoreBoundsParametersPanel2006}. This set can be equivalently expressed as:
\begin{align}
\label{Theta0I}
    \Theta^0_\mathrm{I} = \{\theta \in \Theta: \mu_Z^* \in \mathcal{M}_{\theta}\}.
\end{align}
\label{IDnormclosure}

\begin{remark}[The identified set and the closure of the set of model probabilities]
\label{IDsetandclosure}
The identified set \(\Theta_{\mathrm{I}}\) in \eqref{E:essentialIDset} is defined using the closure \(\overline{\mathcal{M}}_\theta\). We consider \(\Theta_{\mathrm{I}}\) instead of \(\Thetai\) in \eqref{convtheta0} or \eqref{Theta0I} for two reasons. First, we avoid issues related to infeasible inference. In Remark \ref{rem:essentiallyindistinguishable}, we show that elements of \(\Theta_{\mathrm{I}}\) are \textit{indistinguishable} from those of \(\Thetai\), making it impossible to construct a hypothesis test that can differentiate between \(\Theta_{\mathrm{I}}\) and \(\Theta^0_{\mathrm{I}}\) with power exceeding size.

Second, the need to define the identified set in terms of \(\overline{\mathcal{M}}_\theta\) also arises in nonlinear panel models. For instance, allowing the support of the fixed effects to include $\{-\infty,\infty\}$ is crucial, as discussed in \citet{Chamberlain2010}. Without finiteness assumptions on the support, the model probability lies in \(\overline{\mathcal{M}}_\theta\) but not in \(\mathcal{M}_\theta\). Thus, the identified set must include values of \(\theta\) consistent with \(\mu^*_Z \in \overline{\mathcal{M}}_\theta\). Under finiteness assumptions, such as in \citet{honoreBoundsParametersPanel2006}, there is no distinction between \(\mathcal{M}_\theta\) and \(\overline{\mathcal{M}}_\theta\), implying \(\Theta^0_{\mathrm{I}} = \Theta_{\mathrm{I}}\).

\citet{schennachEntropicLatentVariable2014} also considers a refinement of the conventional notion of the identified set within her framework by defining it in terms of the closure of the set of all feasible values for the moment function. This allows for the support of the unobserved heterogeneity in her framework to be unbounded, see also \citet{li2019identification}.

\end{remark}

\begin{remark}[Norm choice for defining the closure]
\label{choiceofclosure}
    The choice of norm defining the closure is important. In the literature on impossible inference, the total variation (TV) and the L\'evy-Prokhorov (LPv) norms play important roles, see e.g., \citet{BERTANHA2020247}. Naturally, the choice of norm impacts the resulting identified set. For example, our identified set $\Thetaw$ involves the TV norm, and includes values of $\theta$ such that $\mu^*_Z \in \overline{\mathcal{M}}_\theta$. Our main result in Theorem \ref{P:main} establishes that $\Theta_\mathrm{I}$ includes values of $\theta$ that cannot be distinguished from those in $\Thetai$ with any bounded function $\phi:\mathcal{Z}\to [0,1]$. 

    Instead, consider the closure of $\mathcal{M}_\theta$ with respect to the topology induced by the LPv norm, denoted by $\overline{\mathcal{M}}^\mathrm{LP}_\theta$, and define $\Theta^\mathrm{LP}_\mathrm{I} \equiv \{\theta \in \Theta: \mu_Z^* \in \overline{\mathcal{M}}^\mathrm{LP}_\theta\}$. It is possible to show that $\Theta^\mathrm{LP}_\mathrm{I}$ includes values of $\theta$ that cannot be distinguished from those in $\Thetai$ with any \textit{continuous} and bounded function $\phi_c:\mathcal{Z}\to [0,1]$. 
    Since $\overline{\mathcal{M}}_\theta \subseteq \overline{\mathcal{M}}^\mathrm{LP}_\theta$, it follows that $\Theta_\mathrm{I} \subseteq \Theta^\mathrm{LP}_\mathrm{I}$. Under additional assumptions, these sets can be made equal. For example, when $\mathcal{Z}$ is discrete, $\Theta_\mathrm{I} = \Theta^\mathrm{LP}_\mathrm{I}$ indicating robustness to the choice of norm. See \citet{BERTANHA2020247} for alternative assumptions on $\mu_{Z,(\theta,\gamma)}$ that lead to this equality result.
\end{remark}

\begin{remark}[Necessity of Assumption \ref{A:one}]
\label{A2.2_necessary} 
    Assumption \ref{A:one} is essential for Theorem \ref{P:main}. Consider, for example, the case where \(\mu^*_Z\) is the Lebesgue measure on \([0,1]\) and, for some \(\theta \in \Theta\), \(\mathcal{M}_\theta = \mathrm{co}(\{\delta_z: z \in [0,1]\})\), which does not satisfy Assumption \ref{A:one}. For all bounded Borel maps \(\phi: [0,1] \to [0,1]\), we have \(\mathbb{E}_{\mu^*_Z}[\phi] \leq \sup_{z \in [0,1]} \phi(z) = \sup_{\mu \in \mathcal{M}_\theta} \mathbb{E}_{\mu}[\phi]\), implying \(\theta \in \Theta^0_\mathrm{I}\). However, for every \(\mu \in \sum_{i=1}^k \alpha_i \delta_{z_i} \in \mathcal{M}_\theta\), \(d_{\mathrm{TV}}(\mu^*_Z, \mu) = 1\) since \(\mu^*_Z(\{z_1, \ldots, z_k\}) = 0\) and \(\mu(\{z_1, \ldots, z_k\}) = 1\). Therefore, \(\mu^*_Z \notin \overline{\mathcal{M}}_\theta\) and \(\theta \notin \Theta_\mathrm{I}\).
\end{remark}

\begin{remark}[The discrepancy function]
\label{discrepacycomparisons}
    The discrepancy function $T(\theta)$ in \eqref{def:T_theta} is new to the literature. Computing the zeros of this function involves a search over two infinite-dimensional spaces: the space of features $\Phi_b(\mathcal Z)$ \textit{and} the space of model probabilities $\mathcal M_\theta$:
    \begin{align}
    T(\theta) 
    \equiv 
    \sup_{\phi \in \Phi_b(\mathcal Z)} 
        \left(
            \EE{\mu_Z^*}{\phi} - \sup_{\mu \in \mathcal{M}_\theta}\EE{\mu}{\phi}
        \right).
\end{align}
    
    The function $T(\theta)$ can be seen as a generalization of the Maximum Mean Discrepancy (MMD) measure defined in \eqref{MMD} below. The MMD is an integral probability metric that has been used in the machine learning literature to discriminate between two observed probability measures, see, e.g., \citet{Gretton2007,Gretton2012}. To see this, suppose that $\mathcal M_{\theta} = \{\overline{\mu}\}$ is a singleton (be it because the critic knows or has identified $\overline{\mu}$, the defender plays a fixed strategy $\overline{\mu}$, etc). Then the critic selects $\phi$ that solves: 
\begin{align}
\label{MMD}
    \textrm{MMD}(\mu^*_Z,\overline{\mu}) & \equiv \sup_{\phi\in\Phi_b(\mathcal{Z})}\left(\EE{\mu_Z^*}{\phi} - \EE{\overline{\mu}}{\phi}\right) \text{ for known } \mu^*_Z \text{ and } \overline{\mu}.
\end{align}

In our framework $\mathcal{M}_\theta$ is not a singleton. Consequently, \(T(\theta)\) can be seen as a generalization of the MMD to situations where $\overline{\mu}$ in \eqref{MMD} is not known. 

In \cite{schennachEntropicLatentVariable2014} and \cite{li2019identification}, see also \cite{guDualApproachWassersteinRobust2023} -- and more generally, approaches based on support-functions, e.g., \citet{beresteanuSharpIdentificationRegions2011} and \citet{ChesherRosenZhang}, the criterion function that characterizes the identified set depends on a finite dimensional moment function $g\in\mathbb{R}^{d_g},\; d_g < \infty$, and it is given by:
\begin{align}
\label{T_SL}
    T_{\mathrm{SF}}(\theta) &\equiv \sup_{\eta\in\mathbb{R}^{d_g}:\vert\vert\eta\vert\vert=1}\mathbb{E}_{\mu^*_Z}[\inf_{u\in\mathcal{U}}\eta^\prime g\left(u,Z,\theta\right)],
\end{align}
where $\mathcal U$ is the support of the latent random variables.

Supposing that $\mu$ is absolutely continuous with respect to $\mu^*_Z$, our criterion function $T(\theta)$ can be written as
\begin{equation}
\label{T_AC}
    T(\theta) = \sup_{\phi\in\Phi_b(\mathcal Z): \phi\in [0,1]}\inf_{\mu\in\mathcal M_\theta}\mathbb{E}_{\mu^*_Z}\left[\phi\left(\frac{d\mu}{d\mu^*_Z}-1\right)\right].
\end{equation}

A comparison of \eqref{T_SL} and \eqref{T_AC} reveals several differences: (i) the supremum in \eqref{T_SL} is over a finite-dimensional vector, whereas in \eqref{T_AC}, it is over a space of functions; (ii) the infimum in \eqref{T_SL} is taken inside the expectation, while in \eqref{T_AC}, it is taken outside; and (iii) the infimum in \eqref{T_SL} is over the support of unobserved random variables, whereas in \eqref{T_AC}, it is over a set of probability measures. Additionally, note that $T_{\mathrm{SF}}(\theta) = 0$ at $\theta = \theta_0$ (the true parameter value), whereas $T(\theta) = 0$ when $\mu^*_Z = \mu$ (the observed measure matches the model probability measure). These distinctions reflect fundamental differences in the approaches: an approach based on support functions versus one based on set membership for observed probability measures.

Essentially, the difference in approaches is driven by the space in which they operate. Existing approaches operate in the space of random variables defined by (a finite number of) moment conditions, using hyperplanes to test specific linear combinations of these moments via support functions. In contrast, our framework operates in the infinite-dimensional space of probability measures over observables, where hyperplanes separate the true probability measure $\mu^*_Z$ from the set $\overline{\mathcal{M}}_\theta$. The hyperplanes in our framework correspond to bounded functions in $\Phi_b(\mathcal{Z})$.

All approaches are grounded in convex analysis and leverage the fundamental principle of separating hyperplanes. For all approaches, convexity of the feasible set - whether in the moment space or the probability measure space - is essential for ensuring the applicability of separation results.

Finally, note that the dual formulation of optimal transport resembles the MMD problem, despite arising from different principles. In the dual formulation, the separating hyperplane can be thought of as operating in the space of potential functions. These potential functions encode the cost structure and enforce the feasibility of the transport plan, often subject to smoothness assumptions and additional constraints dictated by the cost function. Moreover, as in the MMD problem, the dual formulation of optimal transport does not involve minimization over the model probability measure.

\end{remark}

\section{Proofs}
\label{sec:proofs}

Our main result, Theorem \ref{P:main}, is an implication of the more general result in Proposition \ref{P:main_appendix} stated below. We first state and prove that more general result. The proofs of other results in the main text follow. Some supporting lemmas, and the proofs for the results in Section \ref{sec:application_to_panels}, are in the supplementary material.

The result here uses $\Thetai$ defined in \eqref{Theta0I}, $\Theta_\mathrm{I}$ defined in \eqref{E:essentialIDset}, and $\Thetam$ defined in \eqref{E:MI}. 
We denote by $\mu^n$ the product measure on $\mathcal{D}^n$ with the corresponding product topology. This is the joint distribution of $n$ points drawn independently from distribution $\mu$.

\begin{proposition} 
\label{P:main_appendix}
For $n \in \N$ let $\M_\theta^{(n)} = \{\mu^n: \mu \in \M_\theta\}$. Consider the following statements:

\begin{enumerate}
    \item[I.] $\theta \in \Theta_{\mathrm{I}}$.
    \item[II.] For all $n \in \N$ and all $f: \mathfrak{B}(\mathcal{Z})^n \ra \R$ continuous with respect to the total variation norm, $f((\mu_Z^*)^n) \in \overline{f(\M_\theta^{(n)})}$.
    \setcounter{enumi}{2}
    \item[III.] $\theta \in \Thetam$.
    \item[IV.] For all compactly supported Borel $\phi: \mathcal{Z} \ra [0,1]$, 
    \begin{align} 
    \label{E:lowerb}
        \EE{\mu^*_Z}{\phi} \le \sup_{\mu \in \M_\theta} \EE{\mu}{\phi}.
    \end{align}
\end{enumerate}

 Then the following are true:

\noindent(a) Statements I and II are equivalent.

\noindent(b) Statements III and IV are equivalent, and implied by statements I and II.

\noindent(c) Let Assumptions \ref{A:zero} and \ref{A:one} hold, and suppose that $\overline{\mathcal{M}}_\theta$ is convex. Then, statements I, II, III, and IV are equivalent. This equivalence holds for every $\mu^*_Z \in \mathcal{P}(\mathcal{Z})$ only if $\overline{\mathcal{M}}_\theta$ is convex.
\end{proposition}

\begin{proof}
    See page \pageref{proofProposition3}.
\end{proof}

\begin{remark}
\label{rem:essentiallyindistinguishable}
The equivalence of I.\ and II.\ in Proposition \ref{P:main_appendix} implies that points \(\theta \in \Thetaw\) are statistically indistinguishable from those in \(\Thetai\). Specifically, if \(\theta \in \Thetaw\), any value attained by a function \(f((\mu^*_Z)^n)\), where \((\mu^*_Z)^n\) is the distribution of \(n\) i.i.d. observations from \(\mu^*_Z\) and \(f\) is continuous in the TV norm, must be a limit point of the image of \(\mathcal{M}_\theta^{(n)}\) under the same function. This makes it impossible to test the hypotheses:
\begin{align*}
    &H_0: \theta \in \Thetai, \\
    &H_1: \theta \notin \Thetai,
\end{align*} 
with nontrivial power.

To illustrate, suppose such a test \(\phi: \mathcal{Z} \to [0,1]\) exists with uniform size \(\alpha\) over \(\mathcal{M}_\theta\), the set of distributions of \(Z\) consistent with \(\theta\):
\begin{align*}
    \sup_{\mu^n \in (\mathcal{M}_\theta)^n} \int \phi(z_1, \ldots, z_n) \, d\mu^n \le \alpha.
\end{align*}
Since the map \(\mu^n \mapsto \int \phi(z_1, \ldots, z_n) \, d\mu^n\) is continuous in the total variation norm, Proposition \ref{P:main_appendix} implies that the test's power to reject \(H_0\) when \(\mu^*_Z\) is the distribution of observables for \(\theta \in \Theta^0_\mathrm{I}\) cannot exceed \(\alpha\):
\begin{align*}
    \int \phi \, d(\mu^*_Z)^n \le \alpha.
\end{align*}

\end{remark}

\begin{proof}[Proof of Proposition \ref{P:main_appendix}]
\label{proofProposition3}
The proof of statements (a) and (b) is routine:
    \begin{enumerate}
    \item[I$\Rightarrow$II] If $\theta \in \Thetaw$, then for all $\ve > 0$ there is some $\mu \in \mathcal{M}_\theta$ which has $d_{\mathrm{TV}}(\mu^*_Z , \mu) < \ve$. For each such $\ve$ and $\mu$, $\mu^*_Z$ and $\mu$ have density with respect to the measure $\lambda \equiv \mu^*_Z + \mu$. Hence, by the arguments below, one has $d_{\mathrm{TV}}((\mu^*_Z)^n, \mu) < n \ve$ for $\ve $ arbitrary, and thus $(\mu^*_Z)^n \in \overline{\M_\theta^{(n)}}$. By a straightforward topological argument, $f(\overline{\M_\theta^{(n)}}) \subseteq \overline{f(\M_\theta^{(n)})}$, so II is implied by I.
    \item[II$\Rightarrow$I] It suffices to consider the case $n = 1$ and the function $f: \mu' \mapsto \inf_{\mu \in \M_\theta} \norm{\mu' - \mu}_{\mathrm{TV}}$, which is continuous by the triangle inequality.
    \item[III$\Rightarrow$IV] This is straightforward from the definition of $\Thetam$, and by invariance of the upper bound in \eqref{E:lowerb} with respect to translations of $\phi$ and rescalings by positive constants.
    \item[IV$\Rightarrow$III] We prove the contrapositive. If $\theta \not\in \Thetam$, there is some bounded Borel $\phi$ for which $\EE{\mu^*_Z}{\phi} > \sup_{\mu \in \mathcal{M}_\theta} \EE{\mu}{\phi}$. By taking a linear transformation, it may be assumed without loss of generality that $\phi$ is positive and bounded above by $1$. By the tightness of Polish Borel measures (\citet{VW1996}, Lemma 1.3.2), there is some compact $K \subseteq \mathcal{Z}$ for which 
     $\EE{\mu^*_Z}{\phi \cdot \one_K} >  \sup_{\mu \in \mathcal{M}_\theta} \EE{\mu}{\phi} \ge \sup_{\mu \in \mathcal{M}_\theta} \EE{\mu}{\phi \cdot \one_{K}}$, so that IV.\ cannot be true.
     \item[II$\Rightarrow$III] For any bounded and Borel $\phi$, the map $\mu \mapsto \EE{\mu}{\phi}$ is continuous with respect to the total variation norm. Thus, \eqref{E:lowerb} follows from II.
    \end{enumerate}

Now, make Assumptions \ref{A:zero} and \ref{A:one}, so that the measures $\mu \in \M_\theta$ and $\mu^*_Z$ are continuous (and have densities with respect to) the $\sigma$-finite measure $\lambda \equiv \mu^*_Z + \lambda_\theta$. Let $f[\mu]$ denote the mapping sending $\mu$ to its density with respect to $\lambda$. By Scheff\'{e}'s Lemma (\citet{tsybakov2008introduction}, Lemma 2.1), $\norm{\mu^n-\phi^n}_{\mathrm{TV}} = \frac{1}{2} \norm{f[\mu^n] - f[\phi^n]}_{L^1(\lambda^n)}$ (so that $f$ is a linear isometry with respect to the total variation norm and $L^1$ topology), and by Lemma B.8 of \citet{Gh2017}, this quantity is bounded above by $n\norm{\mu-\phi}_{\mathrm{TV}}$. Moreover, it follows that all measures $\mu$ and $\mu^*_Z$ invoked in this proof have densities in the Banach space $L^1(\lambda)$, which has as its dual $L^\infty(\lambda)$ (\citet{fremlin2000measure}, Theorem 243G). 
    
Suppose now that $\overline{\M}_\theta$ is convex. Then, a straightforward argument shows that $\overline{\M}_\theta = \overline{\mathrm{co}}(\M_\theta)$. The Hahn-Banach theorem (\citet{C1994}) and the observations above imply that 
\begin{align*}
    \overline{\M}_\theta &= \overline{\mathrm{co}}(\M_\theta) \\
    & = \{\mu'\text{ continuous wrt }\lambda_\theta: \EE{\mu'}{\phi} \le \sup_{\mu \in \overline{\M}_\theta} \EE{\mu}{\phi} \,\forall \text{ bounded, Borel }\phi\}. 
\end{align*}
Suppose that $\theta \in \Thetam$. Because we can take $\phi$ to be the indicator function of any Borel set in \eqref{E:lowerb}, $\mu^*_Z$ must be continuous with respect to $\lambda_\theta$, and it follows from the previous display that $\mu^*_Z \in \overline{\M}_\theta$ and $\theta$ is in the identified set. On the other hand, suppose that for all Borel $\mu^*_Z$, $\theta \in \Thetam$ if and only if $\theta \in \Thetaw$. By the Hahn-Banach theorem, $\theta \in \Thetam$ if and only if $\mu^*_Z \in \overline{\mathrm{co}}(\M_\theta)$, and by definition, $\theta \in \Thetaw$ if and only if $\mu^*_Z \in \overline{\M}_\theta$. Thus, $\overline{\M}_\theta = \overline{\mathrm{co}}(\M_\theta)$, and the former set must be convex. 
\end{proof}

\begin{proof}[Proof of Proposition \ref{thm:convexity}]
\label{proof:thm_convexity}

First, Assumption \ref{A:probabilitymeasure} ensures that $\Gamma_\theta(\mathcal W)$ is convex. If there are no restrictions on $\gamma$, then $\Gamma_\theta(\mathcal{W}) = \mathcal{P}(\mathcal{W})$, which is inherently convex. When $\Gamma_\theta(\mathcal{W}) = \mathcal{P}(\mathcal{W})^g$, convexity is maintained since $\mathcal{P}(\mathcal{W})^g$ is the intersection of convex sets, each corresponding to a linear constraint imposed by the components of $g$. Then, by Corollary \ref{thm:pushforward_sharp}, $\mathcal{M}_\theta$ is convex.
By Theorem \ref{P:main}, $\theta$ is in $\Thetaw$ if and only if 
\begin{align}
\label{E:ps12_1}
    \EE{\mu^*_Z}{\phi} \le \sup_{\gamma \in \Gamma_\theta} \EE{\gamma}{\phi \circ \psi_\theta}, \text{ for all } \phi \in \Phi_b(\mathcal{Z}),
\end{align}  

(i) Suppose that $\Gamma_\theta = \mathcal{P}(\mathcal{W})$. Then the right hand side of \eqref{E:ps12_1} is equal to the right hand side of \eqref{IDset_W}, because $\mathcal{P}(\mathcal{W})$ contains all of the Dirac measures $\delta_w, w \in \mathcal{W}$. 

(ii) Suppose that $\Gamma_\theta = \mathcal{P}(\mathcal{W})^g$. Let $\Delta^*(d_g)$ be the subset of $\Gamma_\theta$ consisting of measures supported on at most $d_g+1$ points:
\begin{align*}
    \Delta^*(d_g) = \{\gamma \in \mathcal{P}(\mathcal{W})^g : \gamma = \sum_{j = 1}^{d_g + 1} c_j \delta_{w_j} : c_j \ge 0, w_i \in \mathcal{W}\}.  
\end{align*}
Let $\gamma \in \Gamma_\theta$, so that $\EE{\gamma}{g} =0$. By the equality constraint case of Theorems 2.1 and 3.1 of \citet{winkler88} (c.f.\ \citet{1907.07934}, Theorem 2.1), there exists some probability measure $\nu$ supported on $\Delta^*(d_g)$ with the property that $\gamma$ is in the barycenter of $\nu$ and
\begin{align*}
    \EE{\gamma}{\phi \circ \psi_\theta} = \int_{\Delta^*(d_g)}\EE{\gamma'}{\phi \circ \psi_\theta} \, \d \nu(\gamma')
\end{align*}
for all $\phi \in \Phi_b (\mathcal{Z})$. Hence, $\EE{\gamma}{\phi \circ \psi_\theta} \le \sup_{\gamma' \in \Delta^*(d_g)} \EE{\gamma'}{\phi \circ \psi_\theta}$. Taking the supremum over $\gamma \in \Gamma_\theta$ implies that $\sup_{\gamma \in \Gamma_\theta} \EE{\gamma}{\phi \circ \psi_\theta} = \sup_{\gamma' \in \Delta^*(d_g)} \EE{\gamma'}{\phi \circ \psi_\theta}$. This identity implies the equivalence of \eqref{E:ps12_1} and \eqref{IDset_restricted}. 
\end{proof}

\begin{proof}[Proof of Corollary \ref{corr1_prime}]
\label{proofcorr1_prime}
    As Assumption \ref{Asm:semip_sharpness} ensures that $\mathcal{Z} = \mathcal{Y} \times \mathcal{X}$ is a Polish space, we verify that Assumption \ref{A:one} holds, and that $\mathcal{M}_\theta$ is convex, allowing us to invoke Theorem \ref{P:main}. Assumption \ref{Asm:semip_sharpness} implies that every measure in $\mathcal{M}_\theta = (\psi_\theta)_* \Gamma_\theta (\mathcal{W})$ has a density with respect to the measure $\lambda_\theta$, which has the marginal distribution $\lambda_{\theta, \mathcal{X}}$ on $\mathcal{X}$ and conditional distributions $\lambda_{\theta, x}$ for all $x \in \mathcal{X}$. By the assumed measurability of the collection $\{\lambda_{\theta, x}: x \in \mathcal{X}\}$, this defines a measure on $\mathcal{Y} \times \mathcal{X}$ (see Section 10.4 in \citet{bogachev2007measure2}). Indeed, if $A \subseteq \mathcal{Y} \times \mathcal{X}$ is a $\lambda_\theta$-null set, then for any $\gamma \in \Gamma_\theta$ and $\mu = (\psi_\theta)_* \gamma$ in $\mathcal{M}_\theta$, one has
    \begin{align*}
        \EE{\mu}{\one_{(Y,X) \in A}} & = \EE{\mu}{ \EE{\mu}{\one_{(Y,X) \in A}|X} } = \int_{\mathcal{X}} \int_{\mathcal{Y}} \one_{(y,x) \in A} \, \d h(x,\cdot;\theta)_*\gamma_{U|x} \, \d \gamma_X \\
        & = \int_\mathcal{X} \int_\mathcal{Y} \one_{(y,x) \in A} \frac{\d h(x,\cdot;\theta)_*\gamma_{U|x}}{\d \lambda_{\theta, x}} \frac{\d \gamma_X}{\d \lambda_{\theta, \mathcal{X}}} \, \d \lambda_{\theta, x}\, \d \lambda_{\theta, \mathcal{X}} = 0. 
    \end{align*}
    Finally, the convexity of $\mathcal{M}_\theta$ follows from the convexity assumptions of Assumption \ref{Asm:semip_sharpness}, by Lemma \ref{L:convexreg1} in Section \ref{sec:convex_disintegration}, and by the linearity of the pushforward map $\psi_\theta$.
\end{proof}

\begin{proof}[Proof of Proposition \ref{cor:panel_ID_beta}]
This proof makes use of Lemma \ref{L:convexreg1} and Lemma \ref{L:marginalfree} that can be found in the supplementary material.
Convexity of $\mathcal{M}_\beta$ is a consequence of Assumption \ref{A:momentdiscrete}(ii) and Lemma \ref{L:convexreg1}. An application of Lemma \ref{L:marginalfree} with $g = 0$ implies that $\beta$ is in the identified set if and only if $\EE{\mu^*_{Y,X}}{\phi(Y,X)} \le \sup_{\mu \in \mathcal{M}_\beta}\EE{\mu}{\phi(Y,X)}$ for all $\phi \in \Phi_b(\mathcal{Y} \times \mathcal{X})$. Because the parameter $\gamma_{X,\alpha}$ is not constrained in any way in the definition of $\mathcal{M}_\beta$, one has 
\begin{align*}
    \sup_{\mu \in \mathcal{M}_\beta}\EE{\mu}{\phi(Y,X)} = \sup_{\substack{ x \in \mathcal{X} \\ a \in \mathcal{A} }} \int_{\mathcal{Y}} \phi(y,x) f_{Y|X,\alpha} (y|x,a ; \beta) \, \d \lambda_{\mathcal{Y}},
\end{align*}
which implies \eqref{eq:max_for_all_phi_panel}. 

\end{proof}

\begin{proof}[Proof of Proposition \ref{cor:counterfactual_ID}]
    This proof makes use of Lemma \ref{L:convexreg1} and Lemma \ref{L:marginalfree} that can be found in the supplementary material.
    Let $\mu_{Y,X} \in \mathcal{M}_\theta$ correspond to some $\gamma_{X,\alpha} \in \mathcal{P}(\mathcal{X} \times \mathcal{A})$ 
    and let $\phi \in \Phi_b(\mathcal{Y} \times \mathcal{X})$. Then, 
    \begin{align}
        \EE{\mu_{Y,X}}{\phi(Y,X)} = &\int_{\mathcal{X}} \int_{\mathcal{A}} \int_{\mathcal{Y}} \phi(y,x) f_{Y|X,\alpha}(y|x,a; \beta) \d \lambda_{\mathcal{Y}}\, \d \gamma_{\alpha|x} \, \d \gamma_X \text{, where} \nonumber\\
        &\int_{\mathcal{A}} g(x,a, \beta, \tau) \, \d \gamma_{\alpha|x} = 0\text{, } \gamma_X\text{-almost surely}. \label{E:ps131}
    \end{align}
    By modifying $\gamma$ on a set of measure $0$, we may assume that the second line of \eqref{E:ps131} holds for all $x$. It follows that 
    \begin{align}
        \EE{\mu_{Y,X}}{\phi(Y,X)} \le \sup_{\substack{x \in \mathcal{X}, \gamma_{\alpha|x} \\ \int_{\mathcal{A}} g(x,a, \beta, \tau) \, \d \gamma_{\alpha|x} = 0 }} \int_{\mathcal{A}} \int_{\mathcal{Y}} \phi(y,x) f_{Y|X,\alpha}(y|x,a; \beta) \d \lambda_{\mathcal{Y}}\, \d \gamma_{\alpha|x}.  \label{E:ps132}
    \end{align}
    By Theorems 2.1 and 3.1 of \citet{winkler88}, the right hand side of \eqref{E:ps132} is bounded above by the right hand side of \eqref{E:ps133}. On the other hand, for any particular $x \in \mathcal{X}$ and sequence of $c_j \in [0,1]$, $a_j \in \mathcal{A}$ meeting the constraints of \eqref{E:ps133}, the measure $\gamma_0 = \sum_{j = 1}^{d_g + 1} c_j \delta_{(x, a_j)}$ is in $\Gamma_{\tau}$, and letting $\mu_{Y,X}$ be the distribution associated with $\gamma_0$ 
    has the property that 
    \[
    \EE{\mu_{Y,X}}{\phi(Y,X)} = \sum_{j = 1}^{d_g + 1} c_j \int_{\mathcal{Y}} \phi(y,x) f_{Y|X,\alpha} (y|x,a_j; \beta) \, \d \lambda_{\mathcal{Y}}.
    \]
    Hence, $\sup_{\mu \in \mathcal{M}_{\theta}} \EE{\mu}{\phi(Y,X)}$ is precisely the right hand side of \eqref{E:ps133}, and Lemma \ref{L:marginalfree} concludes the proof.
\end{proof}

\end{appendix}


\bibliographystyle{ecta-fullname} 
\bibliography{svd,zotero,nonlinear-FE-fixed-T} 

\clearpage
\setcounter{page}{1}  

\beginsupplement

\section*{Supplementary Material}

\section{Convex disintegration lemma} 
\label{sec:convex_disintegration}

Let $W=(W_1,W_2)$ with $W_1 \in \mathcal {W}_1$ and $W_2 \in \mathcal {W}_2$ with $\mathcal{W}_1$, $\mathcal{W}_2$ Polish spaces. Let $\Gamma_{\theta, \mathcal{W}_1}(\mathcal{W}_1)\subseteq\mathcal{P}(\mathcal{W}_1)$, and for all $w_1 \in \mathcal{W}_1$ let $\Gamma_{\theta, w_1}(\mathcal{W}_2)\subseteq\mathcal{P}(\mathcal{W}_2)$. Additionally, let $\Gamma_\theta$ be the set of distributions of $W$ whose marginals over $\mathcal{W}_1$ belong to $\Gamma_{\theta, \mathcal{W}_1}(\mathcal{W}_1)$, and whose conditional distributions for $W_2|W_1 = w_1$ belong to $\Gamma_{\theta, w_1}(\mathcal{W}_2)$ for all $w_1$. 
Formally, by the disintegration theorem (Corollary 10.4.15 in \citet{bogachev2007measure2}),
$\Gamma_{\theta}(\mathcal{W})$ is the set of measures $\gamma \in \mathcal{P}(\mathcal{W})$ for which 
\begin{align} \label{E:disintegrationlemma}
    \d \gamma = \d \gamma_{W_1} \d \gamma_{W_2| w_1} \text{ for } \gamma_{W_1} \in \Gamma_{\theta, \mathcal{W}_1}(\mathcal{W}_1) \text{ and } \gamma_{W_2|w_1} \in \Gamma_{\theta, w_1}(\mathcal{W}_2)\text{ for all }w_1.
\end{align}

\begin{lemma}
\label{L:convexreg1}
    Suppose that $\Gamma_{\theta, \mathcal{W}_1}(\mathcal{W}_1)$ is convex and that $\Gamma_{\theta, w_1}(\mathcal{W}_2)$ is convex for every $w_1 \in \mathcal{W}_1$. Then, the set $\Gamma_\theta(\mathcal{W})$ of measures having disintegrations as in \eqref{E:disintegrationlemma} is a convex subset of $\mathcal{P}(\mathcal{W})$. 
\end{lemma}

\begin{proof}[Proof of Lemma \ref{L:convexreg1}]
\label{proofofLconvreg1}
    Let $\gamma^1$ and $\gamma^2$ be in $\Gamma_\theta$ and let $c \in (0,1)$ be arbitrary. For $\iota \in \{1,2\}$ (indeed, for every $\gamma \in \mathcal{P}(\mathcal{W})$), the measure $\gamma^\iota$ can be disintegrated into a marginal distribution $\gamma_{W_1}^\iota$ over $\mathcal{W}_1$ and a set of conditional distributions $\gamma_{W_2|w_1}^\iota$, $w_1 \in \mathcal{W}_1$, over $\mathcal{W}_2$. 
    By the Radon-Nikodym theorem, the measures $\gamma^\iota_{W_1}$ have nonnegative densities $f^\iota$ with respect to the dominating measure $\gamma_{W_1} = \gamma_{W_1}^1 + \gamma_{W_1}^2$. For all $w_1 \in \mathcal{W}_1$, define the Borel map $\rho(w_1) = \frac{\alpha f^1(w_1)}{\alpha f^1(w_1) + (1-\alpha) f^2(w_1)} \in [0,1]$ with the convention $\frac{0}{0} = 0$, and define $\gamma^\rho_{w_1} = \rho(w_1) \gamma^1_{W_2|w_1} + (1-\rho(w_1)) \gamma^2_{W_2|w_1}$ for all $w_1 \in \mathcal{W}_1$. Let $\gamma^c_{\mathcal{W}_1} = c \gamma_{\mathcal{W}_1}^1 + (1 - c) \gamma_{\mathcal{W}_1}^2$. By convexity, $\gamma^c_{\mathcal{W}_1} \in \Gamma_{\theta, \mathcal{W}_1}(\mathcal{W}_1)$, and $\gamma^\rho_{w_1} \in \Gamma_{\theta, w_1}(\mathcal{W}_2)$ for all $w_1$. 
    
    Let $\gamma^c$ denote the measure whose marginal distribution over $\mathcal{W}_1$ is $\gamma^c_{\mathcal{W}_1}$ and whose conditional distribution over $\mathcal{W}_2$ is $\gamma^\rho_{w_1}$, for all $w_1 \in \mathcal{W}_1$. By definition, $\gamma^c \in \Gamma_\theta(\mathcal{W})$. Moreover, for any Borel map $\phi: \mathcal{W}_1 \times \mathcal{W}_2 \ra \R$, one has 
    \begin{align*}
        \int_{\mathcal{W}_1 \times \mathcal{W}_2} \phi \, \d \gamma^c & = \int_{\mathcal{W}_1} \int_{\mathcal{W}_2} \phi \, \d \gamma^\rho_{w_1} \, \d \gamma^c_{\mathcal{W}_1} = \int_{\mathcal{W}_1} \int_{\mathcal{W}_2} \phi \, \d \gamma_{w_1}^\rho \, (c f^1 + (1-c) f^2) \, \d \gamma_{\mathcal{W}_1}\\
        & = \int_{\mathcal{W}_1} \int_{\mathcal{W}_2} \phi (c f^1  \gamma_x^1 + (1-c) f^2 \gamma_{w_1}^2) \, \d \gamma_{\mathcal{W}_1} = \int_{\mathcal{W}} \phi \, (c \d \gamma^1 + (1-c) \d \gamma^2 ),
    \end{align*}
    so that $\gamma^c = c  \gamma^1 + (1-c) \gamma^2 $. 
\end{proof}

\section{Lemmas: parametric restrictions}\label{supp:parametric}

This appendix contains some lemmas used in the proofs of the main results of Section \ref{sec:nonlinearpanels}.

\begin{lemma} \label{L:marginalfree}
    Let Assumption \ref{A:marginalcond} hold. 
     Then, $\mu^*_{Y,X} \in \overline{\mathcal{M}}_{\theta}$ if and only if 
    \begin{align*}
        \EE{\mu^*_{Y,X}}{\phi(Y,X)} \le \sup_{\mu \in \mathcal{M}_{\theta}} \EE{\mu}{\phi(Y,X)} \text{ for all compactly supported }\phi \in \Phi_b(\mathcal{Y} \times \mathcal{X}).
    \end{align*}
\end{lemma}
\begin{proof}
    We prove the ``if" part of the lemma, as the converse is trivial. Let $\mu_{Y,X}^*$ be a model probability and suppose that $\mu^*_{Y,X} \not\in \overline{\mathcal{M}}_{\theta} = \overline{\mathcal{M}}_{\beta, \tau}$. Because $\mu_{Y,X}^*$ is a model probability, there is some $\beta_0$, some $\sigma$-finite measure $\lambda_{\mathcal{Y}} \in \mathfrak{B}(\mathcal{Y})$, and some measure $\gamma_{X,\alpha}^* \in \mathcal{P}(\mathcal{X} \times \mathcal{A})$ such that $\mu_{Y,X}^*$ has density 
    \begin{align*}
        \int_{\mathcal{A}} f_{Y|X,\alpha}(y|x,a ;\beta_0) \, \d \gamma^*_{\alpha|x}    
    \end{align*}
    with respect to measure $\lambda_{\mathcal{Y}} \times \gamma_X^*$ on $\mathcal{Y} \times \mathcal{X}$. By substituting $\lambda_{\mathcal{Y}} + \lambda_{\mathcal{Y}}'$ instead of $\lambda_{\mathcal{Y}}$ for another $\sigma$-finite measure $\lambda_{\mathcal{Y}}'$, we may also assume that every $\mu \in \mathcal{M}_{\theta}$ has density given by 
    \begin{align} \label{E:semiparametric}
    f_{Y|X}(y|x) = \int_{\mathcal{A}} f_{Y|X, \alpha} (y|x, a; \beta) \, \d \gamma_{\alpha|x},
    \end{align}
    for some $\gamma_{X,\alpha} \in \Gamma_{\tau}$. 
    
    Let $\mathcal{M}_{\theta}^*$ be the subset of $\mathcal{M}_{\theta}$ consisting of measures $\mu$ for which one has $\mu_X = \gamma_X^*$, which is to say that the marginal of $\mu$ over $\mathcal{X}$ must be $\gamma^*_X$. Similarly, let $\Gamma_{\tau}^*$ be the subset of measures in $\Gamma_{\tau}$ that have $X$-marginal $\gamma_X^*$. Then every measure in $\mathcal{M}_{\theta}^*$ has density with respect to $\lambda_{\mathcal{Y}} \times \gamma^*_X$, which is $\sigma$-finite. Moreover, $\mu^*_{Y,X}$ is certainly not in the total variation-norm closure of $\mathcal{M}^*_{\theta}$. Then, 
    the same arguments that establish Proposition \ref{P:main_appendix} imply that there is some $\phi_0 \in \Phi_b(\mathcal{Y} \times \mathcal{X})$ such that 
    \begin{align*}
        \EE{\mu^*_{Y,Z}}{\phi_0(Y,X)} &> \sup_{\mu \in \mathcal{M}_{\theta}^*} \EE{\mu}{\phi_0(Y,X)}  = \sup_{\gamma_{X,\alpha} \in \Gamma_{\tau}^*} \EE{\gamma}{\E{\phi_0(Y,X)|X,\alpha}}  \\
        & = \sup_{\gamma_{X,\alpha} \in \Gamma_{\tau}^*} \int_{\mathcal{X}} \int_{\mathcal{A}} h(x,a) \, \d \gamma_{\alpha|x} \, \d \gamma_X^*,
    \end{align*}
    where for $x \in \mathcal{X}$ and $a \in \mathcal{A}$, we have let $h(x,a) = \int_{\mathcal{Y}} \phi_0(y,x) f_{Y|X,\alpha} (y|x,a;\beta) \, \d \lambda_{\mathcal{Y}}$ (which is a measurable function by Fubini's theorem). By Lemma \ref{L:marginalmax} below, there is some Borel set $\Omega \subseteq \mathcal{X}$ such that $\gamma_X^*(\Omega) = 1$ with the property that the function 
    \[
    \rho(x) \cdot \one_{x \in \Omega} =  \sup_{\substack{\gamma_{\alpha|x} \in \mathcal{P}(\mathcal{A}) \\ \int_{\mathcal{A}} g(x,a,\beta, \tau) \, \d \gamma_{\alpha|x} = 0}} \int_{\mathcal{A}} h(x,a) \, \d \gamma_{\alpha|x} \, \d \gamma_X^* \cdot \one_{x \in \Omega}
    \]
    is measurable over $\mathcal{X}$, and
    \begin{align*}
       \sup_{\gamma_{X,\alpha} \in \Gamma_\tau^*} \int_{\mathcal{X}} \int_{\mathcal{A}} h(x,a) \, \d \gamma_{\alpha|x} \, \d \gamma_X^* = \int_{\Omega} \rho(x) \, \d \gamma_X^*.
    \end{align*}
    Let $h_0(x,\alpha) = \int_{\mathcal{Y}} \phi_0(y,x) f_{Y|X,\alpha} (y|x,a;\beta_0) \, \d \lambda_{\mathcal{Y}}$. The previous display implies that 
    \begin{align*}
        \int_{\Omega} \int_{\mathcal{A}} h_0(x,a) \d \gamma_{\alpha|x}^* \, \d \gamma_X^* &=\int_{\mathcal{X}} \int_{\mathcal{A}} h_0(x,a) \,  \d \gamma_{\alpha|x}^*\, \d \gamma_X^*  > \int_{\Omega} \rho(x) \, \d \gamma_X^*.
    \end{align*}
    Let $\Omega_1 \subseteq \Omega$ be the subset of $\Omega$ consisting of the points $x \in \mathcal{X}$ such that $\int_{\mathcal{A}} h_0(x,a) \d \gamma_{\alpha|x}^*$ is strictly greater than $\rho(x)$. The previous display and the choice of $\phi_0$ implies that $\gamma_X^*(\Omega_1) > 0$. For a fixed number $\ve > 0$, define the following bounded function $\phi_\ve$:
    \begin{align*}
        \phi_\ve: (y,x) \mapsto \left\{ \begin{array}{ll}
             \phi_0(y,x) / (\max \{ \rho(x), \ve \} )\} & \text{ if }x \in \Omega_1 \\
             1 & \text{ otherwise} 
        \end{array} \right.
    \end{align*}
    Note that, for any $x \in \Omega_1$ and any measure $\gamma_{\alpha|x}$ over $\mathcal{A}$ such that $\int_{\mathcal{A}} g(x,\alpha, \beta, \tau) \, \d \gamma_{\alpha|x} = 0$,
    \begin{align}
        &\int_{\mathcal{A}} \int_{\mathcal{Y}} \phi_\ve(y,x) f_{Y|X,\alpha} (y|x,a;\beta) \, \d \lambda_{\mathcal{Y}} \, \d \gamma_{\alpha|x}  =  \frac{\int_{\mathcal{A}} h(x,a) \, \d \gamma_{\alpha|x} }{ \max \{ \rho(x), \ve \}} \le 1, \text{ and } \nonumber \\
        &\int_{\mathcal{A}} \int_{\mathcal{Y}} \phi_\ve(y,x) f_{Y|X, \alpha} (y|x,a;\beta_0) \, \d \lambda_{\mathcal{Y}} \, \d \gamma_{\alpha|x}   = \frac{\int_{\mathcal{A}} h_0(x,a) \, \d \gamma_{\alpha|x} }{ \max \{ \rho(x), \ve \}}. \label{E:ps146}
    \end{align}
    As $\ve$ decreases to $0$, the last line of \eqref{E:ps146} increases to a point that is strictly greater than $1$ by definition of $\Omega_1$. By the monotone convergence theorem, there must exist some $\ve > 0$ such that $\EE{\mu^*_{Y,Z}}{\phi_\ve(Y,X)} > 1$. On the other hand, by the first line of \eqref{E:ps146}, $\sup_{\mu \in \mathcal{M}_{\theta}} \EE{\mu}{\phi_\ve(Y,X)} \le 1$. By rescaling and applying the compactness argument of Proposition \ref{P:main_appendix}, we finally arrive at a compactly supported $\phi \in \Phi_b (\mathcal{Y} \times \mathcal{X})$ which satisfies 
    \begin{align*}
        \EE{\mu^*_{Y,X}}{\phi(Y,X)} > \sup_{\mu \in \mathcal{M}_{\theta}} \EE{\mu}{ \phi(Y,X)},
    \end{align*}
    which completes the proof of the lemma.
\end{proof}

\begin{lemma} \label{L:marginalmax}
    Let $\mathcal{X}$ and $\mathcal{A}$ be Polish spaces, $\gamma_X^* \in \mathcal{P}(\mathcal{X})$, and $\Gamma^*_\tau \subseteq \Gamma_\tau$ be the subset of probability measures in $\Gamma_\tau$ having marginal $\gamma_X^*$ over $\mathcal{X}$. Let $g(\cdot, \cdot, \theta) : \mathcal{X} \times \mathcal{A}  \ra \R$ be Borel measurable and satisfy \eqref{E:taueqn}.
    Furthermore, let $\Delta^{d_g}$ denote the $d_g$-dimensional unit simplex and let $h: \mathcal{X} \times \mathcal{A} \ra \R$ be bounded and Borel. Then, for all $\ve > 0$ there exists a Borel map $\rho: \mathcal{X} \ra \R$ such that
    \begin{align}
         &\rho(x) = \sup_{\substack{\gamma_{\alpha|x} \in \mathcal{P}(\mathcal{A}) \\ \int_{\mathcal{A}} g(x,a, \theta) \, \d \gamma_{\alpha|x} = 0}} \int_{\mathcal{A}} h(x,a) \, \d \gamma_{\alpha|x}\text{ }\gamma_{X}^*\text{-a.s., and } \nonumber \\
        &\sup_{\gamma \in \Gamma^*_\tau} \EE{\gamma}{ h(X, \alpha) } = \EE{\gamma_X^*}{ \rho(X)}. \label{E:ref12}
    \end{align}
\end{lemma}
\begin{proof}
    Assume without loss of generality that $h$ has range in $[0,1]$. Let $\Delta^{d_g} \subseteq \R^{d_g + 1}$ denote the $d_g$-dimensional unit simplex. Define the function $H: \mathcal{X} \times \mathcal{A}^{d_g + 1} \times \Delta^{d_g} \ra \R$ by letting
    \begin{align*}
        H(x, a_1, \ldots, a_{d_g + 1}, c_1, \ldots, c_{d_g +1} ) = \sum_{j = 1}^{d_g + 1} c_j h(x,a_j). 
    \end{align*}
    Moreover, define the correspondence $S: \mathcal{X} \ra \mathcal{A}^{d_g + 1} \times \Delta^{d_g}$ by letting
    \begin{align*}
        S(x) = \{(a_1, \ldots, a_{d_g + 1}, c_1, \ldots, c_{d_g + 1}): \sum_{j = 1}^{d_g + 1} c_j g(x, a_j, \theta) = 0\}.
    \end{align*}
    Note that \eqref{E:taueqn} is sufficient (and necessary; see Theorems 2.1 and 3.1 of \citet{winkler88}) to guarantee that $S(x)$ is nonempty for all $x$.
    The graph of $S(x)$ is the preimage of the set $\mathcal{X} \times \{0\}$ under the (Borel; see Lemma 2.12.5 of \citet{bogachev2007measure}) map 
    \[
    (x, a_1, \ldots, a_{d_g + 1}, c_1, \ldots, c_{d_g + 1}) \mapsto (x, \sum_{j = 1}^{d_g + 1} c_j g(x, a_j, \theta)),
    \]
    so it is itself Borel. Therefore, implication (a) of Theorem 2.17 of \citet{sw1992} implies that the map 
    \begin{align} \label{E:rhoeqn1}
    x \mapsto \sup_{(a_1, \ldots, a_{d_g + 1}, c_1, \ldots, c_{d_g + 1}) \in  S(x) } H(x, a_1, \ldots, a_{d_g + 1}, c_1, \ldots, c_{d_g + 1})
    \end{align}
    is analytic, whence universally measurable. In particular, this map is measurable with respect to the completion $\overline{\gamma}^*_X$ of $\gamma^*_X$ (ibid, Fact 2.9). By Proposition 2.1.11 of \citet{bogachev2007measure}, there exists some Borel map $\rho: \mathcal{X} \ra \R$ such that $\rho(x)$ equals the map in \eqref{E:rhoeqn1} for all $x$ in some Borel set $\Omega$ having $\gamma_X^*$-measure equal to $1$. By Theorem 3.1 of \citet{winkler88}, the right hand side of \eqref{E:rhoeqn1} is equal to the first line of the right hand side of \eqref{E:ref12}, and the first claim of the lemma is established. 
    
    Now for $\ve \in (0,1)$, consider the Borel function $H_\ve: \mathcal{X} \times \mathcal{A}^{d_g + 1} \times \Delta^{d_g}  \ra \R$ which is set equal to $ \min \{ H(x, a_1, \ldots, a_{d_g + 1}, c_1, \ldots, c_{d_g + 1}), (1- \ve) \rho(x)\}$. For all $x \in \Omega$, the function $H_\ve(x,\cdot)$ clearly achieves its supremum $(1- \ve) \rho(x)$ over $S(x)$. Therefore, implication (d) of Theorem 2.17 of \citet{sw1992} implies that there exists a universally measurable map $m_\ve: \mathcal{X} \ra \mathcal{A}^{d_g + 1} \times \Delta^{d_g} $ such that 
    \[
    \overline{\gamma}_X^*(\{x: m_\ve(x) \in S(x), H_\ve(x,m_\ve(x)) = (1-\ve) \rho(x) \}) = 1.
    \]
    By another application of Proposition 2.1.11 of \citet{bogachev2007measure}, we may again assume that $m_\ve$ is a Borel map bymodifying it on a $\gamma_X^*$-null set. Denote the coordinate functions of $m_\ve(x)$ as $a_1(x), \ldots, a_{d_g + 1}(x), c_1(x), \ldots, c_{d_g + 1}(x)$. For $j = 1, \ldots, d_g + 1$, define $C_j(x) = \sum_{j = 1}^j c_j(x)$, and define $C_0(x) = 0$. 

    Now, define a map $M_\ve : \mathcal{X} \times [0,1) \ra \mathcal{A} $ by writing 
    \[
    M_\ve(x, v) = \left\{\begin{array}{ll}
         a_1 & \text{ if } C_0(x) \le v < C_1(x)  \\
         & \vdots \\
         a_{d_g + 1} & \text{ if } C_{d_g}(x) \le v < 1
    \end{array}. \right.
    \]
    Let $A \subseteq \mathcal{A}$ be measurable and let $\mathbb{Q}$ be the rational numbers; then 
    \begin{align*}
        M_\ve^{-1}(A) & = \bigcup_{j = 1}^{d_g + 1} \{(x,v): a_j(x) \in A, v \in [0,1)\} \cap \{(x,v): C_{j -1}(x) \le v \} \cap \{(x,v): v  < C_{j}(x) \},
    \end{align*} 
    where e.g.\ $\{(x,v): v < C_j(x) \}$ is the subgraph of the measurable function $C_j(x)$ and is itself measurable. Therefore, $M_\ve$ is measurable (\citet{bogachev2007measure}, Proposition 3.3.4). The map $\tilde{M}_\ve: (x, v) \mapsto (x, M_\ve(x,v))$ must also be measurable by Lemma 2.12.5 of \citet{bogachev2007measure}. 
    
    Let $\lambda_{U}$ be the uniform measure on $[0,1)$. Then, the pushforward $\gamma = (\tilde{M}_\ve)_* (\gamma_X^* \times \lambda_U)$ is a Borel measure on $\mathcal{X} \times \mathcal{A}$ and has $\mathcal{X}$-marginal $\gamma_X^*$. Moreover, 
    \begin{align*}
        \int_{\mathcal{A}} g(x, a, \theta ) \, \d \gamma_{\alpha|x} = &\int_{0}^1 g(x,M_\ve(x,v), \theta) \, \d v = \sum_{j = 1}^{d_g + 1} c_j(x) g(x, a_j(x), \theta) \overset{\gamma_X^*\text{-a.s.}}{=} 0,
    \end{align*}
    so that $\gamma \in \Gamma^*_\tau$. Similarly,
    \begin{align*}
         \int_{\mathcal{A}} h(x, a) \, \d \gamma_{\alpha|x} = & \sum_{j = 1}^{d_g + 1} c_j(x) h(x, a_j(x))  = H(x, m_\ve(x)) \ge H_\ve(x, m_\ve(x)) \\
        \overset{\gamma_X^*\text{-a.s.}}{=} & (1-\ve) \rho(x).
    \end{align*}
    Hence, 
   $
        \EE{\gamma}{h(X,\alpha)} \ge (1-\ve)\EE{\gamma_X^*}{\rho(X)}.
   $
    On the other hand, 
    \begin{align*}
        \sup_{\gamma \in \Gamma^*_\tau} \EE{\gamma}{h(X,\alpha)} = \sup_{\gamma \in \Gamma^*_\tau} \EE{\gamma}{h(X,\alpha) \one_{X \in \Omega}} \le \sup_{\gamma \in \Gamma_{\tau}^*} \EE{\gamma_X}{\rho(X)} = \EE{\gamma_X^*}{\rho(X)}.
    \end{align*}
    Because $\ve$ was arbitrary, the second claim of the lemma follows. 
\end{proof}

\section{Proofs: semiparametric panel models}
\label{app:panel_proofs}

\begin{proof}[Proof of Theorem \ref{thm:joint_stat}]
Via the pushforward representation, the set $\mathcal M_\theta$ has elements
\begin{equation}
\mu_{(Y, X)}(Y_1, Y_2, X_1, X_2) = \int_{\mathcal{A} \times \mathcal{X}} P(Y_1, Y_2 \mid \alpha, X_1, X_2) \, d\gamma_{(\alpha, X)}.
\end{equation}

By the disintegration theorem, the set $\Gamma_\theta(\mathcal W)$ consists of joint distributions with differential
$d\gamma = d\gamma_{V|\alpha, X} \, d\gamma_{\alpha, X},$
where \(d\gamma_{\alpha, X}\) is the marginal distribution over \((\alpha, X)\) and \(d\gamma_{V|\alpha, X}\) is the conditional distribution of \(V\) given \(\alpha\) and \(X\). 

The proof proceeds by showing that Assumption \ref{a:conditional_stationarity} and the definitions of $\tau_t$ are linear restrictions on \(d\gamma_{V|\alpha, X}\), which implies that $\Gamma_\theta$ is convex. The result then follows from Corollary \ref{thm:pushforward_sharp}. 
The proof consists in verifying the assumptions of Corollary \ref{thm:pushforward_sharp}. 

First note that $\Gamma_\theta(\mathcal W)$ is convex. To see this, note that Assumption \ref{a:conditional_stationarity} implies that for any bounded measurable function \(\xi(v)\), 
\[
\int \xi(v) \, d\gamma_{V_1|\alpha, X_1, X_2} = \int \xi(v) \, d\gamma_{V_2|\alpha, X_1, X_2} \quad \text{for all } \alpha, X_1, X_2.
\]
Using the joint density \(\gamma\), this equality can be rewritten as
\begin{equation}
\label{stricto}
    \int_{\mathcal V_2} \xi(v) \, d\gamma(\alpha, v_1, v, x_1, x_2) \, dv = \int_{\mathcal V_1} \xi(v) \, d\gamma(\alpha, v, v_2, x_1, x_2) \, dv,
\end{equation}
for all \((\alpha, x_1, x_2) \in \mathcal{A} \times \mathcal{X}\).

Define the following linear operator acting on \(\gamma\):
\[
\mathcal{L}[\gamma] \equiv \int_{\mathcal V_2} \gamma(\alpha, v_1, v, x_1, x_2) \, dv - \int_{\mathcal V_1} \gamma(\alpha, v, v_2, x_1, x_2) \, dv.
\]
Then, for \eqref{stricto} to hold, \(\mathcal{L}[\gamma] = 0\) for all \((\alpha, v_1, v_2, x_1, x_2) \in \mathcal W\).
\footnote{The definition of $\mathcal{L}[\gamma]$ implicitly assumes that $V_t$ are continuously distributed. This is inconsequential for establishing the linearity of the strict exogeneity restriction.} 

The measure \(\gamma\) satisfies the linear restrictions \(\mathcal{L}[\gamma] = 0\), as well as those imposed by the definition of \(\tau_t(x)\), which are also linear. Given that \(\mathcal{W}\) is a Polish space, the set \(\Gamma_\theta(\mathcal{W})\) represents all distributions that adhere to these linear constraints. Consequently, \(\Gamma_\theta(\mathcal{W})\) forms a convex set.

The remaining assumptions are trivially satisfied, with $\lambda_\theta$ the counting measure and $\phi_\theta$ given by \eqref{binarypsi}.
 
\end{proof}

\subsection*{Convexity under sequential exogeneity}

\begin{asm}\label{A:support_gamma}
Let \(\mathcal{X}, \mathcal{V}, \mathcal{A}\) be Polish spaces representing the supports of \(X\), \(V\), and \(\alpha\), respectively.
\end{asm}

In the following assumption, we deem a map $h: (v_1,v_2, a) \ra (v_1, v_2)$ to be a \textit{stationary map} if $h$ acts on $(v_1, v_2)$ component-wise, which is to say that $h(v_1, v_2, a) = (\tilde{h}(v_1,a), \tilde{h}(v_2,a))$ for some map $\tilde{h}$. We say that an invertible map is bimeasurable if it and its inverse are measurable.

\begin{asm} \label{A:pred}
    (i) $\mathcal{M}_\theta$ is the set of measures $\mu_{(Y,X)}$ which have the pushforward representation $\mu_{(Y,X)} = (\psi_\theta)^*\gamma$, where $\psi_\theta: (X,V,\alpha) \mapsto (Y,X)$ is Borel measurable, for some distribution $\gamma$ that conforms to \eqref{E:seqex}; 
    (ii) For any measurable function $\rho: \mathcal{A} \ra \mathcal{A}$, there is some measurable and stationary function $h[\rho]:  \mathcal{V} \times \mathcal{A} \ra \mathcal{V}$ satisfying that $\psi_\theta(x,v,a) = \psi_\theta(x, h[\rho](v,a), \rho(a))$ for all $x,v,a$; 
    (iii) There exist a strict subset $\mathcal{A}_0 \subset \mathcal{A}$ and bimeasurable bijections $\rho_0: \mathcal{A} \ra \mathcal{A}_0$ and $\rho_1: \mathcal{A} \ra \mathcal{A}_0^c$, where $\mathcal{A}_0^c$ is the complement of $\mathcal{A}_0$ in $\mathcal{A}$.
\end{asm}

\begin{lemma}
\label{convexitypredeterminess}
Under Assumptions \ref{A:support_gamma} and \ref{A:pred}, $\mathcal{M}_\theta$ is a convex set.
\end{lemma}
\begin{proof}
\label{proofofconvexitypredeterminess}
    We note that the proof uses the following implication of Assumption \ref{A:pred}(i). The set $\Gamma_\theta(\mathcal W)$ is the set of joint distributions \(\gamma\) of \((X, V, \alpha)\) with differential
    \[
    d\gamma = d\gamma_{V|\alpha, X} \, d\gamma_{\alpha, X},
    \]
    where \(d\gamma_{\alpha, X}\) is the marginal distribution over \((\alpha, X) \in \mathcal{A} \times \mathcal{X}\) and \(d\gamma_{V|\alpha, X}\) is the conditional distribution of \(V \in \mathcal{V}\) given \(\alpha\) and \(X\). The support of \(\gamma\) is given by:
    \[
    \text{supp}(\gamma) \subseteq \{(x, v, \alpha) \in \mathcal{X} \times \mathcal{V} \times \mathcal{A} : \gamma_{V|\alpha, X}(v | \alpha, x) > 0 \text{ for some } \gamma_{\alpha, X}(\alpha, x) > 0\}.
    \]

Let $\mu_0 = (\psi_\theta)_*\gamma_0$ and $\mu_1 = (\psi_\theta)_*\gamma_1$ be measures in $\mathcal{M}_\theta$ (so that, in particular, $\gamma_\iota$ satisfies \eqref{E:seqex} for $\iota = 1,2$), and let $c \in (0,1)$. Define the maps $h_\iota: (x,v,a) \mapsto (x, h[\rho_\iota](v,a), \rho_\iota(a))$ for $\iota = 0,1$, and let $\gamma_0^h = (h_0)_* \gamma_0$ and $\gamma_1^h = (h_1)_* \gamma_1$. Then, $\gamma_0^h$ is supported on $(\mathcal{X} \times \mathcal{V} \times \mathcal{A}_0)$ and $\gamma_1^h$ is supported on $(\mathcal{X} \times \mathcal{V} \times \mathcal{A}_0^c)$. By definition of the maps $h[\rho]$, one clearly has $\mu_\iota = (\psi_\theta)_* \gamma_\iota^h$, so that 
    \[
    \mu_0 + (1 - c) \mu_1 = (\psi_\theta)_* ( c \gamma_0^h + (1-c) \gamma_1^h).
    \]
    The proof will be complete if we can show that the measure $\gamma_c^h \equiv c \gamma_0^h + (1-c) \gamma_1^h$ satisfies the conditional restriction in \eqref{E:seqex}. Note that \eqref{E:seqex} holds if and only if $V_2|\alpha, X_1, X_2 \overset{d}{=} V_2|\alpha, X_1$ and $V_2|\alpha, X_1 \overset{d}{=} V_1|\alpha, X_1$ almost surely, which is true if and only if one has 
\begin{align}
    \label{E:momenteq1}
    \EE{\gamma_c^h}{\xi_2(V_2) |\alpha, X_1, X_2} &= \EE{\gamma_c^h}{ \xi_2(V_2) | \alpha, X_1} \\
    \EE{\gamma_c^h}{\xi_1(V_2) |\alpha, X_1} &= \EE{\gamma_c^h}{ \xi_1(V_1) | \alpha, X_1}
\end{align}
for all bounded and measurable functions $\xi_1,\xi_2$ almost surely. Equation \eqref{E:momenteq1}, in turn, holds if and only if 
\begin{align}
    \label{E:momenteq2}
    \EE{\gamma_c^h}{\xi_2(V_2) \xi'_2(\alpha, X_1, X_2)} = 0 \text{ and }\EE{\gamma_c^h}{(\xi_1(V_2) - \xi_1(V_1)) \xi_1'(\alpha, X_1)} = 0
\end{align}
for all bounded and measurable $\xi_1$, $\xi_2$, $\xi_1'$, and $\xi_2'$, where $\xi_2'$ has the property that $\EE{\gamma_c^h}{\xi_2' (\alpha, X_1, X_2)| \alpha, X_1} \eqas 0$, as one can consider letting $\xi_1'$ and $\xi_2'$ be such that
\begin{align*}
    &\xi_1'(\alpha, X_1) = \EE{\gamma_c^h}{(\xi_1(V_2) - \xi_1(V_1))|\alpha, X_1} \\
    &\xi_2'(\alpha, X_1, X_2) = \EE{\gamma_c^h}{\xi_2(V_2) | \alpha, X_1, X_2 } - \E{ \xi_2(V_2) | \alpha, X_1}.  
\end{align*}
Let the functions $\xi_\iota, \xi'_\iota$ be as above. Then, one has 
\small
\begin{align*}
    \EE{\gamma_c^h}{\xi_2(V_2) \xi'_2(\alpha, X_1, X_2)} & =  c \EE{\gamma_0^h}{\xi_2(V_2) \xi'_2(\alpha, X_1, X_2)} + (1 - c) \EE{ \gamma_1^h}{\xi_2(V_2) \xi'_2(\alpha, X_1, X_2)}
\end{align*}
\normalsize
where, for $\iota = 1,2$, 
\small
\begin{align}
    \EE{\gamma_\iota^h}{\xi_2(V_2) \xi'_2(\alpha, X_1, X_2)} & = \EE{\gamma_\iota}{\xi_2(\tilde{h}[\rho_\iota](V_2, \alpha)) \xi'_2(\rho_\iota(\alpha), X_1, X_2)} \nonumber \\
    & = \EE{\gamma_\iota^h}{ \EE{\gamma_\iota^h}{ \xi_2(\tilde{h}[\rho_\iota](V_2, \alpha)) | \alpha, X_1} \xi'_2(\rho_\iota(\alpha), X_1, X_2) }, \label{E:momenteq3}
\end{align}
\normalsize
where the second line is a consequence of the law of iterated expectations and the first equality of \eqref{E:momenteq1}. Now, by construction of $\xi_2'$, for all bounded functions $\xi_\alpha$ in $(a, x_1)$ which vanish whenever $a \in\mathcal{A}_0^c$ one has 
\small
\begin{align*}
    0 & = \EE{\gamma_c^h}{\xi_2'(\alpha, X_1, X_2) \xi_\alpha(\alpha, X_1)} \\
    & = c \EE{\gamma_0}{\xi_2'(\rho_0(\alpha), X_1, X_2) \xi_\alpha(\rho_0(\alpha), X_1)} + (1 - c) \EE{\gamma_1}{\xi_2'(\rho_1(\alpha), X_1, X_2) \xi_\alpha(\rho_1(\alpha), X_1)} \\
    & = c \EE{\gamma_0}{\xi_2'(\rho_0(\alpha), X_1, X_2) \tilde{\xi}_\alpha(\alpha, X_1)}, 
\end{align*}
\normalsize
where we have let $\tilde{\xi}_\alpha(a, x_1) = \xi_\alpha(\rho_0(a), x_1)$ for all $a, x_1$. As $\tilde{\xi}_\alpha$ ranges over the space of all bounded and measurable functions in $(a,x_1)$ as $\xi_\alpha$ ranges over the space of all bounded and measurable functions which vanish when $a$ is in $\mathcal{A}_0^c$, the previous display implies that $\EE{\gamma_0}{\xi_2'(\rho_0(\alpha), X_1, X_2) | \alpha, X_1} = 0$ (a similar argument can be made in the case $\iota = 1$). Hence, the second line of \eqref{E:momenteq3} vanishes for $\iota = 1,2$ by the law of iterated expectations, and the first equality of \eqref{E:momenteq2} is fulfilled. Similarly, by the second equality of \eqref{E:momenteq1},
\small
\begin{align*}
    \EE{\gamma^h_\iota}{( \xi_1(V_2) - \xi_1(V_1)) \xi_1'(\alpha, X_1)} & = \EE{\gamma_\iota}{(\xi_1(\tilde{h}[\rho_\iota](V_2)) - \xi_1(\tilde{h}[\rho_\iota](V_1))) \xi_1'(\rho_\iota(\alpha), X_1) } = 0
\end{align*}
for $\iota = 0,1$, so that the second equality of \eqref{E:momenteq2} also holds. 
\normalsize
\end{proof}

\begin{proof}[Proof of Theorem \ref{thm:sharpness_sequential_exogeneity}]
    The proof consists by verifying Assumptions \ref{A:support_gamma} and \ref{A:pred} for the binary choice model under sequential exogeneity. The conclusion then follows from Lemma \ref{convexitypredeterminess}.

    Assumption \ref{A:support_gamma} holds: regressors are discrete, errors live on (a subset of) the real line, and the support of fixed effects can be taken to be the real line $\mathcal A = \mathbb R$. To see that Assumption \ref{A:pred}(i) holds, see the mapping in Section \ref{sec:panel_partial_effects}.

    Assumption \ref{A:pred}(ii) holds with 
    \(h\left[\rho\right]\left(v,a\right)=\left(v_{1}+a-\rho\left(a\right),v_{2}+a-\rho\left(a\right)\right).\)
    The map $h$ is stationary, and, since
    \(v_t + a = [v_{t}+a-\rho\left(a\right)] + \rho(a),\)
    we have that $\psi_\theta(x,v,a) = \psi_\theta(x,h[\rho](v,a),\rho(a)$.
    Many bijections $\rho$ can be used to satisfy Assumption \ref{A:pred}(iii), e.g. $\rho_t(a) = (-1)^t \exp(a)$.
\end{proof}

\section{Computation}
\label{app:computation_general}

This appendix provides additional details on computation in support of Sections 3--5.

\subsection{Linear programming}
\label{app:linear_programming}

We now demonstrate that the bilinear program in equation \eqref{eq:T_is_bilinear} in Section  
\ref{sec:adversarial} can be reformulated as a linear program (LP). Consider the inner minimization problem:
\[
\begin{array}{ll}
\text{min}        & \phi^\prime C_\theta p_W  \\
\text{subject to} & A_\theta p_W = b_\theta, \\
                  & p_W \geq 0,
\end{array}
\]
which is an LP with coefficients $\phi' C_\theta$ and decision variables $p_W$. The dual of this LP is given by:
\[
\begin{array}{ll}
\text{max}        & \lambda^\prime b_\theta \\
\text{subject to} & A_\theta^\prime \lambda \le C_\theta^\prime \phi,
\end{array}
\]
where $\lambda$ are the dual variables corresponding to the constraints in the primal problem.

By strong duality, we can replace the inner minimization in the bilinear program with its dual maximization problem, provided two conditions are met. First, there must exist at least one $p_W$ satisfying $A_\theta p_W = b_\theta$ and $p_W \geq 0$ (feasibility). Second, the objective function must be bounded from below on the feasible region (boundedness).

Importantly, we note that both $\phi$ and $C_\theta$ have non-negative entries by construction. Given that $p_W \geq 0$, this ensures that the objective function $\phi' C_\theta p_W$ is bounded from below by zero. Thus, the boundedness condition is automatically satisfied in our context.

Regarding the feasibility condition, we observe that if $\theta \in \Theta_I$, then by definition, there exists a feasible $p_W$. Consequently, for $\theta \in \Theta_I$, strong duality holds, and our approach of replacing the inner minimization with its dual is valid.

In the alternative case where $\theta \notin \Theta_I$, the value of $T(\theta)$ is strictly positive. Even if strong duality does not hold in this case, weak duality ensures that the optimal value of the dual problem is always greater than or equal to the optimal value of the primal problem. Therefore, when we replace the inner minimization with its dual maximization, we are, for each $\phi$, replacing the original value with a weakly greater value. Consequently, the supremum over $\phi$ is weakly increased. It follows that we will still correctly conclude that $T(\theta) > 0$ in this case, thereby correctly deciding that $\theta$ is not in the identified set.

This reasoning demonstrates that our approach of replacing the inner minimization with its dual is valid for all $\theta$, regardless of whether $\theta$ is in the identified set or not. When $\theta \in \Theta_I$, strong duality holds and the replacement is exact. When $\theta \notin \Theta_I$, the replacement may overestimate $T(\theta)$, but this overestimation does not affect our ability to correctly classify $\theta$ as being outside the identified set.

Given these observations, we rewrite $T(\theta)$ as:

\begin{equation}
T(\theta) = \left\{
\begin{array}{ll}
\text{max}_{\lambda,\phi}   & \lambda^\prime b_\theta \\
\text{subject to}           & A_\theta^\prime \lambda \le C_\theta^\prime \phi, \\
                            & 0 \leq \phi \leq 1.
\end{array}
\right.
\label{eq:LP_formulation_for_Theta_I}
\end{equation}
In conclusion, our approach of reformulating the bilinear program as an LP through duality is valid and leads to correct identification results for all $\theta$, regardless of whether strong duality holds in all cases.

\subsection{Semiparametric regression models}
\label{sec:computation_examples}

For computation in specific models,
it will be useful specialize the results above
to the semiparametric regression models in Section \ref{sec:convexityinmodels}
with regressors, dependent variables, and unobserved heterogeneity
\[
X \in \mathcal X = \{x_1,\cdots,x_{K_x}\}, \;
Y \in \mathcal Y = \{y_1,\cdots,y_{K_y}\}, \;
U \in \mathcal U = \{u_1,\cdots,u_{K_u}\}.
\]
This fits into the framework above,
with $Z = (X,Y)$ and $W = (X,U)$,
$\mathcal Z = \mathcal X \times \mathcal Y$ and $L = K_x K_y$, and
$\mathcal W = \mathcal X \times \mathcal U$ and $M = K_x K_u$.

For any semiparametric model in Section \ref{sec:semiparametric_pushforward}, the function
\[
m_\theta(x,y,u) = \begin{cases}
    1 & \text{ if } y = h(x,u;\theta), \\
    0 & \text{ otherwise},
\end{cases}
\]
is determined by its outcome equation $h$ in \eqref{eq:model_y_h}, and is known for each $\theta$.
As an example, for the semiparametric binary choice model in Section \ref{sec:example_manski},
\[
m_\theta(x,y,u) = y \times 1\{\beta_0 + \beta_1 X - U \geq 0\} + (1-y) \times 1\{\beta_0 + \beta_1 X - U < 0\}.
\]
This allows us to write the conditional model probabilities as
\[
P(Y = y | X = x) = \sum_{u \in \mathcal U} m_\theta(x,y,u) P(U = u | X = x),
\]
and multiplying both sides by the marginal probability $P(X=x)$ obtains
\[
P(X = x, Y = y) = \sum_{u \in \mathcal U} m_\theta(x,y,u) P(X = x, U = u).
\]

Let  denote the $L \times 1$ model probability vector with rows corresponding to values $z=(x,y)$, and let $p_W$ denote the $M \times 1$ probability vector with rows corresponding to values $w=(x',u)$.
Then we can write 
$p_{Z,(\theta,\gamma)} = \widetilde C_\theta p_W,$
where the pushforward matrix 
\[
    \widetilde C_\theta = (C_{\theta,(x,y),(x',u)})
\]
has rows correspond to values for $(x,y)$ in the same order as $p_Z$, and columns corresponding to values for $w=(x',u)$ in $p_W$, with elements equal to
\begin{align}
\widetilde C_{\theta,(x,y),(x',u)} = 
\begin{cases}
    m_\theta(x,y,u) & \text{ if }x = x', \\
    0               & \text{ otherwise}.
\end{cases}
\label{eq:C_theta_elementwise}
\end{align}

In the next section, we will express $\widetilde C_\theta$ for a specific models, 
and discuss how to formulate the restrictions in terms of $(A_\theta,b_\theta)$.

\subsection{Maximum score}
\label{app:computation_maximum_score}

We consider a binary choice model with a binary regressor and an error term with 3 points of support,
\[
    Y \in \{0,1\}, \; X \in \{x_1,x_2\} \subset \mathbb R, \; U \in \{-1,0,1\}.
\]
We choose the support of $X$ and $U$ small enough so that we can print $(A_\theta,b_\theta,C_\theta)$.
Everything that follows is trivially extended to arbitrary support for $(U,X)$ for the purpose of computation. 

The model is of the maximum score type, and we have analyzed it in detail in Section \ref{sec:example_manski}. 
Recall that the outcome equation is
\[
    Y = 1\{\beta_0 + X \beta_2 - U \geq 0\}.
\]
For this model, 
\begin{align*}
m_\theta(y,x,u) 
&=
\begin{cases} 
    1\{\widetilde x^\prime \theta - u \geq 0\} & \text{ if } y = 1 \\
    1\{\widetilde x^\prime \theta - u < 0\}    & \text{ if } y = 0
\end{cases} \\
&= y \times 1\{\widetilde x^\prime \theta - u \geq 0\} + (1-y) \times 1\{\widetilde x^\prime \theta - u < 0\}.
\end{align*}
where $\widetilde x = (1, x)$.
The pushforward representation 
$p_Z^{(\theta,\gamma)} = \widetilde C_\theta p_W$ in \eqref{eq:pushforward_pmf} is
\tiny
\[
\begin{bmatrix}
    p_{Z}^{(\theta,\gamma)}(x_1,1) \\ 
    p_{Z}^{(\theta,\gamma)}(x_1,0) \\ 
    p_{Z}^{(\theta,\gamma)}(x_2,1) \\ 
    p_{Z}^{(\theta,\gamma)}(x_2,0)
\end{bmatrix}
= 
\begin{bmatrix}
    1\{\widetilde x_1^\prime \theta + 1 \geq 0\} & 1\{\widetilde x_1^\prime \theta \geq 0\} &  1\{\widetilde x_1^\prime \theta - 1 \geq 0\} & 0 & 0 & 0 \\
    1\{\widetilde x_1^\prime \theta + 1 < 0\} & 1\{\widetilde x_1^\prime \theta < 0\} &  1\{\widetilde x_1^\prime \theta - 1 <    0\} & 0 & 0 & 0 \\
    0 & 0 & 0 & 1\{\widetilde x_2^\prime \theta + 1 \geq 0\} & 1\{\widetilde x_2^\prime \theta \geq 0\} &  1\{\widetilde x_2^\prime \theta - 1 \geq 0\} \\
    0 & 0 & 0 & 1\{\widetilde x_2^\prime \theta + 1 < 0\} & 1\{\widetilde x_2^\prime \theta < 0\} &  1\{\widetilde x_2^\prime \theta - 1 < 0\}
\end{bmatrix}
\begin{bmatrix}
    p_{W}(x_1,-1) \\ 
    p_{W}(x_1,0) \\ 
    p_{W}(x_1,1) \\ 
    p_{W}(x_2,-1) \\
    p_{W}(x_2,0) \\ 
    p_{W}(x_2,1)
\end{bmatrix}.
\]
\normalsize

The restrictions in this model are $A_\theta p_{X,U} = b_\theta$,
with
\[
A_\theta
=
\begin{bmatrix}
1 & 1 &  1 & 1 & 1 &  1 \\
1 & 0 & -1 & 0 & 0 &  0 \\
0 & 0 &  0 & 1 & 0 & -1
\end{bmatrix},
\;
b_\theta
=
\begin{bmatrix}
1 \\
0 \\
0
\end{bmatrix},
\]
where the first constraint makes sure that $p_W$ is a probability vector, $\sum_{u,x} P(X=x,U=u) = 1$. 
We do not need to enforce knowledge of the marginal probability $P(X)$, as this information is embedded in $T(\theta)$.
Constraints 2 and 3 ensure that the median is zero, i.e. that 
\[
\sum_{u < 0} P(X=x,U=u) = \sum_{u > 0} P(X=x,U=u),
\]
thus implementing \eqref{medzero}.

\subsubsection{Results: regression coefficients}

Figure \ref{fig:max_score_computation}, panel ``Design 1'', shows the results for the model described here, with
regressor values $(x_1 = 0, x_2 = 1)$,
true regression coefficients $(1,-0.5)$, 
and a grid 
$$\{\theta = (1,\theta_2): \theta_2 \in \{-1.5,-1.49,\cdots,0.5\}\}$$
of 201 values for $\theta$ that imposes the normalization that $\theta_1 = 1$.
In design 1, the true distribution of $X$ is discrete uniform on its support,
and the true conditional distribution of the error term is $P(U=u|X=x) \propto 1/(1+u^2)$.

The results easily generalize to the case where $(X,U)$ have more points of support.
Design 2 has $\mathcal U = \{-5,-4.9,\cdots,4.9,5\}$, but is otherwise like Design 1.
Design 3 is like design 2, but has $\mathcal X = \{-3,-2,\cdots,2,3\}$.
Design 4 is also like designs 2 and 3, but with $\mathcal X = \{-3,-2.75,\cdots,2.75,3\}$.
The results are reported in Figure \ref{fig:max_score_computation}.

\begin{figure}
  \centering

  \begin{subfigure}{0.45\textwidth}
    \centering
    \includegraphics[width=\textwidth]{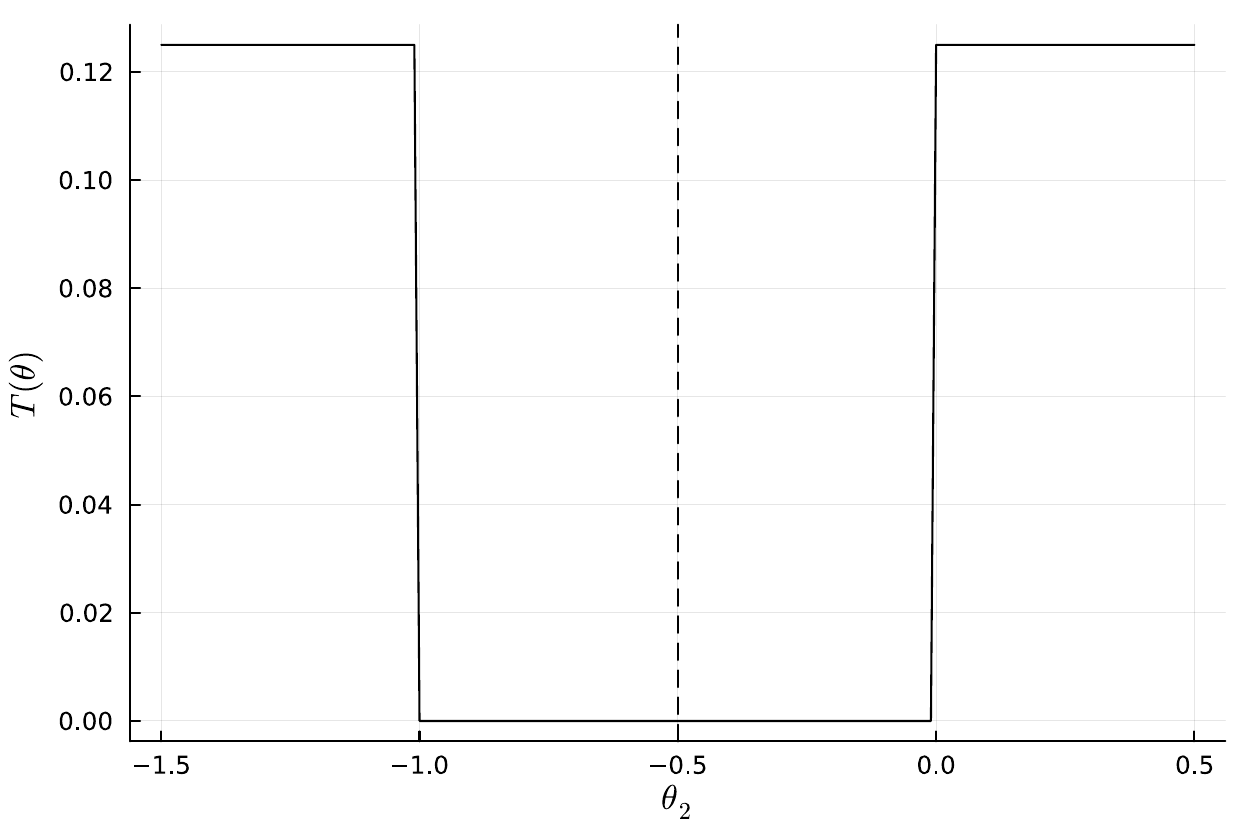} 
    \caption{Design 1}
    \label{fig:maxscore_figure1}
  \end{subfigure}
  \hfill
  \begin{subfigure}{0.45\textwidth}
    \centering
    \includegraphics[width=\textwidth]{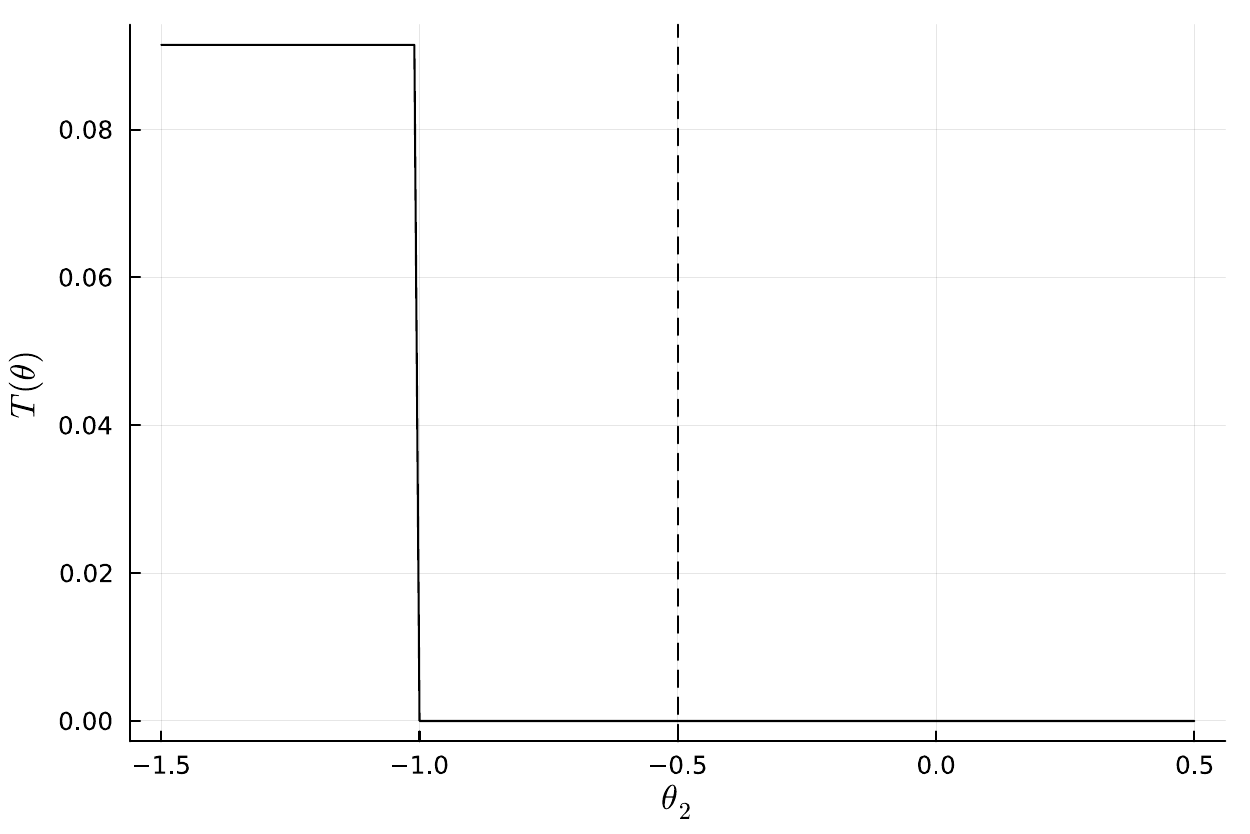} 
    \caption{Design 2: $\mathcal U = \{-5,-4.9,\cdots,5\}$}
    \label{fig:maxscore_figure2}
  \end{subfigure}

  \begin{subfigure}{0.45\textwidth}
    \centering
    \includegraphics[width=\textwidth]{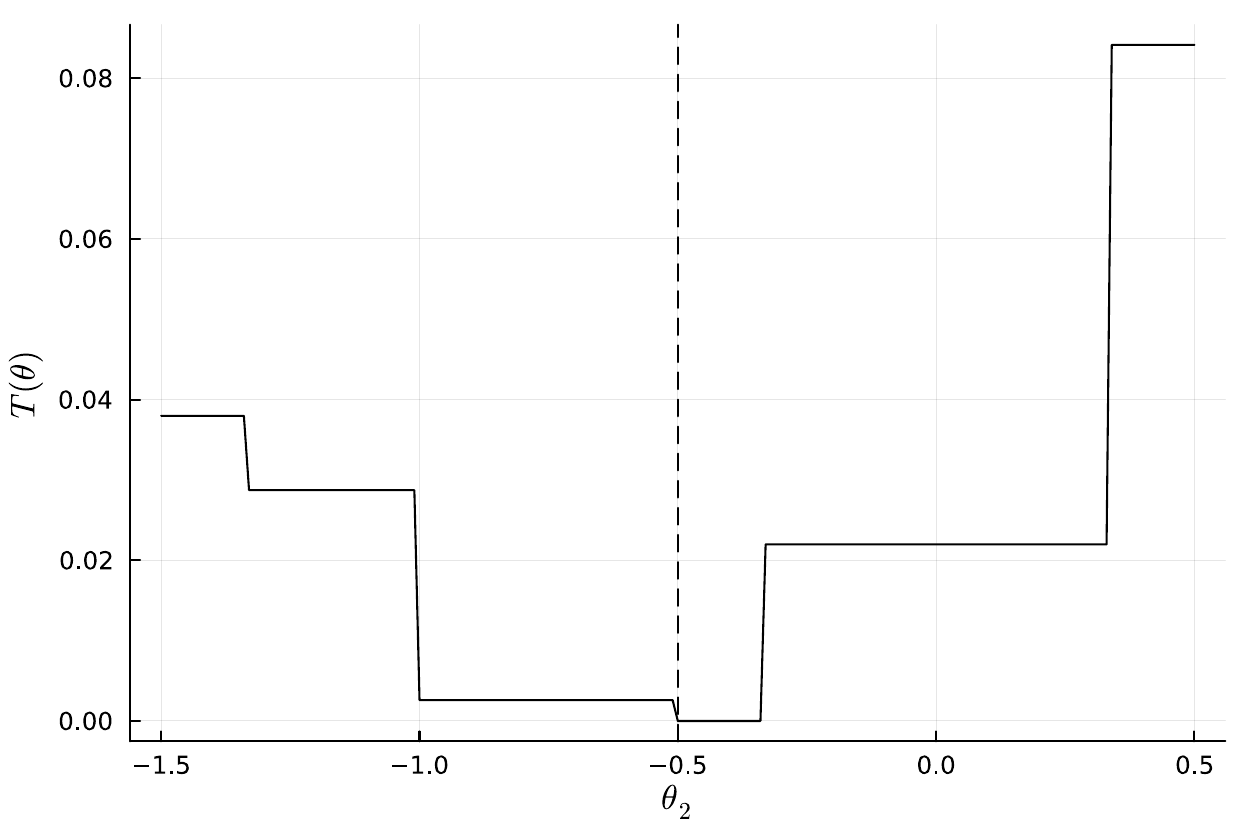}  
    \caption{Design 3: $\mathcal X = \{-3,-2,\cdots,3\}$}
    \label{fig:maxscore_figure3} 
  \end{subfigure}
  \hfill
  \begin{subfigure}{0.45\textwidth}
    \centering
    \includegraphics[width=\textwidth]{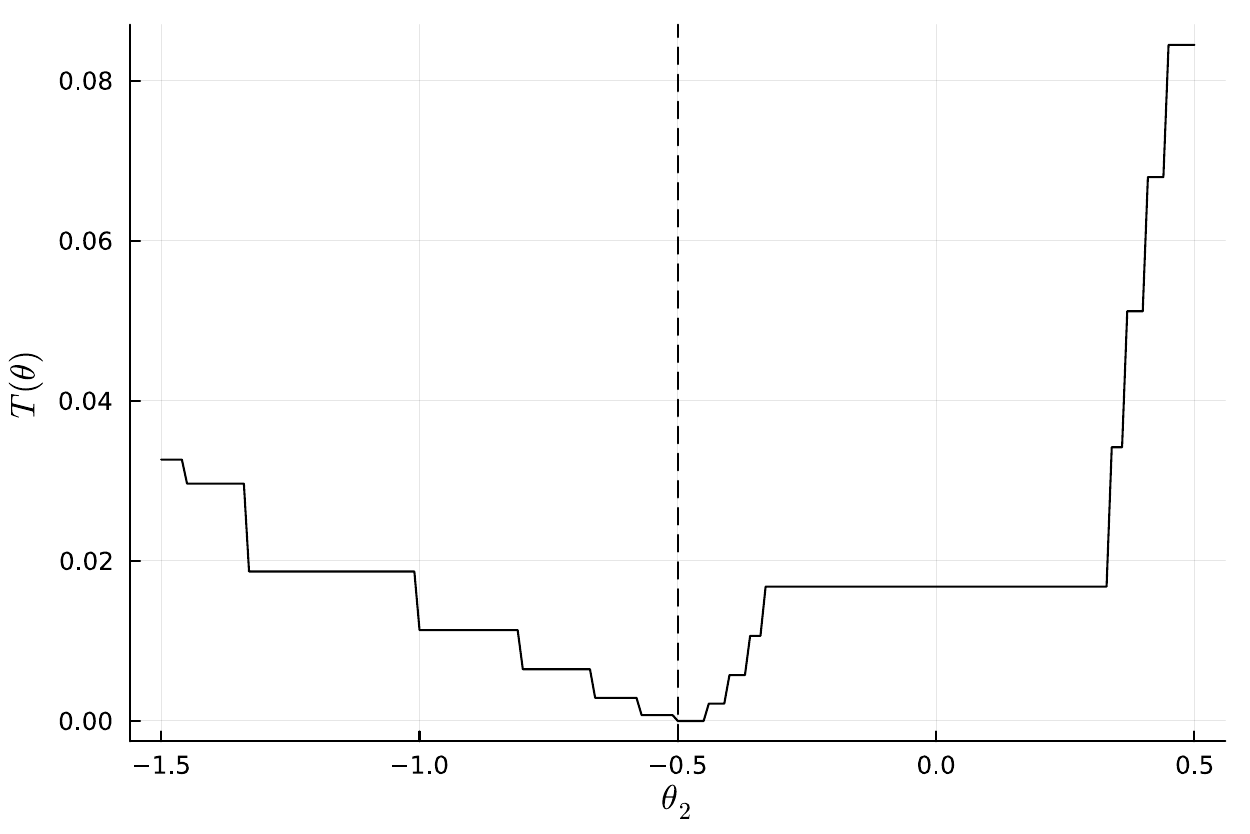} 
    \caption{Design 4: $\mathcal X = \{-3,-2.75,\cdots,3\}$}
    \label{fig:maxscore_figure4}  
  \end{subfigure}  

  \caption{$T(\theta)$ for maximum score.}
  \label{fig:max_score_computation}
\end{figure}

\subsubsection{Partial effects}

We now discuss computation of the identified sets of average partial effects, by describing how they can be incorporated in the constraint set by modifying $A_\theta$ and $b_\theta$. 

Denote by
\[
Y(x,u) = 1\{\beta_0 + x \beta_1 - u \geq 0\}
\]
the value of $Y$ that an individual with error term value $u$ would have under regressor value $x$. 
Then
$\Delta(x,u) = Y(x+1,u) - Y(x,u)$
is the change in $Y$ the individual experiences when increasing their regressor by one unit.

The average partial effect is given by 
\[
E[\Delta(X,U)] = \sum_{(x,u)} \Delta(x,u) p_{X,U}(x,u).
\]
and conditional partial effects are given by 
\[
E[\Delta(X,U)|X=x] = \sum_{u} \Delta(x,u) \frac{p_{X,U}(x,u)}{P(X=x)}.
\]

We now modify the setup described above to incorporate the identification of the conditional partial effect for $X=0$:
\begin{enumerate}
\item expand $\theta = (\theta_1,\theta_2,\theta_3)$, 
where
$\theta_3 = E[\Delta(X,U)|X=0]$
\item pushforward matrix $\widetilde C_{\theta}$ is unchanged
\item add to $A_\theta$ and $b_\theta$ one row for the partial effect.
\end{enumerate}
Only item 3 requires discussion. 
For Design 1, the modified objects are
\[
A_\theta
=
\begin{bmatrix}
1 & 1 &  1 & 1 & 1 &  1 \\
1 & 0 & -1 & 0 & 0 &  0 \\
0 & 0 &  0 & 1 & 0 & -1 \\
\Delta(0,-1) & \Delta(0,0) & \Delta(0,1) &  0 & 0 & 0 
\end{bmatrix}, 
b_\theta = (1, 0, 0, \theta_3 P(X=0)). 
\]
The value $P(X=0)$ is known.

Let $\Theta_{I1}$ be the identified set of values of $(\theta_1,\theta_2)$ determined in the previous section.
Figure \ref{fig:maximum_score_partial_effects} plots $\max_{(\theta_1,\theta_2)} T(\theta)$, 
the maximum value of $T$ over $(\theta_1,\theta_2) \in \Theta_{I1}$,
as a function of the conditional partial effect parameter $\theta_3$. The left panel is for Design 4. The right panel is for Design 4b, which has the conditional distribution of $U|X$ as discrete uniform.

\begin{figure}
    \centering
    \begin{subfigure}{0.45\textwidth}
        \centering
        \includegraphics[width=\textwidth]{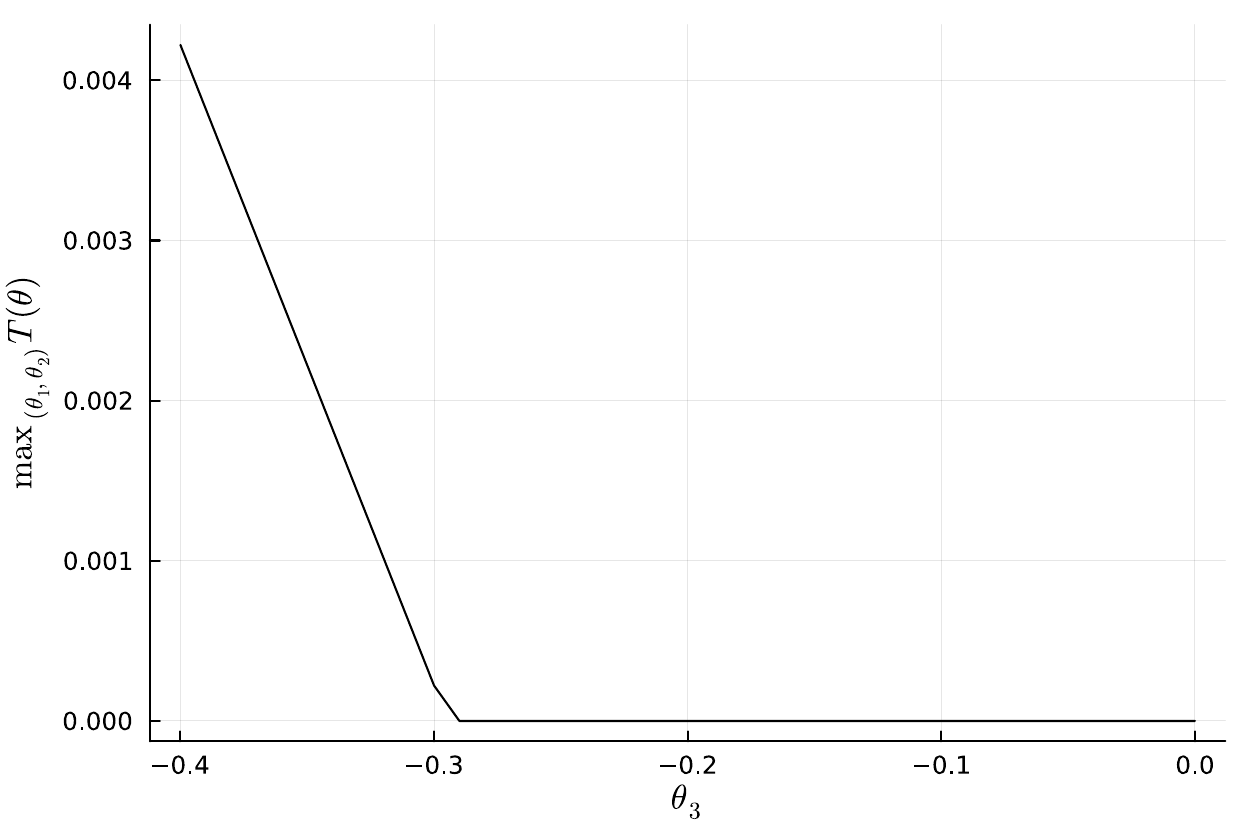} 
        \caption{Design 4}
        \label{fig:maximum_score_partial_effects_4}
    \end{subfigure}
    \hfill
    \begin{subfigure}{0.45\textwidth}
        \centering
        \includegraphics[width=\textwidth]{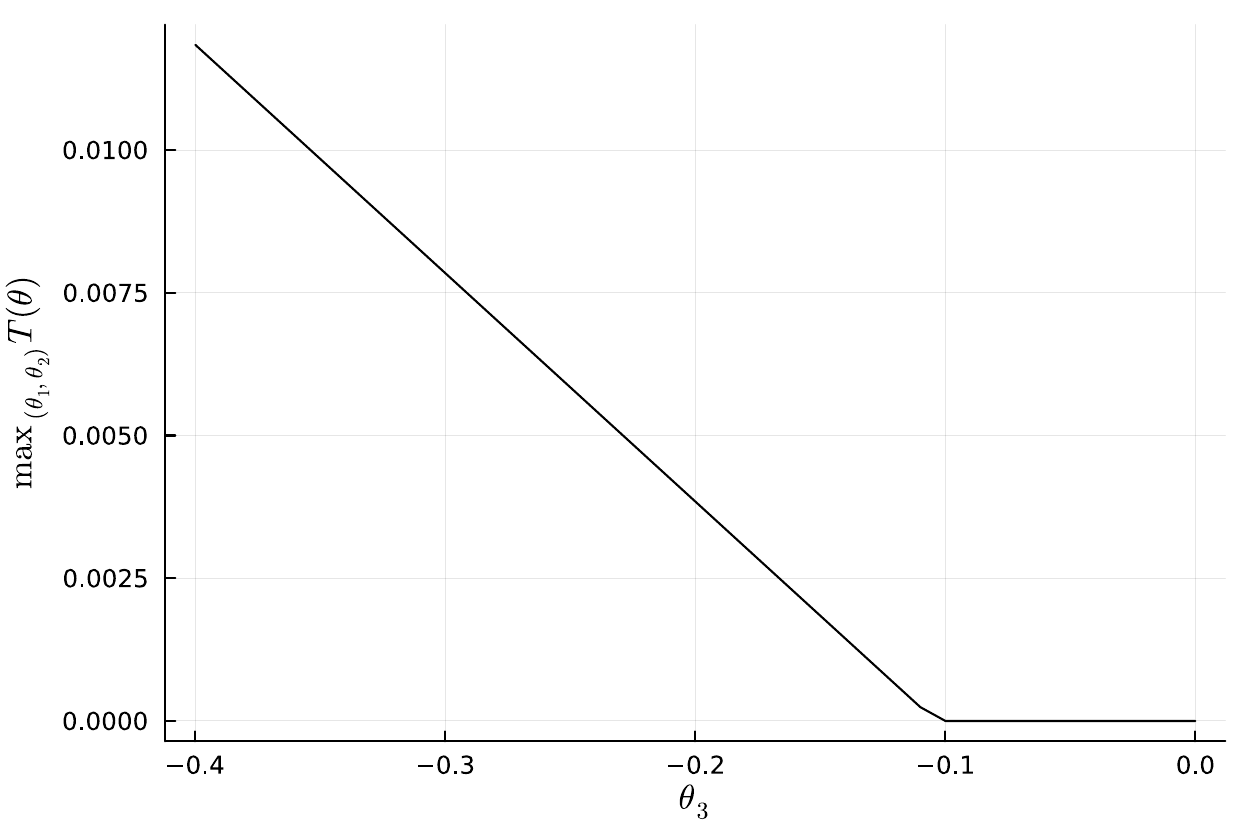} 
        \caption{Design 4b: $U|X \sim$ discrete uniform}
        \label{fig:maximum_score_partial_effects_4b}
    \end{subfigure}
    \caption{Identified set for the partial effect in the maximum score model, Design 4.}
    \label{fig:maximum_score_partial_effects}
\end{figure}

\subsubsection{Computation speed}
\label{app:computation_speed_maxscore_cs}

Each panel in Figure \ref{fig:max_score_computation} requires computing $T(\theta)$ by solving the LP in \eqref{eq:LP_formulation_for_Theta_I}.
for 201 candidate values $\theta$.
No optimizations specific to the model were done: we simply call an off the shelf LP solver (Gurobi 11.0.2).
All timing done on a single core Intel i7-11370H at 3.30GHz.

Table \ref{tab:computation_times}, column $\theta_2$, lists the computation time, in seconds, per evaluation of $T(\theta)$. 
Clearly, the computation of $T(\theta)$ is very fast. 
It appears to be less than linear in $K_u$, and slightly worse than linear in $K_x$.
Table \ref{tab:computation_times}, column $\theta_3$ reports computation times for Figure \ref{fig:maximum_score_partial_effects}. 

\begin{table}
    \centering
    \begin{tabular}{ccccc}
    \toprule
    Design & $\theta_2$ & $\theta_3$ & $K_u$ & $K_x$ \\
    \midrule
    1      & 0.0024 & 0.0016 & 3   & 2 \\
    2      & 0.0033 & 0.0023 & 101 & 2 \\
    3      & 0.0083 & 0.0077 & 101 & 7 \\
    4      & 0.0522 & 0.0536 & 101 & 25 \\
    \bottomrule
    \end{tabular}
    \caption{Time, in seconds, for one evaluation of $T(\theta)$.}
    \label{tab:computation_times}
\end{table}

\subsubsection{Panel data}
\label{app:compute_binary_panel_semiparametric}

Computation for the panel model follows the template of the cross-sectional model, with the following modifications.

\textbf{Puhsforward matrix.} For this model, the pushforward matrix is given by 
\[
\widetilde C_{\theta,(x_1,x_2,y_1,y_2),(x_1',x_2',u_1,u_2}
=
\begin{cases}
m_\theta(x_1,x_2,y_1,y_2,u_1,u_2) & \text{ if }x_1 = x_1',x_2 = x_2', \\
0 & \text{ otherwise},
\end{cases}
\]
and 
\[
m_\theta(x_1,x_2,y_1,y_2,u_1,u_2) = \prod_{t=1}^2 y_t \times 1\{\widetilde x_t^\prime \theta + u_t \geq 0\} + (1-y_t) \times 1\{\widetilde x_t^\prime \theta + u_t < 0\}.
\]

\textbf{Restrictions.} 
The stationarity restriction is given by a matrix $A_1$ that has a row for each value of $X=x$, and for every value $u$ that $U_t$ may take. Every column corresponds to a value of $(X=x,U_1=u_1,U_2=u_2)$. The element
\[
A_{1,(x,u),(x',u_1,u_2)} = 
\begin{cases}
    1  & \text{ if }x=x', u_1 =    u \neq u_2 \\
    -1 & \text{ if }x=x', u_1 \neq u =    u_2 \\
    0  & \text{ otherwise}.
\end{cases}
\]
The corresponding $b_1 = 0$. These restrictions replace the median-zero restrictions.

The adding-up constraint is unchanged. 
The restrictions associated with the counterfactual choice probability must be modified to reflect the panel data setting.

\subsection{Parametric models} \label{sec:comparison_computation}

We now discuss how to extend
the LP approach developed in Section \ref{sec:adversarial} and Appendix \ref{sec:computation_examples} to the parametric models of Section \ref{sec:nonlinearpanels}.
To see this, 
denote regressors, dependent variables, and nonparametric unobservables by
\[
X \in \mathcal X = \{x_1,\cdots,x_{K_x}\}, \;
Y \in \mathcal Y = \{y_1,\cdots,y_{K_y}\}, \;
\alpha \in \mathcal A = \{a_1,\cdots,a_{K_a}\},
\]
with $Z = (X,Y)$ and $\overline W = (X,\alpha)$,
$\mathcal Z = \mathcal X \times \mathcal Y$ and $L = K_x K_y$, and
$\overline{\mathcal W}= \mathcal X \times \mathcal A$ and $M = K_x K_a$.
Compared to the semiparametric regression setup in Section \ref{sec:computation_examples}, we have unobserved heterogeneity $\alpha$ instead of $U$. 

The parametric model starts from $U = (\alpha,V)$. It is given that $V$ is subject to parametric restrictions, and can be integrated out, see Section \ref{sec:nonlinearpanels}.
This treatment of $V$ requires a modification of the $m$ function, to
\[
\widetilde m_\theta(x,y,a) = f_{Y|X,\alpha}(y|x,a;\beta).
\]
Now define the matrix
\begin{align}
\widetilde C_{\theta,(x,y),(x',a)} = 
\begin{cases}
    \widetilde m_\theta(x,y,a) & \text{ if }x = x', \\
    0               & \text{ otherwise}.
\end{cases}
\label{eq:C_theta_parametric}
\end{align}
This is the direct analog of \eqref{eq:C_theta_elementwise} for the parametric model. 

Note that \eqref{eq:C_theta_parametric} does not represent a pushforward operation, but rather it is the matrix representation of an integral operator, see Section \ref{sec:nonlinearpanels}.
Thus, any model in Section \ref{sec:nonlinearpanels} can be represented by 
\[
p_Z = \widetilde C_{\theta} p_{\overline W},
\]
where $p_{\overline W}$ is the pmf associated with a discretized $\overline W$, cf.  \eqref{eq:pushforward_pmf}.
The computational tools in Section \ref{sec:computation_examples} only use linearity.
In conclusion, the LP approach extends to the parametric models in Section \ref{sec:nonlinearpanels}.

\subsubsection{Parametric binary choice models with fixed effects}
\label{app:computation_parametric_binary_T2}

We now obtain explicit expressions for $(A,b,C)$ in the parametric binary choice model with $T=2$. 
In this case, the support of $Y$ is $\mathcal{Y}=\{0,1\}^2 = \left\{ (0,0),(0,1),(1,0),(1,1)\right\}$, and
\begin{equation*}
f_{Y|X,\alpha}(y|X;\beta,\alpha)=\begin{cases}
(1-H(X_{1}^{\prime}\beta+\alpha))(1-H(X_{2}^{\prime}\beta+\alpha)) & \text{if }Y=(0,0)\\
(1-H(X_{1}^{\prime}\beta+\alpha))H(X_{2}^{\prime}\beta+\alpha) & \text{if }Y=(0,1)\\
H(X_{1}^{\prime}\beta+\alpha)(1-H(X_{2}^{\prime}\beta+\alpha)) & \text{if }Y=(1,0)\\
H(X_{1}^{\prime}\beta+\alpha)H(X_{2}^{\prime}\beta+\alpha) & \text{if }Y=(1,1).
\end{cases}
\end{equation*}

Consider a setting with a single value of $X = x = (x_1, x_2)$, where $x_1$ is the period-1 value and $x_2$ is the period-2 value. In the model with only a time dummy, $x = (x_1,x_2) = (0,1).$

Then the $2^T \times K_a$ matrix
\begin{equation}
\widetilde C_\theta(x) =
\begin{bmatrix}
    (1-H(a_1+x_1\beta))(1-H(a_1+x_2\beta)) & \cdots & (1-H(a_M+x_1\beta))(1-H(a_M+x_2\beta)) \\
    (1-H(a_1+x_1\beta))H(a_1+x_2\beta)     & \cdots & (1-H(a_M+x_1\beta))H(a_M+x_2\beta) \\
    H(a_1+x_1\beta))(1-H(a_1+x_2\beta))    & \cdots & H(a_M+x_1\beta)(1-H(a_M+x_2\beta)) \\
    H(a_1+x_1\beta))H(a_1+x_2\beta)        & \cdots & H(a_M+x_1\beta)H(a_M+x_2\beta),
\end{bmatrix}
\label{eq:pushforward_binary_binary}
\end{equation}
maps any distribution of fixed effects to a distribution of $(Y_1,Y_2)$ for a given parameter value $\theta$.
Initially, we do not impose additional assumptions on $p_W$ beyond the adding-up constraints $A_\theta = \iota_M^\prime$ and $b_\theta = 1$.

The description of $C_\theta$, $A_\theta$, and $b_\theta$ is complete, and provides a full description of how to compute the identified set for $\beta$ for the parametric binary choice model with fixed effects, $T=2$, and a single-valued regressor sequence.

We now describe how to generalize this procedure to general $\mathcal X$ and $T$, and how to include the ATE.
To construct $\widetilde C_\theta(x)$ for general $T$, allocate one row for each $y \in \{0,1\}^T$, and fill it with the probabilities given in \eqref{eq:static_binary_model_probabilities}.
For non-singleton $\mathcal X$, 
the matrix $\widetilde C_\theta$ is blockdiagonal, with each block as in \eqref{eq:pushforward_binary_binary}:
\[
\widetilde C_\theta = \mathrm{diag}((\widetilde C_\theta(x))_{x \in \mathcal X}).
\]
Setup $p_W$ and $p_Z^*$ analogously, from blocks for each value of $x$.
In the absence of parameters beyond $\beta$, the constraints are unchanged, $A_\theta = \iota^\prime, \; b_\theta = 1$, except that $A$ has $K_x K_a$ columns. 
Finally, denote by $\overline \tau$ a hypothesized value of the ATE
$E[H(\alpha + \beta) - H(\alpha)$.
Add one constraint by adding a row $A_\tau^\prime$ to $A_\theta$, with $A_\tau = (A_\tau(x,a))$, and $A_\tau(x,a) = H(a + \beta) - H(a)$.
Add $\overline \tau$ as the corresponding element in $b_\theta$.


\newpage 

\section{Inference}
\label{sec:inference}

\subsection{Inferential results}

We present results for inference on the identified set $\Thetaw$ through hypothesis testing and confidence sets for $\Thetam$. Our main insight is that the asymptotic distribution of an empirical analog of the discrepancy function, denoted $T_n(\theta)$, can be estimated by imposing a penalty function on the space of features, combining information on the features themselves with the behavior of $\mathcal{M}_\theta$ local to $\mu^*_Z$. We provide conditions under which a test statistic based on $T_n(\theta)$ has a limiting distribution that can be estimated via the bootstrap. These conditions are applied uniformly over a class of possible distributions $\mu_Z^*$, allowing for the construction of uniformly valid confidence sets for $\theta$.

In Section \ref{inference_lit} we briefly place our inference procedure in the literature. In Section \ref{sec:inference_preconditions}, we introduce the necessary notation and preconditions for our inferential results. The main inferential results are then stated in Section \ref{sec:inference_mainresults}. Given the crucial role of features in both identification and inference, we further discuss their dual role in Section \ref{sec:inference_dualrole}.

\subsection{Related Literature}
\label{inference_lit}

Our discrepancy function naturally leads to a test statistic for valid inference. Our analysis frames the inference problem in terms of an infinite number of moment inequalities. Contributions to the literature on inference with moment inequalities include \citet{GUGGENBERGER2008}, \citet{andrews2010}, \citet{S2012}, \citet{bugni2015}, \citet{CNS2023}.
Our inferential procedure differs by aggregating over potentially infinite moment inequalities into one test statistic - the sample analog of the discrepancy function. Additionally, our test statistic consists of an outer maximization of an inner minimization.\footnote{\citet{loh2024} discuss a general implementation of this strategy to min-max type test statistics.} As the inner minimization is tractable, this simplifies the computational burden of our inference strategy and makes it broadly applicable to a range of models.\footnote{For more on advantages of this, see the discussion and references in \citet{marcoux2024}.} The asymptotic distributions of our test statistic are also straightforward to estimate and may be used to form uniformly valid confidence sets over sets of underlying probability distributions \`a la \citet{im2004}.

\subsubsection{Notation and Preconditions}
\label{sec:inference_preconditions}
We begin by introducing notation required for the inferential results.
For conciseness, we denote by $\mu^*$ the true distribution $\mu^*_Z$ and by $\EE{n}{\cdot}$ the sample expectation based on $n$ observations.
Given a class of distributions $\mathcal{P} \subset \mathcal{P}(\mathcal{Z})$, we say that a set of random variables $\{A_{n,\mu^*}: \mu^* \in \mathcal{P}\}$ is $o_p(1)$ uniformly in $\mu^*$ if, for all $c > 0,$ 
\begin{align} \label{E:op1}
    \limsup_{n \ra \infty } \sup_{\mu^* \in \mathcal{P}} \mu^*(\{|A_{n , \mu^*}| > c \}) = 0,
\end{align}
and $O_p(1)$ uniformly in $\mu^*$ if the right hand side of \eqref{E:op1} can be made arbitrarily small by taking $c$ to be sufficiently large. 

Let $\ell^\infty(\Xi)$ be the set of uniformly bounded maps from $\Xi$ to $\R$, equipped with the uniform norm. Let $\mathrm{BL}_1$ denote the set of all Lipschitz functions from $\ell^\infty(\Xi)$ to $\R$ uniformly bounded by $1$. Similarly, let $\mathrm{BL}_1(\R)$ denote the set of bounded Lipschitz maps from $\R$ to $\R$. A set $\{\mathbb{G}_{n,\mu^*}: \mu^* \in \mathcal{P}\}$ of empirical processes in $\ell^\infty(\Xi)$, each indexed by $\phi \in \Xi$, is uniformly Donsker, or weakly converges uniformly to a set of limit processes $\{\mathbb{G}_{\mu^*}: \mu^* \in \mathcal{P}\}$ in $\ell^\infty(\Xi)$, if \begin{align} \label{E:bl11}
    \limsup_{n \ra \infty} \sup_{\mu^* \in \mathcal{P}}\sup_{h \in \mathrm{BL}_1} | \EE{\mu^*}{h(\mathbb{G}_{n,\mu^*})} - \E{h(\mathbb{G}_{\mu^*})} | = 0.
\end{align}
Probabilistic statements concerning empirical processes are meant to hold in outer measure.\footnote{See \citet{VW1996}, \S1.2, c.f.\ \citet{S2012}, Remark A.1.} 

Bootstrap analogs $\mathbb{G}_{n,\mu^*}^*$ of $\mathbb{G}_{n, \mu^*}$ are uniformly consistent for $\mathbb{G}_{\mu^*}$ if 
\begin{align} \label{E:bootstrapconsistent}
    \sup_{h \in \mathrm{BL}_1} | \EE{n}{h(\mathbb{G}_{n,\mu^*}^*)} - \E{h(\mathbb{G}_{\mu^*})} | = o_p(1) \text{ uniformly in }\mu^*.
\end{align}
If instead \eqref{E:bootstrapconsistent} holds only for a fixed $\mu^* \in \mathcal{P}$, we write $\mathbb{G}_{n,\mu^*}^* \overset{\mu^*}{\rs} \mathbb{G}_{\mu^*}$.

\subsubsection{Asymptotic distribution and bootstrap consistency}
\label{sec:inference_mainresults}

For each $\mu^* \in \mathcal{P}$, the identified set is given by $
    \Thetaw(\mu^* ) = \{\theta \in \Theta: \mu^* \in \overline{\mathcal{M}}_\theta\}$,
and our test statistic is the sample analog of $T(\theta)$ defined as
\begin{align}
\label{E:firstform}
    T_n(\theta) = \sup_{\phi \in \Xi}  \inf_{\mu \in \mathcal{M}_\theta} (\EE{n}{\phi} - \EE{\mu}{\phi}).
\end{align}
It is also convenient to define the penalty function 
\begin{align}
\label{test_penalty}
    \eta_{\theta, \mu^*}( \phi) &\equiv \inf_{\mu \in \mathcal{M}_\theta} (\EE{\mu^*}{\phi} - \EE{\mu}{\phi}), 
\end{align}
which can be estimated by its empirical analog $\eta_{\theta, n} (\phi) \equiv \inf_{\mu \in \mathcal{M}_\theta} (\EE{n}{\phi} - \EE{\mu}{\phi})$. Note that $\eta_{\theta, \mu^*}$ is always nonpositive when $\theta \in \Thetaw(\mu^*)$. We denote the nonpositive part of $\eta_{\theta, n}$ by: 
\begin{align}
    \eta_{\theta, n }^-(\phi) = \min\{ \eta_{\theta, n}(\phi), 0\}.
\end{align}

\begin{asm} 
\label{A:partialident}
    For $\mu^* \in \mathcal{P}$, let $\mathbb{G}_{n, \mu^*}(\phi) = \sqrt{n}(\EE{n}{\phi} - \EE{\mu^*}{\phi})$. Then:
    \begin{enumerate}
        \item $\Xi$ is a convex set of uniformly bounded Borel functions containing $0$. There exists a topology $\mathcal{U}$ on $\Xi$ such that $\Xi$ is compact, and the map $\phi \mapsto \EE{\mu^*}{\phi} - \EE{\mu}{\phi}$ is continuous for all $\theta \in \Thetaw$ and $\mu \in \mathcal{M}_\theta$.
        \item $\{\mathbb{G}_{n, \mu^*}: \mu^* \in \mathcal{P}\}$ are uniformly Donsker with tight limits $\{\mathbb{G}_{\mu^*}: \mu^* \in \mathcal{P}\}$ in $\ell^\infty(\Xi)$.
        \item There exist seminorms $\{\rho_{\mu^*}: \mu^* \in \mathcal{P}\}$ such that $\{\mathbb{G}_{n, \mu^*}: \mu^* \in \mathcal{P}\}$ are asymptotically equicontinuous.
        \item Bootstrap analogs  $\{\mathbb{G}_{n, \mu^*}^*: \mu^* \in \mathcal{P}\}$ are uniformly consistent for $\{\mathbb{G}_{\mu^*}: \mu^* \in \mathcal{P}\}$.
    \end{enumerate}
\end{asm}

The primary aim of Assumption \ref{A:partialident} is to regularize the set of features $\Xi$. The first condition requires that $\Xi$ be compact in a topology that ensures the continuity of the expectations $\EE{\mu}{\phi}$. The second condition imposes that the empirical processes $\mathbb{G}_{n,\mu^*}$ indexed by $\Xi$ are uniformly Donsker. Assumption \ref{A:partialident}(3) further requires the asymptotic equicontinuity of these empirical processes with respect to the chosen topology on $\Xi$. Assumption \ref{A:partialident}(4) provides uniformly consistent bootstrap estimates for $\mathbb{G}_{\mu^*}$. Lemma A.2 of \citet{LS2010} establishes this uniform bootstrap consistency under Assumption \ref{A:partialident}(2) and a slightly stronger, uniform version of Assumption \ref{A:partialident}(3). In Section \ref{sec:inference_dualrole}, we show that a simple regularity condition on the sets $\mathcal{M}_\theta$ is sufficient to guarantee that all of our inference assumptions hold for a suitable choice of $\Xi$.

Assumption \ref{A:partialident} allows us to  characterize the asymptotic distribution of $T_n(\theta)$ in the following proposition.

\begin{proposition}
\label{P:inf}[Inference under partial identification]
    Let Assumption \ref{A:partialident} hold and $\lambda_n \le \sqrt{n}$ be $o(\sqrt{n})$. For $\mu^* \in \mathcal{P}$ and $\theta \in \Theta$, define $K_{\mu^*}(\theta) = \{\phi \in \Xi: \eta_{\theta, \mu^*}(\phi) = \sup_{\phi'\in \Xi } \eta_{\theta, \mu^*}(\phi' )\}$, where $\eta_{\theta, \mu^*}(\phi)$ is defined in \eqref{test_penalty}. Then, 
    \begin{align} \label{E:inf1}
        &\sqrt{n} T_n(\theta)  \le \sup_{\phi \in \Xi}(\mathbb{G}_{n, \mu^*}(\phi) + \lambda_n \eta_{\theta, \mu^*}(\phi))\text{, and} \nonumber\\
        &\sup_{h \in \mathrm{BL}_1(\R)} \Big| \EE{\mu^*}{h(\sup_{\phi \in \Xi} (\mathbb{G}_{n, \mu^*}(\phi) + \lambda_n \eta_{\theta, \mu^*}(\phi)))} - \EE{n}{ h(\sup_{\phi \in \Xi}(\mathbb{G}_{n,\mu^*}^* + \lambda_n \eta_{\theta,n}^-(\phi)))} \Big| = o_p(1)
    \end{align}
    uniformly in $\mu^* \in \mathcal{P}$ and $\theta \in \Thetaw(\mu^*)$.
    Moreover, for all $\mu^*$ and $\theta$, $K_{\mu^*}(\theta) \neq \emptyset$ and  
    \begin{align} 
        &\sqrt{n}( T_n(\theta) - \max_{\phi \in \Xi} \eta_{\theta, \mu^*}(\phi) ) \rs \sup_{\phi \in K_{\mu^*}(\theta)} \mathbb{G}_{n, \mu^*}(\phi)\text{, and } \nonumber\\
        &\sup_{\phi \in \Xi}(\mathbb{G}_{\mu^*}^* + \lambda_n \eta_{\theta,n}^-(\phi)) \overset{\mu^*}{\rs} \sup_{\phi \in K_{\mu^*}(\theta)} \mathbb{G}_{\mu^*}(\phi) \text{ whenever } \theta \in \Thetaw(\mu^*). \label{E:inf2}
    \end{align}
\end{proposition}

\begin{proof}
    The proof can be found in Section \ref{proofofPinf}, page \pageref{proofofPinf}.
\end{proof}

Proposition \ref{P:inf} establishes that $\sqrt{n} T_n(\theta)$ can be bounded by a random term, which is consistently approximated by the bootstrap uniformly in both $\mu^*$ and $\theta$. Equation \eqref{E:inf2} shows that this bootstrap estimate and $\sqrt{n} T_n(\theta)$ share the same limiting distribution, making the bound tight as $n \to \infty$ for any specific $\mu^*$ and $\theta \in \Thetaw$. This upper bound provides critical values for the asymptotic distribution of $\sqrt{n} T_n(\theta)$ uniformly over $\mu^* \in \mathcal{P}$ and $\theta \in \Thetaw(\mu^*)$.

An additional assumption, stating that $\Xi$ is sufficiently rich to detect deviations of $\mu^*$ from $\overline{\mathcal{M}}_\theta$, guarantees the consistency of our testing procedure. If this condition fails, the test can still detect deviations from the set $\Thetam$.\footnote{Specifically, \eqref{E:boot12} below holds with $\Thetam(\mu^*)$ replacing $\Thetaw(\mu^*)$.}

\begin{asm}
    \label{A:partialident2} For all $\mu^* \in \mathcal{P}$ and $\theta \in \Theta$, $\Xi$ is such that $\sup_{\phi \in \Xi}  \inf_{ \mu \in \mathcal{M}_\theta} (\EE{\mu^*}{\phi} - \EE{\mu}{\phi})> 0$ whenever $\mu^* \not\in \overline{\mathcal{M}}_\theta$ (i.e., when $\theta \not\in \Thetaw(\mu^*)$).
\end{asm}

\begin{corollary} 
\label{C:inf2}
    Let Assumption \ref{P:inf} hold, and let $\ve > 0$ be arbitrary. Let $\hat{c}_{1-\alpha}(\theta) = \inf\{x: \PP{n}{\sup_{\phi \in \Xi} (\mathbb{G}_{n,\mu^*}^* + \lambda_n \eta_{\theta,n}^-(\phi)) \le x} \ge  1- \alpha \}$ denote the $(1 - \alpha )^\text{th}$ quantile of $\sup_{\phi \in \Xi} (\mathbb{G}_{n,\mu^*}^* + \lambda_n \eta_{\theta,n}^-(\phi))$. Then, 
    \begin{align}
        \liminf_{n \ra \infty} \inf_{\substack{\mu^* \in \mathcal{P} \\ \theta \in \Thetaw(\mu^*)}} \mu^*(\{\sqrt{n} T_n(\theta) \le \hat{c}_{1-\alpha}(\theta) + \ve \} ) \ge 1 - \alpha. \label{E:crit1}
    \end{align}
If Assumption \ref{A:partialident2} also holds and $\theta \not\in \Thetaw(\mu^*)$, then 
\begin{align} \label{E:boot12}
    \limsup_{n \ra \infty} \mu^*(\{ \sqrt{n} T_n(\theta) \le \hat{c}_{1-\alpha}(\theta) + \ve\} ) = 0.
\end{align}
\end{corollary}
\begin{proof}
    The proof can be found in Section \ref{proofofCinf2}, page \pageref{proofofCinf2}.
\end{proof}

Corollary \ref{C:inf2} provides conditions for constructing uniformly valid confidence sets for $\Thetaw(\mu^*)$. Let
\begin{align*}
    \mathrm{CS}_{\alpha}(\mu^*) =   \{\theta \in \Theta: \sqrt{n} T_n(\theta) \le \hat{c}_{1-\alpha}(\theta) + \ve\} 
\end{align*}
where $\hat{c}_{1-\alpha}(\theta)$ is a bootstrap-based critical value. Under the assumptions of Corollary \ref{C:inf2}, $\mathrm{CS}_{\alpha}(\mu^*)$ achieves the coverage properties of the confidence intervals proposed in \citet{im2004} uniformly over $\mu^* \in \mathcal{P}$, ensuring it contains the true value of the parameter $\theta \in \Thetaw(\mu^*)$ with probabilities uniformly bounded below by $1-\alpha$. 

We calculate $T_n(\theta)$ and $\hat{c}_{1-\alpha}(\theta)$ for a range of $\theta$ in a variation of DGP2 of Section \ref{sec:numerical_panel} with $\mathcal{X} = \{0,1,2\}$ and $\theta = \beta_2$. 
Setting $\theta_0 = 0.7$, $n = 2000$, $\lambda_n = 10$, and $\alpha = 0.05$, Figure \ref{fig:inf_0} plots the discrepancy function $T(\theta)$ alongside the estimated probability that $\sqrt{n}T_n(\theta)$ exceeds $\hat{c}_{1-\alpha}(\theta) + 0.001$ across a range of candidate values $\theta$. Figure \ref{fig:inf_1} repeats this exercise for several $\theta_0$. In these simulations, tests for the null hypothesis $H_0: \theta \in \Theta_{\mathrm{I}}$ based on \eqref{E:crit1} have uniform size control over $\Theta_{\mathrm{I}}$ and ample statistical power. Because $\hat{c}_{1-\alpha}$ can also be calculated by solving a linear program in each bootstrap sample, implementation of the tests is also computationally efficient. 

The fact that the tests depicted in Figure \ref{fig:inf} appear conservative is consistent with a central finding of \cite{Chamberlain2010}. When $V$ is not constrained to be logistically distributed, the sets $\mathcal{M}_\theta$ of model probabilities associated with DGP2 have nonempty interior, and a typical point $\theta$ in $\Theta_\mathrm{I}$ has $\mu^*$ in the interior of $\mathcal{M}_\theta$. When this is the case, $K_{\mu^*}(\theta)$ is the singleton $\{0\}$, and the test statistic $\sqrt{n} T_n(\theta)$ described by Proposition \ref{P:inf} converges in probability to $0$. Size control then comes from the tolerance parameter $\ve$ in Corollary \ref{C:inf2}, which induces Type I error rates less than the nominal rate over $\Theta_{\mathrm{I}}$. 

\begin{figure}
    \centering
    \begin{subfigure}{0.5\textwidth}
        \centering
        \includegraphics[width=\textwidth]{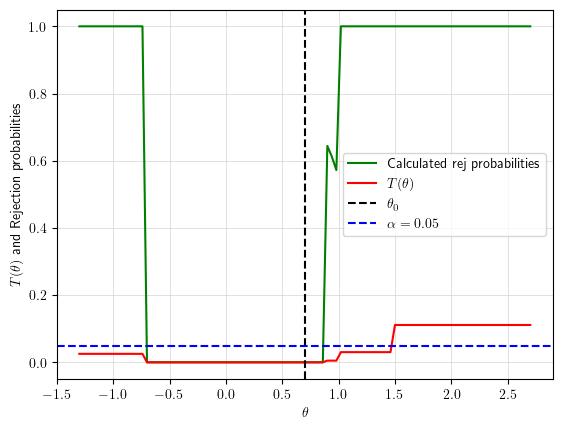} 
        \caption{Rejection probabilities vs $T(\theta)$ for $\theta_0 = 0.7$.}
        \label{fig:inf_0}
    \end{subfigure}\hfill
    \begin{subfigure}{0.5\textwidth}
        \centering
        \includegraphics[width=\textwidth]{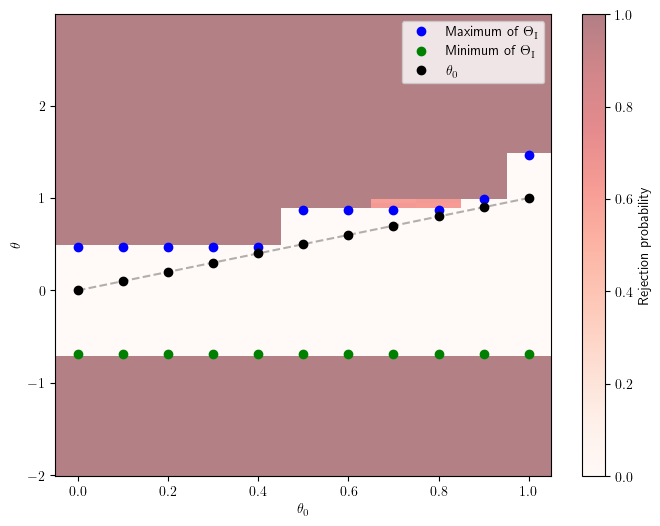} 
        \caption{Rejection probabilities with identified sets.}
        \label{fig:inf_1}
    \end{subfigure}
    \caption{Inference with bootstrap critical values in DGP2 of Section \ref{sec:numerical_panel}.}
    \label{fig:inf}
\end{figure}

\begin{remark}
    The assumptions employed in the appendix of \citet{loh2024} could be applied to guarantee the uniform convergence of the distribution of $\mathbb{G}_{n, \mu^*}(\phi) + \lambda_n \eta_{\theta, \mu^*}(\phi)$ to the distribution $\sup_{\phi \in K_{\mu^*}(\theta)} \mathbb{G}_{\mu^*}(\phi)$. Then, a straightforward uniform equicontinuity assumption on the distributions of $\sup_{\phi \in K_{\mu^*}(\theta)} \mathbb{G}_{\mu^*}(\phi)$, $\mu^* \in \mathcal{P}$ would be sufficient to take $\ve = 0$ in the statement of Corollary \ref{C:inf2} (see Assumption B.7 in \citet{Zhu2020}). 
\end{remark}

\subsubsection{The role of features in identification and inference}
\label{sec:inference_dualrole}

The functions $\phi \in \Xi$ are central to both identification and inference. On one hand, the identification problem requires a rich enough set $\Xi$ to capture the discrepancies between the observed measures and the model probabilities. On the other hand, the $\phi$ functions must be sufficiently regularized to construct test statistics for hypothesis testing. This section provides relaxations $\Xi$ of the set $\Phi_b$ introduced in Section \ref{sec:mainIDresult} which satisfy both criteria.  



Theorem \ref{P:main} establishes that the set of features $\Xi = \Phi_b(\mathcal{Z})$ is sufficient to detect deviations of $\mu^*$ from $\overline{M}_\theta$
Under this choice of $\Xi$, $T(\theta)$ represents the total variation distance between $\mu^*$ and $\mathcal{M}_\theta$. 
Under additional assumptions on the elements of $\mathcal M_\theta$, it is possible to consider a smaller set of functions that can be used to characterize the identified set. 


\begin{corollary} 
\label{C:reg1}
    Let $p \in (1,\infty]$ with $1/p + 1/q = 1$. Suppose Assumption \ref{A:one} holds, and let $f[\mu]$ denote the density of $\mu$ with respect to $\lambda_\theta$. If $f[\mu^*] \in L^p(\mathcal{Z}, \lambda_\theta)$ and $\sup_{\mu \in \mathcal{M}_\theta} \| f[\mu] \|_{L^p(\mathcal{Z}, \lambda_\theta)} < \infty$, let $\Xi \subseteq L^q$ be any set of functions such that the $L^p(\mathcal{Z}, \lambda_\theta)$-closure of the set of positive dilations $\bigcup_{c > 0} c \Xi$ includes all compactly supported Borel functions $\phi: \mathcal{Y} \to [0,1]$. Then $\theta$ is in the identified set if and only if, for all $\phi \in \Xi$, 
    \begin{align} 
    \label{E:Creg1}
        \EE{\mu^*}{\phi} \le \sup_{\mu \in \mathcal{M}_\theta} \EE{\mu}{\phi}.
    \end{align}
\end{corollary}
\begin{proof}[Proof of Corollary \ref{C:reg1}]
Suppose first that $\theta$ is in the identified set, so that \eqref{E:ineqthm} holds for all compactly supported and Borel $\phi: \mathcal{Z} \ra [0,1]$.  Write $C = \sup_{\mu \in \mathcal{M}_\theta} \norm{ f[\mu]}_{L^p}$. By H\"{o}lder's inequality, for all $c > 0$ one has 
\begin{align}\label{E:reg1}
    \sup_{\mu \in \mathcal{M}_\theta} |\EE{\mu}{c\phi} - \EE{\mu}{c_1\phi + c_2}| \le C \norm{c \phi - c_1\phi - c_2}_{L^q}
\end{align}
for $c_1, c_2 \in \R$, $\phi \in \Xi$. Assume that $C \ge \norm{f[\mu^*]}_{L^p}$ so that \eqref{E:reg1} pertains to $|\EE{\mu^*}{c\phi} - \EE{\mu^*}{c_1\phi + c_2}|$ as well. Now, pick $c, c_1 > 0$, $c_2 \in \R$, and $\phi: \mathcal{Z} \ra [0,1]$ such that $\norm{c\phi - c_1\phi - c_2}_{L^q} < \ve$ for any arbitrary $\ve > 0$ (for instance, consider truncating $\phi$ above and below). Then \eqref{E:ineqthm} implies that 
\[
\EE{\mu^*}{c \phi}- C \ve \le \EE{\mu^*}{c_1\phi + c_2} \le \sup_{\mu \in \mathcal{M}_\theta} \EE{\mu}{c_1\phi + c_2} \le \sup_{\mu \in \mathcal{M}_\theta} \EE{\mu}{c \phi} + C\ve.
\]
As $\ve$ was arbitrary, \eqref{E:Creg1} follows. In the same way, by approximating functions $\phi: \mathcal{Z} \ra [0,1]$ by functions $c\phi \in c \Xi, c > 0$ in the $L^q$-norm, the ``if" part of the Corollary may be established. 
\end{proof}

If $\mathcal{Z} = \R^{d_Z}$ and $\lambda_\theta$ is the Lebesgue measure, Corollary \ref{C:reg1} allows for $\Xi$ to be the set of smooth, compactly supported functions on $\R^{d_Z}$ (\cite{F1999}, \S8.2) bounded in magnitude by $1$ (or a subset whose positive dilations include all such functions). If $\mathcal{Z}$ is a bounded domain in $\R^{d_Z}$, $\Xi$ can be a set of functions with bounded entropy (\cite{VW1996}, \S2.7) corresponding to various Sobolev constraints.



Corollary \ref{C:reg1} may be employed to provide a host of examples which satisfy our inferential assumptions.
Assumption \ref{A:partialident} requires that \(\Xi\) be structured in a way that the empirical processes indexed by \(\phi\) converge uniformly to a limit. This condition is crucial for \(T_n(\theta)\) to have a well-defined limiting distribution. Assumption \ref{A:partialident2} further stipulates that \(\Xi\) must be rich enough to detect deviations of \(\mu^*\) from \(\overline{\mathcal{M}}_\theta\), ensuring the consistency of the testing procedure.

When \(\mathcal{Z}\) is a compact subset of Euclidean space, Example \ref{Ex:sobolev} illustrates a choice of \(\Xi\) that satisfies the assumptions required for both sharp identification and valid inference.

\begin{example} \label{Ex:sobolev}
    Suppose the measures \(\mu \in \mathcal{M}_\theta, \theta \in \Theta\) all share a compact support \(\mathcal{Z} \subseteq \R^{d_Z}\). Assume \(\mathcal{Z}\) is convex with a nonempty interior (or use its convex hull). By Theorems 2.7.1 and 2.8.3 in \citet{VW1996}, the empirical processes \(\mathbb{G}_{n,\mu^*}\) are uniformly Donsker if \(\Xi\) is defined as the class of smooth functions \(\phi\) over \(\mathcal{Z}\) that are partially differentiable up to order \(L > d_Z / 2\) and satisfy the following Lipschitz condition:
    \begin{align} \label{E:sobolev}
        \norm{\phi}_{L,\mathcal{Z}} \equiv \max_{
        \ell_\cdot < \underline{L}} \sup_{z \in \mathrm{int}(\mathcal{Z})} |D^\ell \phi(z)| + \max_{\ell_\cdot = \underline{L}} \sup_{z,z' \in \mathrm{int}(\mathcal{Z})} \frac{|D^\ell \phi(z) - D^\ell \phi(z')|}{\norm{z - z' }^{L - \underline{L}}} \le 1,
    \end{align}
    where \(\ell = (\ell_1, \ldots, \ell_{d_Z})\) is a multiindex, \(\ell_\cdot = \sum_{\iota =1 }^{d_Z} \ell_\iota\), \(\underline{L}\) is the greatest integer strictly smaller than \(L\), and \(D^\ell\) denotes the differential operator associated with \(\ell\). This norm \(\norm{\cdot}_L\) makes \(\Xi\) a convex set that is also totally bounded under the sup-norm.
\end{example}

Lemma \ref{L:reg1} now extends this example, showing that compactness of \(\mathcal{Z}\) is not necessary to define a suitable class \(\Xi\). 


\begin{lemma}
\label{L:reg1}
    Suppose every \(\mu \in \bigcup_{\theta \in \Theta} \mathcal{M}_\theta\) has a density \(f[\mu]\) with respect to some \(\sigma\)-finite Borel measure \(\lambda\) on a second-countable space \(\mathcal{Z}\). Further, assume \(\mathcal{M}_\theta\) is convex for each \(\theta \in \Theta\) and \(\sup_{\mu \in \mathcal{M}_\theta} \norm{f[\mu]}_{L^p(\mathcal{Z}, \lambda)} < \infty\) for some \(p \in (1, \infty]\). Then, there exists a set \(\Xi\) that satisfies Assumptions \ref{A:partialident} and \ref{A:partialident2}.

    If, in addition, \(\mathcal{Z}\) is a bounded subset of \(\R^{d_Z}\) and \(\lambda\) is bounded on bounded sets, \(\Xi\) can be taken as the \(L^\infty(\mathcal{Z}, \lambda)\)-closure of the restrictions of smooth, compactly supported functions \(\phi \in C_0^\infty(\R^{d_Z})\) satisfying \(\norm{ \phi}_{L,\mathcal{Z}} \le 1\) for any \(L > d_Z/2\), where \(\norm{\cdot}_{L,\mathcal{Z}}\) is as defined in \eqref{E:sobolev}.
\end{lemma}
\begin{proof}
    The proof can be found in Section \ref{proofofLreg1}, on page \pageref{proofofLreg1}.
\end{proof}

\subsection{Proofs for inferential results}

\begin{proof}[Proof of Lemma \ref{L:reg1}]
\label{proofofLreg1}
    Let $q < \infty$ be such that $1/p + 1/q = 1$. By Proposition 3.4.5 of \citet{cohn2013measure}, $L^q(\mathcal{Z}, \lambda)$ is separable, so there exists some sequence $(f_\ell)_{\ell \in \N}$ of Borel functions whose $L^q$-closure contains the compactly supported Borel maps $\phi: \mathcal{Z} \ra [0,1]$. Because the $L^q$-distance of $f_\ell$ to any point in the latter set only decreases if we replace it with the projection $\Pi_{[0,1]}f_\ell$, $\Pi_{[0,1]}$ denoting the metric projection of $\R$ onto $[0,1]$, we may assume without loss of generality that every $f_\ell$ has range $[0,1]$.

    Let $\Xi$ be the (sup-norm, i.e.\ $L^\infty(\mathcal{Z}, \lambda)$) closed convex hull of the points $\{f_\ell/2^\ell: \ell \in \N\}.$ Fix $\ve > 0$ and let $L$ be such that $\ve > 2^{-L}$ (note that the smallest choice of such an $L$ satisfies $L \le - \log \ve / \log 2 + 1$, which is bounded above by $-2 \log \ve$ for $\ve$ sufficiently small). For all $\ell \le L$, define $A_\ell$ to be a set of cardinality at most $(L+1) 2^{L + 1 - \ell}$ contained in $[0,1]$ that contains $0$ and divides the interval into subintervals of length at most $2^{\ell - L - 1} / (L+1)$. Then, because each $f_\ell$ is bounded in magnitude by $1$, the set 
    \begin{align*}
        \left\{ \sum_{\ell = 1}^{L+1} a_{\ell}\frac{f_\ell}{2^\ell}: a_{\ell} \in A_\ell \, \forall \ell  \right\}
    \end{align*}
    defines a $2^{-L}$-cover of $\Xi$ (in the supremum norm). The cardinality of this set is at most $\prod_{\ell = 1}^{L+1} |A_\ell| = (L+1)^{L+1} 2^{L(L+1)/2}$. Letting $N(\cdot, \Xi, \norm{\cdot}_\infty)$ denote the covering number of $\Xi$ under the $L^\infty(\mathcal{Z}, \lambda)$ norm (\citet{VW1996}, \S2.1), conclude that 
    \begin{align*}
        \int_0^\infty \sqrt{\log N(\ve, \Xi, \norm{\cdot}_\infty) } \, \mathrm{d} \ve &\le \sum_{L = 1}^\infty 2^{-L} \sqrt{ \log( (L+1)^{L+1} 2^{L(L+1)/2})  } \\
        & \le \sum_{L = 1}^\infty 2^{-L} \sqrt{(L+1)^2 + L(L+1)\log(2) / 2 } < \infty.
    \end{align*}
    As the functions in $\Xi$ are uniformly bounded by $1$, the bracketing numbers $N_{[]}(\ve, \Xi, \norm{\cdot}_\infty)$ are bounded by $N(\ve/2, \Xi, \norm{\cdot}_\infty)$ (\citet{VW1996}, \S2.7) and $\norm{\cdot}_\infty$ upper bounds $\norm{\cdot}_{L^2(\mathcal{Z}, P)}$ for any probability measure $P$, Theorem 2.8.4 of \citet{VW1996} implies that the class $\Xi$ is Donsker and pre-Gaussian uniformly in $\mu^* \in \mathcal{P}$, and Lemma A.2 of \citet{LS2010} implies that Assumption \ref{A:partialident}(4) holds. The uniform $O_p(1)$ requirement of Assumption \ref{A:partialident}(2) is fulfilled by Markov's inequality and the fact that $\Xi$ is pre-Gaussian uniformly in $\mu^*$ (\citet{VW1996}, \S2.8). Theorem 2.8.2 of \citet{VW1996} implies that the asymptotic equicontinuity condition of Assumption \ref{A:partialident}(3) is satisfied uniformly in $\mu^* \in \mathcal{P}$ (which is stronger than our requirement) with seminorms 
    \begin{align} \label{E:rhonu}
    \rho_{\mu^*}(\phi, \phi' ) \equiv \EE{\mu^*}{(\phi - \phi' - \EE{\mu^*}{\phi - \phi'} )^2}.
    \end{align}
    Moreover, $\Xi$ is clearly totally bounded and closed in $L^\infty(\mathcal{Z}, \lambda)$, so it satisfies the requirements of Assumption \ref{A:partialident}(1) with $\mathcal{U}$ the $L^\infty$-topology. This choice of $\mathcal{U}$ clearly makes the map $\phi \mapsto \EE{\mu^*}{\phi} - \EE{\mu}{\phi}$ continuous for all $\theta \in \Theta$, $\mu \in \mathcal{M}_\theta$. The seminorms $\rho_{\mu^*}$ defined in \eqref{E:rhonu} are also continuous with respet to the $L^\infty$-norm, so Assumption \ref{A:partialident}(3) is satisfied. Finally, by Corollary \ref{C:reg1}, a point $\theta \in \Theta$ is in $\Thetaw$ if and only if 
    \begin{align} \label{E:minimaxlemma}
        \inf_{\phi \in \Xi} \sup_{\mu \in \mathcal{M}_\theta} ( \EE{\mu^*}{\phi} - \EE{\mu}{\phi}) \le 0. 
    \end{align}
    Here, $\Xi$ may be regarded as a compact and convex subset of $L^\infty(\mathcal{Z}, \lambda)$, and $\mathcal{M}_\theta$ may be regarded as a convex subset of $L^1(\mathcal{Z} , \lambda )$ by identifying measures with their densities with respect to $\lambda$. An application of the minimax theorem (\citet{Sion1958}) implies that \eqref{E:minimaxlemma} is equivalent to $\sup_{\mu \in \mathcal{M}_\theta} \inf_{\phi \in \Xi} (\EE{\mu^*}{\phi} - \EE{\mu}{\phi}) \le 0$, so that Assumption \ref{A:partialident2} also holds. 

    Now we consider the second claim of the lemma. Under the stated conditions, the class of smooth and compactly supported functions $C_0^\infty(\R^{d_Z})$ is dense in $L^q(\mathcal{Z}, \lambda)$ (\citet{bogachev2007measure} Corollary 4.2.2). Let $\tilde{\Xi}$ be the set of restrictions $\phi|_{\mathcal{Z}}$, $\phi \in C_0^\infty(\R^{d_Z}): \norm{\phi}_{L, \mathcal{Z}} \le 1$. Let $\Xi$ be the $L^\infty(\mathcal{Z}, \lambda)$-closure of $\tilde{\Xi}$. The set of positive dilations of functions in $\Xi$ contains $C_0^\infty(\R^{d_Z})$ so that Corollary \ref{C:reg1} implies that Assumption \ref{A:partialident2} is fulfilled if we can likewise show that $\Xi$ is compact in the norm $L^\infty(\mathcal{Z}, \lambda)$. Assume without loss of generality, potentially by extending the set $\mathcal{Z}$, that $\mathcal{Z}$ is bounded, convex, and has nonempty interior. Theorem 2.7.1 of \citet{VW1996} implies that
    \begin{align*}
        \log N(\ve, \Xi, \norm{\cdot}_\infty) \le \log N(\ve/2, \tilde{\Xi}, \norm{\cdot}_\infty) \le K \ve^{-d_Z / L},
    \end{align*}
    where $K$ depends only on $L$, $d_Z$, and $\mathcal{Z}$. Hence, $\Xi$ is totally bounded and closed in the sup-norm topology $\mathcal{U}$. Finally, because $d_Z/L < 2$, Theorems 2.8.2 and 2.14.10 of \citet{VW1996} still apply, and imply that $\Xi$ meets the requirements of Assumptions \ref{A:partialident} and \ref{A:partialident2}. 
\end{proof}

\begin{proof}[Proof of Proposition \ref{P:inf}]
\label{proofofPinf}
    The first inequality of \eqref{E:inf1} follows because $\eta_{\theta, \mu^*}(\phi) \le 0$ for all $\mu^*$, $\theta \in \Thetaw(\mu^*)$, and $\phi \in \Xi$, so that
    \begin{align} \nonumber
        \sqrt{n} T_n(\theta) & = \sqrt{n} \sup_{\phi \in \Xi} \inf_{\mu \in \mathcal{M}_\theta} ((\EE{n}{\phi} - \EE{\mu^*}{\phi}) + (\EE{\mu^*}{\phi} - \EE{\mu}{\phi})) \\
        & = \sup_{\phi \in \Xi} (\mathbb{G}_{n, \mu^*} (\phi) + \sqrt{n} \eta_{\theta, \mu^*}(\phi)) \le \sup_{\phi \in \Xi}  (\mathbb{G}_{n, \mu^*} (\phi) + \lambda_n\eta_{\theta, \mu^*}(\phi)). \label{E:tnsqn}
    \end{align}
    Using Assumption \ref{A:partialident}(2), we may write 
    \begin{align*}
        \sup_{\phi \in \Xi} |\eta_{\theta, \mu^*}(\phi) - \eta_{\theta, n}^-(\phi)| & \le \sup_{\phi \in \Xi} |\eta_{\theta, \mu^*}(\phi) - \eta_{\theta, n}(\phi)| 
        \le \sup_{\phi \in \Xi} | \EE{n}{\phi} - \EE{\mu^*}{\phi}|\\
        &= n^{-1/2} \norm{\mathbb{G}_{n, \mu^*}}_{\ell^\infty} = n^{-1/2} O_p(1)
    \end{align*}
    uniformly in $\mu^*$ and $\theta \in \Thetaw(\mu^*)$. Hence, 
    \begin{align} 
    \label{E:ps100}
        | \sup_{\phi \in \Xi}  (\mathbb{G}_{n, \mu^*}^* (\phi) + \lambda_n\eta_{\theta,n}^-(\phi)) - \sup_{\phi \in \Xi}  (\mathbb{G}_{n, \mu^*}^* (\phi) + \lambda_n\eta_{\theta, \mu^*}(\phi)) | \\ 
        \le  \lambda_n\sup_{\phi \in \Xi} |\eta_{\theta, \mu^*}(\phi) - \eta_{\theta, n}^-(\phi)| \\
        = o_p(1)
    \end{align}
    uniformly in $\mu^*$ and $\theta \in \Thetaw(\mu^*)$. Moreover, for any $h \in \mathrm{BL}_1(\R)$, the map $\mathbb{G}_{n,\mu^*} \mapsto h((\sup_{\phi \in \Xi} \mathbb{G}_{n,\mu^*}(\phi) + \lambda_n \eta_{\theta, \mu^*}(\phi)))$ is in $\mathrm{BL}_1$, so that Assumption \ref{A:partialident}(4) implies that 
    \begin{align}
        \sup_{h \in \mathrm{BL}_1(\R)} | \EE{\mu^*}{h(\sup_{\phi \in \Xi} (\mathbb{G}_{n, \mu^*}(\phi) + \lambda_n \eta_{\theta, \mu^*}(\phi)))} - \EE{n}{h(\sup_{\phi \in \Xi}(\mathbb{G}_{n,\mu^*}^* + \lambda_n \eta_{\theta, \mu^*}(\phi)))} | = o_p(1) \label{E:ps101}.
    \end{align}
    uniformly in $\mu^*$. The triangle inequality, \eqref{E:ps100}, and \eqref{E:ps101} imply the second line of \eqref{E:inf1}.

    By Assumption \ref{A:partialident}(1), $\eta_{\theta, \mu^*}$ is the infimum of a collection of $\mathcal{U}$-continuous functions, and is therefore $\mathcal{U}$-upper semicontinuous. Therefore, it achieves its maximum on $\Xi$, and $K_{\mu^*}(\theta)$ is always nonempty.
    
    Now, we prove the second line of \eqref{E:inf2}. Fix $\mu^*$ and $\theta \in \Thetaw(\mu^*)$. The preceding arguments imply that $\sup_{\phi \in \Xi} ( \mathbb{G}_{n,\mu^*}^*(\phi) + \lambda_n \eta_{\theta, n}^- (\phi)) =\sup_{\phi \in \Xi} ( \mathbb{G}_{n,\mu^*}^*(\phi) + \lambda_n \eta_{\theta, \mu^*} (\phi)) + o_p(1)$, uniformly in $\mu^*$. Let $(b_n)$ be a diverging sequence that is $o(\lambda_n)$, and let $\PP{n}{\cdot}$ denote the sample probability associated with a (random) sample of size $n$ drawn from $\mu^*$. By Assumption \ref{A:partialident}(4), $\PP{n}{\sup_{\phi \in \Xi} \mathbb{G}_{n,\mu^*}^*(\phi) > b_n}  = o_p(1)$. By Assumption \ref{A:partialident}(1), $T_n(\theta)$ is nonnegative, so by \eqref{E:inf1}, $\EE{\mu^*}{\PP{n}{\sup_{\phi \in \Xi} (\mathbb{G}_{n,\mu^*}^*(\phi) + \lambda_n \eta_{\theta, \mu^*}(\phi)) < c}} = o(1)$ for any $c < 0$. Accordingly, we may choose a sequence $(c_n) \uparrow 0$ such that $$\PP{n}{\sup_{\phi \in \Xi} (\mathbb{G}_{n,\mu^*}^*(\phi) + \lambda_n \eta_{\theta, \mu^*}(\phi)) < c_n} = o_p(1).$$ By the union bound, we may conclude that 
    \begin{align} \label{E:ps103}
        \PP{n}{\sup_{\phi \in \Xi} (\mathbb{G}_{n,\mu^*}^*(\phi) + \lambda_n \eta_{\theta, \mu^*}(\phi))  = \sup_{\phi : \eta_{\theta}(\phi) \ge (c_n - b_n)/\lambda_n} (\mathbb{G}^*_{n,\mu^*}(\phi) + \lambda_n \eta_{\theta, \mu^*}(\phi))  } \\
        = 1 - o_p(1).
    \end{align}
    By Assumption \ref{A:partialident}(3), for any $\delta > 0$ and $\phi \in \Xi$, there is some neighborhood $U_\phi$ of $\phi$ satisfying that $\rho_{\mu^*}(\phi, \phi') < \delta$ whenever $\phi' \in U_\phi$. The union $U_K \equiv \bigcup_{\phi \in K_{\mu^*}(\theta)} U_\phi$ is an open neighborhood of $K_{\mu^*}(\theta)$ satisfying that 
    \begin{align*}
        \sup_{\phi' \in U_K}\inf_{\phi \in K_{\mu^*}(\theta)} \rho_{\mu^*}(\phi, \phi') \le \delta.
    \end{align*}
    By upper semicontinuity and compactness, $\eta_{\theta, \mu^*}$ attains its maximum on $\Xi \setminus U_K$, and by consequence $\sup_{\phi \in \Xi \setminus U_K} \eta_{\theta, \mu^*}(\phi) < 0$. Thus, for all $n$ large enough such that $(c_n - b_n) / \lambda_n > \sup_{\phi \in \Xi \setminus U_K} \eta_{\theta, \mu^*}(\phi)$, the set $\{\phi : \eta_{\theta, \mu^*}(\phi) \ge (c_n - b_n)/\lambda_n\}$ is contained in $U_K$, and  
    \begin{align}
        \sup_{\phi : \eta_{\theta}(\phi) \ge (c_n - b_n)/\lambda_n} (\mathbb{G}^*_{n,\mu^*}(\phi) + \lambda_n \eta_{\theta, \mu^*}(\phi)) \le \sup_{\substack{\phi: \exists \phi' \in K_{\mu^*}(\theta) \text{ s.t.} \\ \rho_{\mu^*}(\mu, \mu') < \delta }} \mathbb{G}^*_{n,\mu^*}(\phi). \label{E:ps102}
    \end{align}
    Let $\ve > 0$ be arbitrary. By Addendum 1.5.8 of \citet{VW1996}, Assumption \ref{A:partialident}(2) and (3) imply that
    \begin{align*}
        \lim_{\delta \downarrow 0} \mu^*(\{\sup_{\rho_{\mu^*} (\phi , \phi') < \delta} | \mathbb{G}_{\mu^*}(\phi) - \mathbb{G}_{\mu^*}(\phi') | > \ve/2 \}) = 0.
    \end{align*}
    Because the map $\mathbb{G}_{\mu^*} \mapsto \sup_{\rho_{\mu^*} (\phi , \phi') < \delta} | \mathbb{G}_{\mu^*}(\phi) - \mathbb{G}_{\mu^*}(\phi') |$ is $\ell^\infty$-continuous, Assumption \ref{A:partialident}(4) implies that there is some $\delta$ sufficiently small so that 
    \begin{align*}
       \PP{n}{\sup_{\rho_{\mu^*} (\phi , \phi') < \delta} | \mathbb{G}_{n,\mu^*}^*(\phi) - \mathbb{G}_{n,\mu^*}^*(\phi') | > \ve } < \ve + o_p(1).
    \end{align*}
    For this choice of $\delta$, 
    \begin{align*}
        \PP{n}{\sup_{\substack{\phi: \exists \phi' \in K_{\mu^*}(\theta) \text{ s.t.} \\ \rho_{\mu^*}(\mu, \mu') < \delta }} \mathbb{G}^*_{n,\mu^*}(\phi) > \sup_{\phi \in K_{\mu^*}(\theta)} \mathbb{G}^*_{n,\mu^*}(\phi) + \ve} < \ve + o_p(1).  
    \end{align*}
    In conjunction with \eqref{E:ps102} and \eqref{E:ps103}, this implies that 
    \begin{align*}
        \PP{n}{\sup_{\phi \in \Xi} (\mathbb{G}_{n,\mu^*}^*(\phi) + \lambda_n \eta_{\theta, \mu^*}(\phi)) > \sup_{\phi \in K_{\mu^*}(\theta)} \mathbb{G}_{n,\mu^*}^*(\phi) + \ve }< \ve + o_p(1).
    \end{align*}
    As $\ve$ was arbitrary, the right hand side of the previous display may be amended to be simply $o_p(1)$. By definition of $K_{\mu^*}(\theta)$, one also has $\sup_{\phi \in \Xi} (\mathbb{G}_{n,\mu^*}^*(\phi) + \lambda_n \eta_{\theta, \mu^*}(\phi)) \ge \sup_{\phi \in K_{\mu^*}(\theta)} \mathbb{G}_{n,\mu^*}^*(\phi)$, so that 
    \begin{align*}
       & \PP{n}{|\sup_{\phi \in \Xi} (\mathbb{G}_{n,\mu^*}^*(\phi) + \lambda_n \eta_{\theta, \mu^*}(\phi)) - \sup_{\phi \in K_{\mu^*}(\theta)} \mathbb{G}_{n,\mu^*}^*(\phi) | > \ve } = o_p(1)\text{, whence} \\
        &\sup_{h \in \mathrm{BL}_1(\R)}\left|\EE{n}{h ( \sup_{\phi \in \Xi} (\mathbb{G}_{n,\mu^*}^*(\phi) + \lambda_n \eta_{\theta, \mu^*}(\phi))) } - \EE{n}{ h(\sup_{\phi \in K_{\mu^*}(\theta)}\mathbb{G}_{n,\mu^*}^*(\phi))  }  \right| \le \ve + o_p(1).
    \end{align*}
    Again, we may use the fact that $\ve$ is arbitrary and $\eta_{\theta, \mu^*}$ is well approximated by $\eta_{\theta,n}^-$ to conclude that 
    \begin{align*}
        \sup_{h \in \mathrm{BL}_1(\R)}\left|\EE{n}{h ( \sup_{\phi \in \Xi} (\mathbb{G}_{n,\mu^*}^*(\phi) + \lambda_n \eta_{\theta, n}^-(\phi))) } - \EE{n}{ h(\sup_{\phi \in K_{\mu^*}(\theta)} \mathbb{G}_{n,\mu^*}^*(\phi)) } \right| = o_p(1)
    \end{align*}
    as desired. The first line of \eqref{E:inf2} is proved similarly by decomposing $\sqrt{n}T_n(\theta)$ as in the first line of \eqref{E:tnsqn} and applying the convergence and asymptotic equicontinuity properties of $\mathbb{G}_{n,\mu^*}$ along with upper semicontinuity of $\eta_{\theta, \mu^*}$. 
\end{proof}

\begin{proof}[Proof of Corollary \ref{C:inf2}]
\label{proofofCinf2}
    For all $n$ and $\mu^*$, let $h_{n,\mu^*}$ be a nonincreasing Lipschitz function such that $h_{n,\mu^*}(x) = 1$ for all $x \le \hat{c}_{1- \alpha }(\theta)$ and $h_{n,\mu^*}(x) = 0$ for all $x \ge \hat{c}_{1- \alpha }(\theta) + \ve$. Then, \eqref{E:inf2} implies that 
    \begin{align*}
        &\inf_{\substack{\mu^* \in \mathcal{P} \\ \theta \in \Thetaw(\mu^*)}}\mu^*({\sqrt{n} T_n(\theta) \le \hat{c}_{1 - \alpha }(\theta) + \ve } ) \ge \inf_{\substack{\mu^* \in \mathcal{P} \\ \theta \in \Thetaw(\mu^*)}} \EE{\mu^*}{h_{n,\mu^*}(\sqrt{n} T_n(\theta)) } \\
        & \quad \ge \inf_{\substack{\mu^* \in \mathcal{P} \\ \theta \in \Thetaw(\mu^*)}} \EE{\mu^*}{ h_{n,\mu^*}(\sup_{\phi \in \Xi}(\mathbb{G}_{n, \mu^*}(\phi) + \lambda_n \eta_{\theta, \mu^*}(\phi)) )} \\
        &\quad \ge \inf_{\substack{\mu^* \in \mathcal{P} \\ \theta \in \Thetaw(\mu^*)}} \EE{\mu^*}{ \EE{n}{h_{n,\mu^*}(\sup_{\phi \in \Xi}(\mathbb{G}_{n, \mu^*}^*(\phi) + \lambda_n \eta_{\theta,n}^-(\phi)) )}} \\
        &\qquad - \sup_{\substack{\mu^* \in \mathcal{P} \\ \theta \in \Thetaw(\mu^*)}} \mathrm{E}_{\mu^*} \Big[ \Big| \EE{n}{h_{n,\mu^*}(\sup_{\phi \in \Xi}(\mathbb{G}_{n, \mu^*}^*(\phi) + \lambda_n \eta_{\theta,n}^-(\phi)) ) } \\
         &\qquad \qquad \qquad- \EE{\mu^*}{h_{n,\mu^*}(\sup_{\phi \in \Xi} (\mathbb{G}_{n, \mu^*}(\phi) + \lambda_n \eta_{\theta, \mu^*}(\phi))) }\Big| \Big].
    \end{align*}
    By boundedness of $h_{n,\mu^*}$ and the second line of \eqref{E:inf1}, the final term of the previous display is $o(1)$. Now $\EE{n}{h_{n,\mu^*}(\sup_{\phi \in \Xi}(\mathbb{G}_{n, \mu^*}^*(\phi) + \lambda_n \eta_{\theta,n}^-(\phi)) )} \ge 1 - \alpha$ by construction, so \eqref{E:crit1} follows.

    Now suppose that $\theta \not\in \Thetaw(\mu^*)$ and Assumption \ref{A:partialident2} holds. For $\ve > 0$ smaller than $\alpha$, let $c_{1-\alpha + \ve}$ denote the $(1 - \alpha + \ve)^\text{th}$ quantile of $\sup_{\phi \in \Xi} \mathbb{G}_{\mu^*}$. An argument replicating the one above implies that 
    \begin{align*}
        \PP{n}{\sup_{\phi \in \Xi} (\mathbb{G}_{n,\mu^*}^* + \lambda_n \eta_{\theta,n}^-(\phi)) \le c_{1 - \alpha + \ve} + \ve} &\ge \PP{n}{\sup_{\phi \in \Xi} \mathbb{G}_{n,\mu^*}^* \le c_{1- \alpha + \ve} + \ve}  \\
        & \ge 1 - \alpha + \ve - o_p(1). 
    \end{align*}
    With probability approaching $1$, the quantity above is bounded below by $1 - \alpha$, and so one has $\hat{c}_{1-\alpha}(\theta) \le c_{1 - \alpha + \ve} (\theta)+ o_p(1)$. On the other hand,  
    \begin{align*}
        \sqrt{n} T_n(\theta) & \ge \sqrt{n}\sup_{\phi \in \Xi} \inf_{\mu \in \mathcal{M}_\theta} (\EE{\mu^*}{\phi} - \EE{\mu}{\phi} ) - \sup_{\phi \in \Xi} \mathbb{G}_{n,\mu^*}(\phi) \\
        & = \sqrt{n}\sup_{\phi \in \Xi} \inf_{\mu \in \mathcal{M}_\theta} (\EE{\mu^*}{\phi} - \EE{\mu}{\phi} ) - O_p(1).
    \end{align*}
    Hence, 
    \begin{align*}
        &\mu^*\big(\{\sqrt{n} T_n(\theta) \le \hat{c}_{1-\alpha}(\theta) + \varepsilon\}\big) \\
        &\le \mu^*\big(\{\sqrt{n} \sup_{\phi \in \Xi} \inf_{\mu \in \mathcal{M}_\theta} (\mathbb{E}_{\mu^*}[\phi] - \mathbb{E}_{\mu}[\phi]) \le c_{1 - \alpha + \varepsilon} + \varepsilon + O_p(1) \}\big) \to 0.
    \end{align*}
\end{proof}

\section{Constraining the search over a smaller space of features}
\label{A:cont_compact_phi}

We show in Lemma \ref{lem:dim_reduc_cont} that, under additional regularity conditions, the dimensionality of the search can be further reduced by considering all continuous compactly supported functions $\phi:\mathcal{Z}\to [0,1]$.

\begin{asm}\label{A:extremepts} 
(i) $\mathcal{Z}$ is locally compact; (ii) $\lambda_\theta \in \mathfrak{B}(\mathcal{Z})$ is finite on compact sets (iii) for some $p \in (1,\infty]$ the densities $f[\mu]$ of measure $\mu \in \mathcal{M}_\theta$ with respect to $\lambda_\theta$ satisfy $\sup_{\mu \in \mathcal{M}_\theta}\norm{f[\mu]}_{L^p(\lambda_\theta)} < \infty$; (iv) $\psi_\theta: \mathcal{W} \ra \mathcal{Z}$ is continuous 
\end{asm}

We provide a discussion of this assumption following the result.

Let $C_c(\mathcal{Z})$ denote the space of continuous functions on $\mathcal{Z}$ with compact support, and $\overline{\Gamma}_\theta^w$ denote the weak closure of $\Gamma_\theta$ in $\mathcal{P}(\mathcal{W})$.

\begin{lemma}
\label{lem:dim_reduc_cont}
    Let the assumptions of Theorem \ref{thm:convexity} hold and suppose that Assumption \ref{A:extremepts} also holds. Then, $\theta \in \Thetaw$ if and only if $\EE{\mu^*_Z}{\phi} \le \sup_{\gamma \in (\overline{\Gamma}_\theta^w)} \EE{\gamma}{\phi \circ \psi_\theta}$ for all $\phi \in C_c(\mathcal{Z})$, and by consequence, $\Gamma_\theta$ can be replaced by $\overline{\Gamma}_\theta^w$ in \eqref{IDset_W} and $\Phi_b$ can be replaced by $C_c(\mathcal{Z})$ in \eqref{IDset_W} and \eqref{IDset_restricted}. 
\end{lemma}

\begin{proof} 
\label{prooflemdim_reduc_cont}
Suppose that Assumption \ref{A:extremepts} holds. 
Let $q$ be such that $1/p + 1/q = 1$. Because $q < \infty$ and $\lambda_\theta$ is finite on compact sets, the space $C_c(\mathcal{Z})$ of compactly supported functions over $\mathcal{Z}$ is dense in $L^q(\lambda_\theta)$ (see e.g.\ \citet{DriverApproximationConvolution}, Theorem 22.8). By Corollary \ref{C:reg1}, $\theta$ is in $\Thetaw$ if and only if \eqref{E:ps12} holds for all $\phi \in C_c(\mathcal{Z})$. Because the composition $\phi\circ \psi_\theta$ is bounded and continuous by Assumption \ref{A:extremepts}(iv), the right hand side of \eqref{E:ps12} is equal to $\sup_{\gamma \in \overline{\Gamma}_\theta^w} \EE{\gamma}{\phi \circ \psi_\theta}$. The arguments establishing \eqref{IDset_W} and \eqref{IDset_restricted} with these substitutions are exactly the same in this setting. 
\end{proof} 

\begin{remark}
    Assumption \ref{A:extremepts}(iv) might seem restrictive (e.g., it fails in the binary choice model). However, by adjusting the topology on the domain \(\mathcal{W}\), the Borel function \(\psi_\theta\) can be made continuous under this new topology, ensuring that all our previous results still apply. The intuition behind this adjustment is outlined as follows.
    
    Let $(\mathcal{W}, \mathcal{T})$ denote a topological space, with $\mathcal{W}$ a Polish space and $\mathcal{T}$ the system of open sets of $\mathcal{W}$. Let $(B_n)_{n \in \mathbb{N}}$ be a sequence of Borel sets in $\mathcal{W}$. There exists a Polish topology $\mathcal{T}' \supset \mathcal{T}$ on $\mathcal{W}$ such that the Borel $\sigma$-algebra generated by $\mathcal{T}'$ is the same as the one generated by $\mathcal{T}$, and every $B_n$ is both closed and open in $\mathcal{T}'$ (\citet{kechris1995classical}, Exercise 13.5). Consider the topology $\mathcal{T}^\prime$ that is is constructed by including all sets $\psi_\theta^{-1}(A)$, where $A \in \mathcal{Z}$ is an open set. Formally, $\mathcal{T}^\prime$ is the coarsest topology on $\mathcal{W}$ that makes $\psi_\theta$ continuous and includes the original topology $\mathcal{T}$.

    Because $\mathcal{Z}$ is Polish, there is a countable basis for its topology consisting of open sets $U_n \subseteq \mathcal{Z}$. Set $B_n = \psi_\theta^{-1}(U_n)$ for every $n$, and extend the topology on $\mathcal{W}$ to $\mathcal{T}'$ as above, so that every $B_n$ is open, but the Borel $\sigma$-algebra on $\mathcal{W}$ is not changed. Note that every open set $U \subseteq \mathcal{Z}$ can be written as a union $\bigcup_{n \in N_U} U_n$ of the sets $U_n$ over an index set $N_U \subseteq \mathbb{N}$, and one has $\psi_\theta^{-1}(U) = \bigcup_{n \in N_U} B_n$. As the latter set is open in the extended topology on $\mathcal{W}$, $\psi_\theta: \mathcal{W} \to \mathcal{Z}$ is continuous in the extended topology. Therefore, \textit{there always exists a Polish topology on $\mathcal{W}$ which preserves its Borel structure and makes $\psi_\theta$ continuous}.
\end{remark}

\end{document}